\documentclass[sigplan,10pt]{acmart}
\AtBeginDocument{%
 }
\setcopyright{acmlicensed}
\copyrightyear{2018}
\acmYear{2018}
\acmDOI{XXXXXXX.XXXXXXX}

\acmConference[Conference acronym 'XX]{Make sure to enter the correct
 conference title from your rights confirmation emai}{June 03--05,
 2018}{Woodstock, NY}
\acmISBN{978-1-4503-XXXX-X/18/06}

\usepackage[T1]{fontenc}
\usepackage{isabelle,isabellesym}
\usepackage{eurosym}
\usepackage[only,bigsqcap,bigparallel,fatsemi,interleave,sslash]{stmaryrd}
\usepackage{textcomp}
\usepackage{algpseudocode}
\usepackage{algorithm}
\usepackage{apxproof}

\isabellestyle{it}

\allowdisplaybreaks[3]

\newcommand{\prob}{\mathrm{Prob}}
\newcommand{\bind}{\mathbin{{>}\!\!\!{>}\!\!{=}}}
\newcommand{\return}{\mathop{\mathrm{return}}}

\newcommand{\inverse}{^{-1}}
\newcommand{\Nat}{\mathbb{N}}

\newcommand{\Real}{\mathbb{R}}

\DeclareMathOperator*{\argmax}{argmax}

\begin{document}

\title{Formalization of Differential Privacy in Isabelle/HOL\\(Draft version)}

%%Authors

\author{Tetsuya Sato}
%\authornote{supported by JSPS KAKENHI Grant Number 20K19775}
\email{tsato@c.titech.ac.jp}
\orcid{0000-0001-9895-9209}
\affiliation{
\institution{Tokyo Institute of Technology, Japan}
\city{Meguro}
\state{Tokyo}
\country{JAPAN}
}

\author{Yasuhiko Minamide}
%\authornote{supported by JSPS KAKENHI Grant Number 19K11899 and 20H04162}
\email{minamide@is.titech.ac.jp}
\affiliation{
\institution{Tokyo Institute of Technology, Japan}
\city{Meguro}
\state{Tokyo}
\country{JAPAN}
}
\renewcommand{\shortauthors}{T. Sato and Y. Minamide}

\begin{abstract}
Differential privacy is a statistical definition of privacy that has attracted the interest of both academia and industry.
Its formulations are easy to understand, but the differential privacy of databases is complicated to determine.
One of the reasons for this is that small changes in database programs can break their differential privacy.
Therefore, formal verification of differential privacy has been studied for over a decade.

In this paper, we propose an Isabelle/HOL library for formalizing differential privacy in a general setting.
To our knowledge, it is the first formalization of differential privacy that supports continuous probability distributions.
First, we formalize the standard definition of differential privacy and its basic properties.
Second, we formalize the Laplace mechanism and its differential privacy.
Finally, we formalize the differential privacy of the report noisy max mechanism.
\end{abstract}
\begin{CCSXML}
<ccs2012>
<concept>
<concept_id>10002978.10002986.10002990</concept_id>
<concept_desc>Security and privacy~Logic and verification</concept_desc>
<concept_significance>500</concept_significance>
</concept>
<concept>
<concept_id>10003752.10010124.10010138.10010142</concept_id>
<concept_desc>Theory of computation~Program verification</concept_desc>
<concept_significance>500</concept_significance>
</concept>
<concept>
<concept_id>10003752.10010124.10010131.10010133</concept_id>
<concept_desc>Theory of computation~Denotational semantics</concept_desc>
<concept_significance>300</concept_significance>
</concept>
</ccs2012>
\end{CCSXML}

\ccsdesc[500]{Security and privacy~Logic and verification}
\ccsdesc[500]{Theory of computation~Program verification}
\ccsdesc[300]{Theory of computation~Denotational semantics}

\keywords{Differential Privacy, Isabelle/HOL, proof assistant, program verification}
% \includegraphics[width=\textwidth]{sampleteaser}
% \caption{Seattle Mariners at Spring Training, 2010.}
% \Description{Enjoying the baseball game from the third-base
% seats. Ichiro Suzuki preparing to bat.}
% \label{fig:teaser}


\maketitle

\section{Introduction} %1page

Differential privacy~\cite{Dwork2006,DworkMcSherryNissimSmith2006} is a statistical definition of privacy
that has attracted the interest of both academia and industry.
The definition of differential privacy is quite simple, and methods for differential privacy are easy to understand.
It also has several notable strengths.
First, it guarantees difficulty in detecting changes in a database by any statistical method, including unknown ones.
Second, it has composition theorems that guarantee the differential privacy of a complex database from its components.

Thanks to these strengths, differential privacy is becoming a de facto standard for the privacy of databases with large amounts of individuals' private data.
For example, the iOS's QuickType suggestions and Google Maps' display of congestion levels use (local) differential privacy
to anonymize individuals' text inputs and location data~\cite{AppleDP,GoogleDP}, respectively.

%\tetsuya{Add here: trade-off between privacy and accuracy}
However, it is complicated and cumbersome to determine the differential privacy of databases.
A few (even one-line) changes of programs in databases can break their differential privacy~\cite[Section 3]{10.14778/3055330.3055331}.

For these reasons, formal verification of differential privacy has been studied actively for over a decade.
The main idea is to characterize differential privacy as a property of probabilistic programs.
The first approach is reformulating differential privacy as a relational property of (two runs of) probabilistic programs.
Based on the relational Hoare logic (RHL)~\cite{Benton2004export:67345},
Barthe et al. proposed a relational program logic for reasoning about differential privacy
named the approximate relational Hoare logic (apRHL) based on the discrete denotational model of probabilistic programs~\cite{Barthe:2012:PRR:2103656.2103670} and formalized the logic and its semantic model in Coq.
%The work attracted the interest of many researchers and many variants of the logic have been studied~\cite{2016arXiv160105047B,Barthe:2015:HAR:2676726.2677000,Barthe:2016:APC:2976749.2978391}.
%
The second is reformulating differential privacy as the sensitivity of functions between (pseudo)-metric spaces and overapproximate them using the type system.
Based on that approach, linear dependent type systems for reasoning about differential privacy were also introduced~\cite{GHHNP13,10.1145/2746325.2746335,10.1145/3589207}.
They are used for automated estimation of differential privacy~\cite{10.1145/3589207}.
The third is formalizing differential privacy directly in proof assistants. 
Tristan et al. are developing the library {\tt SampCert} for verified implementations for differential privacy using Lean~\cite{Tristan_SampCert_Verified_2024}.
It contains the formalization of differential privacy (and RDP and zCDP) in the discrete setting and the discrete Laplace and Gaussian mechanisms (see also \cite{canonne2021discretegaussiandifferentialprivacy}). 

\paragraph{Motivation}
To our knowledge, there has been no formalization of differential privacy in a proof assistant based on the continuous model of probabilistic programs.
In this paper, we propose an Isabelle/HOL library for semantic-level formal verification for differential privacy, supporting \emph{continuous} probability distributions. 

We consider the continuous setting for the following reasons.
First, the continuous setting is more general than the discrete setting. 
Once we formalize differential privacy in the continuous setting, we can apply it to the discrete setting immediately.
In particular, Isabelle/HOL's discrete probability distributions are implemented as probability measures on discrete measurable spaces.
In other words, they are instances of a continuous model.
Second, even in the study of differential privacy in the discrete setting, we possibly need the results in the continuous setting.
For example, in the study~\cite{Mironov2012OnSO} proposing a secure floating-point implementation of the Laplace mechanism, its continuous version (called ``the ideal mechanism'' in the paper) is used in the proof.
In addition, many pencil and paper proofs about differential privacy are written in the continuous setting.
Thus, it is natural and convenient for us to implement them directly.

\paragraph{Contribution}
Our work is based on the seminal textbook~\cite{DworkRothTCS-042} of differential privacy written by Dwork and Roth.
We chose the report noisy max mechanism as the main example in this paper.
It is a famous nontrivial example of differentially private mechanisms, and it is complicated to formalize in relational program logics such as apRHL: we often need tricky loop invariants.
In contrast, our formalization follows the pencil and paper proof.

In this paper, we show the following contributions:
First, we formalize differential privacy, and gave formal proofs of its basic properties.
Second, we implement the Laplace mechanism, and gave formal proof of its differential privacy.
Finally, we formalize the differential privacy of the report noisy max mechanism.
We also formalize the Laplace distribution, and the divergence given in ~\cite{BartheOlmedo2013} for another formulation of differential privacy.

\section{Preliminaries}
We write $\Nat$ and $\Real$ for the set of natural numbers (including $0$) and the set of real numbers respectively.
We write $[0,\infty)$ for the set of nonnegative real numbers.
We write $[0,\infty]$ for the set $[0,\infty) \cup \{ \infty \}$ of \emph{extended} nonnegative real numbers.
By $\Nat \subseteq \Real$, we regard any $n \in \Nat$ as $n \in \Real$ if we need.
%%\tetsuya{TODO: $\vec{c}$ notations of tuples.}
%A $\sigma$-algebra $\Sigma_X$ over a set $X$ is a nonempty family of subsets of $X$ closed under countable union and complement.
%A measurable space is a set equipped with a $\sigma$-algebra on it.
%%
%For a measurable space $X$, 
%we write $|X|$ and $\Sigma_X$ for its underlying set and $\sigma$-algebra respectively. 
%%
%Every $S \in \Sigma_X$ is called a \emph{measurable subset} of $X$.
%We write $S \subseteq_{\mathrm{measurable}} X$ for $S \in \Sigma_X$.
\subsection{Measure Theory}
A \emph{measurable space} $X = (|X|,\Sigma_X)$ is a set $|X|$ equipped with a $\sigma$-algebra $\Sigma_X$ over $|X|$ which is a nonempty family of subsets of $|X|$ closed under countable union and complement.
For a measurable space $X$, 
we write $|X|$ and $\Sigma_X$ for its underlying set and $\sigma$-algebra respectively. 
%Every $S \in \Sigma_X$ is called a \emph{measurable subset} of $X$.
%We write $S \subseteq_{\mathrm{measurable}} X$ for $S \in \Sigma_X$.
If there is no confusion, we write just $X$ for the underlying set $|X|$.
For example, we often write $x \in X$ for $x \in |X|$.
A \emph{measurable function} $f \colon X \to Y$ is a function $f \colon |X| \to |Y|$ such that $f\inverse (S) \in \Sigma_X$ for all $S \in \Sigma_Y$.
A measure $\mu$ on a measurable space $X$ is a function $\mu \colon \Sigma_X \to [0,\infty]$
satisfying the countable additivity.
% such that $\mu(\bigcup_{i \in \Nat }A_i) = \sum_{i \in \Nat }\mu(A_i)$ for a pairwise disjoint family $\{A_i\}_{i \in \Nat}$ of measurable subsets of $X$.
%For two measure $\mu$ and $\pi$ on $X$, $\pi$ \emph{dominates} (or $\mu$ is absolutely continuous with respect to $\pi$) $\mu$ if $\pi(S) = 0 \implies \mu(S) = 0$ holds for all $S \in \Sigma_X$

A measure $\mu$ on $X$ is called \emph{probability measure} if $\mu(X) = 1$.
For a probability measure $\mu$ on $X$ and a predicate $S \in \Sigma_X$, 
we often write $\Pr_{x \sim \mu}[S(x)] $, $\Pr[S(\mu)]$ and $\Pr[\mu \in S]$ for $\mu(S)$.
 
\paragraph{Giry monad}
For a measurable space $X$, the measurable space $\prob(X)$ of probability measures is defined as follows:
$|\prob(X)|$ is the set of probability measures on $X$, and $\Sigma_{\prob(X)}$ is the coarsest $\sigma$-algebra making the mappings $\mu \mapsto \mu(A)$ ($A \in \Sigma_X$) measurable.
The constructor $\prob$ forms the monad $(\prob,\return, {\bind})$ called Giry monad~\cite{Giry1982}.
The unit $\return(x)$ is the Dirac distribution centered at $x \in X$ defined by $\return(x)(S) = 1$ if $x \in S$ and $\return(x) = 0$ otherwise.
The bind $\bind$ is defined by
$(\mu \bind f)(S) = \int f(x)(S) d\mu(x)$ for a measurable $f \colon X \to \prob(Y)$, $\mu \in \prob(X)$ and $S \in \Sigma_Y$.

We then interpret a probabilistic program as a measurable function $c \colon X \to \prob(Y)$, where
the sequential composition of $c \colon X \to \prob(Y)$ and $c' \colon Y \to \prob(Z)$ is $(\lambda x. c(x) \bind c' )$. 

\paragraph{Do notation}
To help intuitive understanding, we often use the syntactic sugar called the \emph{do notation}.
In this paper, we mainly use the following two translations:
\begin{align*}
\{ x \leftarrow e; e'\} &= e \bind (\lambda x. e')\\
\{ x \leftarrow e; y \leftarrow e'; \cdots \} &= \{ x \leftarrow e; \{y \leftarrow e'; \cdots\} \} 
\end{align*}

For example, the product $\mu_1 \otimes \mu_2$ of two probability distributions $\mu_1$ and $\mu_2$
can be written with the do notation.
\begin{align*}
(\mu_1 \otimes \mu_2) 
&= \{x \leftarrow \mu_1, y \leftarrow \mu_2; \return (x,y) \}\\
&= \{y \leftarrow \mu_2, x \leftarrow \mu_1; \return (x,y) \}
\end{align*}
Here, the second equality corresponds to the \emph{commutativity} (cf. \cite[Definition 3.1]{Kock1970}) of Giry monad.

\subsection{Measure Theory in Isabelle/HOL}
Throughout this paper, we work in the proof assistant Isabelle/HOL~\cite{10.5555/1791547}.
We use the standard library of Isabelle/HOL: HOL-Analysis and HOL-Probability~\cite{Avigad2017,10.1007/978-3-662-46669-8_4,10.1145/3018610.3018628,10.1007/978-3-642-22863-6_12,10.1007/978-3-662-49498-1_20,10.5555/1791547}.

The type \isa{{\isacharprime}a measure} is the type of a measure space, namely, a triple $(|X|,\Sigma_X,\mu)$ where $(|X|,\Sigma_X)$ is a measurable space, and $\mu$ is a measure on it. 
Projections for these components are provided in the standard library of Isabelle/HOL.
\begin{align*}
\isa{space\ }&\isa{:: {\isacharprime}a measure {\isasymRightarrow} {\isacharprime}a set } \\
\isa{sets\ }&\isa{:: {\isacharprime}a measure {\isasymRightarrow} {\isacharprime}a set set }\\
\isa{emeasure\ }&\isa{:: {\isacharprime}a measure {\isasymRightarrow} {\isacharprime}a set {\isasymRightarrow} ennreal}
\end{align*}
Here, \isa{ennreal} is the type for the set $[0,\infty]$ of extended nonnegrative real numbers.

We remark that the type \isa{{\isacharprime}a measure} is used for implementing both measures and measurable spaces.

%The constant \isa{borel :: {\isacharprime}a topological{\isacharunderscore}space measure{\kern0pt}} 
%on topological space type class is the constructor of a Borel space, which is the measurable space whose $\sigma$-algebra is generated from the topology.
%The constant \isa{lborel} is the Lebesgue measure\footnote{Strictly speaking, the Lebesgue measure \isa{completion lborel}.}.
In this paper, we use \isa{borel} for the usual Borel space $\Real$ of the real line and \isa{lborel} for the Lebesgue measure on it\footnote{Strictly speaking, the Lebesgue measure is \isa{completion lborel}.}.
The set of measurable functions from a measurable space \isa{M} to \isa{N} 
is denoted by \isa{M\ {\isasymrightarrow}\isactrlsub M\ N}.
For a measure \isa{M :: {\isacharprime}a measure} and a measurable function \isa{f\ {\isasymin}\ M\ {\isasymrightarrow}\isactrlsub M borel},
the Lebesgue integral (over \isa{A}) is denoted by \isa{{\isasymintegral}x{\isachardot}{\kern0pt}\ f\ x\ {\isasympartial}M} (resp. \isa{{\isasymintegral}x{\isasymin}A{\isachardot}{\kern0pt}\ f\ x\ {\isasympartial}M}).

HOL-Probability also contains the formalization of Giry monad and its commutativity\footnote{It is formalized as \isa{bind\_rotate} in the theory \isa{Giry\_Monad}.}.
Each component of the triple $(\prob,\return,{\bind})$ is provided as the following constants:
\begin{align*}
&\lefteqn{\isa{prob\_algebra\ }\isa{:: {\isacharprime}a measure {\isasymRightarrow}}\isa{ {\isacharprime}a measure measure }}\\
&\isa{return\ }\isa{:: {\isacharprime}a measure {\isasymRightarrow}}\isa{ {\isacharprime}a {\isasymRightarrow} {\isacharprime}a measure}\\
&\isa{{\isasymbind}\ }\isa{:: {\isacharprime}a measure {\isasymRightarrow}}\isa{ {\isacharparenleft}{\isacharprime}a {\isasymRightarrow} {\isacharprime} b measure {\isacharparenright} {\isasymRightarrow} {\isacharprime}b measure}
\end{align*}

\section{Differential Privacy} \label{sec:DifferentialPrivacy}

%\tetsuya{In this section, change $X$ to $\Nat^{|\mathcal{X}|}$, and explain the adjacency. }DONE

Differential privacy (DP) is a statistical definition of privacy introduced by Dwork et al.~\cite{Dwork2006,DworkMcSherryNissimSmith2006}. 
It is a quantitative standard of privacy of a program processing data stored in a dataset (database) for noise-adding anonymization.

In this section, we recall the standard definition of differential privacy in the textbook \cite{DworkRothTCS-042} of Dwork and Roth.

We first formulate the domain of datasets as $\Nat^{|\mathcal{X}|}$ where $\mathcal{X}$ is a set of data types.
Each dataset $D \in \Nat^{|\mathcal{X}|}$ is a histogram in which each entry $D[i]$ represents the number of elements of type $i$, where we regard $0 \leq i < |\mathcal{X}|$. 
We can define a metric ($L_1$-norm) on $\Nat^{|\mathcal{X}|}$ as follows:
\[
\| D - D' \|_1 = \sum_{0 \leq i < |\mathcal{X}|} |D[i] - D'[i]|.
\] 
When $\| D - D' \|_1 \leq 1$, the datasets $D, D' \in \Nat^{|\mathcal{X}|}$ are called \emph{adjacent}.
Then they are different only in one type $i$ in $\mathcal{X}$. 

We then give the definition of differential privacy.
\begin{definition}[{\cite[Def. 2.4]{DworkRothTCS-042}} (cf. {\cite{10.1007/11761679_29}})]\label{def:DP}
Let $M \colon \Nat^{|\mathcal{X}|}\to \prob(Y)$ be a randomized algorithm. 
$M$ is $(\varepsilon,\delta)$-\emph{differentially private (DP)} if for any adjacent datasets $D, D' \in \Nat^{|\mathcal{X}|}$, 
\[
 \forall S \in \Sigma_Y.~ \Pr[M(D) \in S] \leq \exp(\varepsilon) \Pr[M(D') \in S] + \delta.
\]
\end{definition}
Here $0 \leq \delta$ and $0 \leq \varepsilon$. 
Intuitively, $\varepsilon$ indicates the upper bound of the probability ratio between $M(D)$ and $M(D')$ (called privacy loss), and $\delta$ is the error.
The distributions of outputs $M(D)$ and $M(D')$ are harder to distinguish if $\varepsilon$ and $\delta$ are small.
In particular, if $(\varepsilon, \delta) = (0,0)$ then $M(D) = M(D')$.
%\tetsuya{TODO: Mechanize this fact.}DONE

\subsection{Basic Properties of Differential Privacy}

Differential privacy has the following basic properties.
They enable us to estimate the differential privacy of large algorithms from their smaller components.

\begin{lemma}\label{DP:basic:trivial}
Supoose that $M \colon \Nat^{|\mathcal{X}|} \to \prob(Y)$ is $(\varepsilon,\delta)$-DP.
If $\varepsilon \leq \varepsilon'$ and $\delta \leq \delta'$ then $M$ is also $(\varepsilon',\delta')$-DP.
\end{lemma}
Differential privacy is stable for post-processing.
\begin{lemma}[post-processing~{\cite[Prop. 2.1]{DworkRothTCS-042}}]\label{DP:basic:postprocessing}
If $M \colon \Nat^{|\mathcal{X}|}\to \prob(Y)$ is $(\varepsilon,\delta)$-DP then $(\lambda D. M(D)\bind N) \colon \Nat^{|\mathcal{X}|} \to \prob(Z)$ is $(\varepsilon,\delta)$-DP for any randomized mapping $N \colon Y \to \prob(Z)$.
\end{lemma}
The following lemma tells that any $(\varepsilon,0)$-DP algorithm is $(k \varepsilon,0)$-DP 
for a group of datasets with radius $k$. 
\begin{lemma}[group privacy~{\cite[Thm. 2.2]{DworkRothTCS-042}}]\label{DP:basic:group}
Suppose that $M \colon \Nat^{|\mathcal{X}|}\to \prob(Y)$ is $(\varepsilon,0)$-DP. If $\| D - D' \|_1 \leq k$ then
\[
 \forall S \in \Sigma_Y.~ \Pr[M(D) \in S] \leq \exp(k\cdot \varepsilon) \Pr[M(D') \in S].
\]
\end{lemma}
\emph{Composition theorems} are a strong feature of differential privacy.
We can estimate the diferential privacy of complex algorithms using ones of their components.
\begin{lemma}[{\cite[Thm. 3.14]{DworkRothTCS-042}}]\label{DP:basic:sequential}
If $M \colon \Nat^{|\mathcal{X}|} \to \prob(Y)$ is $(\varepsilon_1,\delta_1)$-DP and $N \colon \Nat^{|\mathcal{X}|} \to \prob(Z)$ is $(\varepsilon_2,\delta_2)$-DP then 
the composite argorithm $(\lambda D.~M(D) \otimes N(D) ) \colon \Nat^{|\mathcal{X}|} \to \prob(Y \times Z)$ is $(\varepsilon_1+\varepsilon_2,\delta_1+\delta_2)$-DP.
\end{lemma}
\begin{lemma}[{\cite[Thm. B.1]{DworkRothTCS-042}}]\label{DP:basic:adaptive}
If $M \colon \Nat^{|\mathcal{X}|} \to \prob(Y)$ is $(\varepsilon_1,\delta_1)$-DP and $N \colon \Nat^{|\mathcal{X}|} \times Y \to \prob(Z)$ is $(\varepsilon_2,\delta_2)$-DP then 
the composite argorithm $(\lambda D. \{z \leftarrow M (D); N(D, Z)\} ) \colon \Nat^{|\mathcal{X}|} \to \prob(Z)$ is $(\varepsilon_1+\varepsilon_2,\delta_1+\delta_2)$-DP.
\end{lemma}
%\tetsuya{Do we say ``sequential'' composition, ``adaptive'' compositions?}
%\tetsuya{Mechanize parallel compositions?}

\subsection{Laplace Mechanism}\label{sec:laplace_mechanism_sensitivity}

The \emph{Laplace mechanism} is a most-typical differentially private mechanism with the noise sampled from the Laplace distributions.

\subsubsection{Laplace Noise}

The Laplace distribution $\mathrm{Lap}(b,z) \in \prob(\Real)$ with scale $0 < b$ and average $z \in \Real$
is a continuous probability distribution on $\Real$ whose density function and cumulative distribution function are given as follows: 
\begin{align*}
f_{\mathrm{Lap}(b,z)}(x) &= \frac{1}{2b}\exp\left(- \frac{|x - z|}{b}\right) \\
c_{\mathrm{Lap}(b,z)}(x) &= 
\begin{cases}
\frac{1}{2}\exp(\frac{x - z}{b}) & x \leq z \\
1 - \frac{1}{2}\exp(-\frac{x - z}{b}) & x \geq z 
\end{cases}
\end{align*}
We here write $\mathrm{Lap}(b)$ for $\mathrm{Lap}(b,0)$.

We later use the following fact:
$\mathrm{Lap}(b,z) $ can be obtained by adding the given $z$ and the noise sampled from $\mathrm{Lap}(b)$.
\begin{lemma}\label{lem:Lap:shifting}
For any $0 < b$ and $z \in \Real$, 
\[
\mathrm{Lap}(b,z) = \{ x \leftarrow \mathrm{Lap}(b); \return~(z + x) \}.
\]
\end{lemma}

For the differential privacy of Laplace mechanism, the following lemma is essential:
\begin{lemma}[{continuous version of \cite[Proposition 7]{2016arXiv160105047B}}]\label{lem:Lap:DP:finer}
Let $x, y \in \Real$ and $0 < b$. 
If $|x - y| \leq r$ then 
\[
\forall S \in\Sigma_\Real. \Pr[\mathrm{Lap}(b,x) \in S] \leq \exp(r/b) \Pr[\mathrm{Lap}(b,y) \in S].
\]
\end{lemma}
%\begin{proof}[Proof sketch]
%For all $z \in \Real$, we have $f_{\mathrm{Lap}(b,x)}(z) \leq \exp(r/b) \cdot f_{\mathrm{Lap}(b,x')}(z)$ because $- |x' - z| \leq r - |x - z|$.
%By the monotonicity of integrals (over $S$), we conclude the lemma.
%\end{proof}
\subsubsection{Laplace Mechanism}

Consider a function $f \colon \Nat^{|\mathcal{X}|} \to \Real^m$ from datasets to $m$-tuples of values.
Intuitively, $f$ is a query (or queries) from a dataset $D \in \Nat^{|\mathcal{X}|}$, and 
the Laplace mechanism anonymizes $f$ 
by adding the Laplace noise to each component of the outputs of $f$. 

The Laplace mechanism $\mathrm{LapMech}_{f,m,b}$ is defined as the following procedure:
for a dataset $D \in \Nat^{|\mathcal{X}|}$, we take the componentwise sum
\[
f(D) + (r_0,\ldots,r_{m -1}) 
\] 
where each $r_i$ is the noise sampled from $\mathrm{Lap}(b)$.

We can rewrite $\mathrm{LapMech}_{f,m,b}$ using the do notation:
\begin{equation}\label{eq:LapMech_def2}
\mathrm{LapMech}_{f,m,b}(D) = \{ \vec{r} \leftarrow \mathrm{Lap}(b)^m; \return f(D) + \vec{r}\}.
\end{equation}
Here, $\vec{r} = (r_0,\ldots,r_{m -1})$, and $\mathrm{Lap}(b)^m \in \prob(\Real^m)$ is the distribution of $m$ values sampled independently from $\mathrm{Lap}(b)$. 
We can define $\mathrm{Lap}(b)^m$ inductively as follows:
\begin{align*}
\mathrm{Lap}(b)^0 &= \return (),\\
\mathrm{Lap}(b)^{k+1} &= \{ \vec{r} \leftarrow \mathrm{Lap}(b)^{k}; r_0 \leftarrow \mathrm{Lap}(b); \return (r_0,\vec{r}) \}.
\end{align*}

For simple formalization,
we can rewrite $\mathrm{LapMech}_{f,m,b}$:
\begin{equation}\label{eq:LapMech_def2_rec}
\mathrm{LapMech}_{f,m,b}(D) = \mathrm{Lap}^m(b,f(D))
\end{equation}
Here, the procedure $\mathrm{Lap}^m(b,\vec{x})$ adding $\vec{r}$ sampled from $\mathrm{Lap}^m(b)$ to $\vec{x}$
is defined inductively in the similar way as $\mathrm{Lap}(b)^m$.
%\begin{align*}
%&\mathrm{Lap}^0(b,()) = \return ()\\
%&\mathrm{Lap}^{k+1}(b,(x_0,\vec{x})) \\
%&\qquad =\{ \vec{r} \leftarrow \mathrm{Lap}^k(b,\vec{x}); r_0 \leftarrow \mathrm{Lap}(b,x_0); \return (r_0,\vec{r}) \}.
%\end{align*}
Of course, we have $\mathrm{Lap}^m(b) = \mathrm{Lap}^m(b,(0,\ldots,0))$, and 
\begin{equation}\label{eq:LapMech_def3}
\mathrm{Lap}^m(b,\vec{x}) = \{ \vec{r} \leftarrow \mathrm{Lap}^m(b); \return (\vec{x}+\vec{r})\}.
\end{equation}

To make the Laplace mechanism $(\varepsilon,0)$-differentially private for a fixed $0 < \varepsilon$, 
we define the $L_1$-sensitivity of $f$ by 
\[
\Delta f = \sup\{~ \| f(D_1) - f(D_2) \|_1 ~|~ \| D_1 - D_2 \|_ 1 \leq 1 \}. 
\]
We then obtain
\begin{lemma} \label{prop.DP.LapMech}
$\mathrm{LapMech}_{f,m,b}$ is $(\Delta f /b,0)$-DP.
\end{lemma}
By setting $b = \Delta f / \varepsilon$, we conclude the desired property:
\begin{proposition}[{\cite[Theorem 3.6]{DworkRothTCS-042}}]\label{prop:DP:Lapmech}
The Laplace mechanism $\mathrm{LapMech}_{f,m,\Delta f / \varepsilon}$ is $(\varepsilon,0)$-differentially private.
\end{proposition}

%\tetsuya{TODO: Parhaps, we do not need ereal? }
\section{Divergence for Differential Privacy}\label{sec:DP_divergence}

Before formalizing differential privacy, we formalize the divergence\footnote{A divergence is a kind of metric between probability distributions. They may not satisfy the axioms of a metric function.} $\Delta^\varepsilon$ ($0 \leq \varepsilon$) for differential privacy introduced in \cite{BartheOlmedo2013,olmedo2014approximate}.
It is useful for the formalization of differential privacy.
First, basic properties of differential privacy shown in Section \ref{sec:DifferentialPrivacy} are derived from ones of $\Delta^\varepsilon$.
Second, $\Delta^\varepsilon$ is more compatible with the structure of Giry monad.
it forms a divergence on Giry monad, which is also an essential categorical structure of relational program logics reasoning about differential privacy~\cite{Sato_Katsumata_2023}.
In addition, the conversion law form R\'{e}nyi differential privacy~\cite{Mironov17} to the standard differential privacy
is derived from the conversion from R\'{e}nyi divergence to $\Delta^\varepsilon$~\cite{DBLP:journals/corr/abs-1905-09982}.

We define the divergence $\Delta^\varepsilon$ ($0 \leq \varepsilon$) below:
\begin{definition}[{\cite{BartheOlmedo2013}, \cite[Section 5.8]{Sato_Katsumata_2023}} ]
For each measurable space $X$ and $0 \leq \varepsilon$, we define a function $\Delta^\varepsilon_X \colon \prob(X) \times \prob(X) \to [0,\infty)$ by
\[
\Delta^\varepsilon_X (\mu,\nu) = \sup \{\mu(S) - \exp(\varepsilon) \nu(S) ~|~ S \in \Sigma_X \}.
\]
\end{definition}

Differential privacy can be reformulated equivalently using the divergence $\Delta^\varepsilon$:
\begin{lemma}\label{lem:divergenceDP:equivalence_DP}
For any pair $\mu_1,\mu_2 \in \prob(X)$ of probability distributions on a measurable space $X$, we obtain
\begin{align*}
&\Delta^\varepsilon_X (\mu,\nu) \leq \delta \\
&\iff \forall S \in \Sigma_X.~ \Pr[\mu_1 \in S] \leq \exp(\varepsilon) \Pr[\mu_2\in S] + \delta.
\end{align*}
\end{lemma}
\begin{corollary}[cf. {\cite[Section 3.1]{BartheOlmedo2013}}]
A randomized algorithm $M \colon \Nat^{|\mathcal{X}|} \to \prob(Y)$ is $(\varepsilon,\delta)$-DP \emph{if and only if} for any adjacent datasets $D,D' \in \Nat^{|\mathcal{X}|}$,
$\Delta_Y^\varepsilon(M(D),M(D')) \leq \delta$ holds.
\end{corollary}

The divergence $\Delta^\varepsilon$ have the following basic properties (see \cite{BartheOlmedo2013,olmedo2014approximate,Sato_Katsumata_2023}):
for all $\mu, \nu \in \prob(X)$ and $f , g \colon X \to \prob(Y)$,
\begin{description}
\item[Nonnegativity] $0 \leq \Delta^\varepsilon_X(\mu,\nu)$.
\item[Reflexivity] $\Delta^0_X(\mu,\mu) = 0$.
\item[Monotonicity] If $\varepsilon \leq \varepsilon'$ then $\Delta^\varepsilon_X(\mu,\nu) \geq \Delta^{\varepsilon'}_X(\mu,\nu)$.
\item[Composability] If $ \Delta^{\varepsilon_1}_X (\mu,\nu) \leq \delta_1$ and $\Delta^{\varepsilon_2}_Y (f(x),g(x)) \leq \delta_2$ for all $x \in X$
then $\Delta^{\varepsilon_1 + \varepsilon_2}_Y (\mu \bind f, \nu \bind g) \leq \delta_1 + \delta_2$.
\end{description}

Basic properties of differential privacy can be derived from those of the divergence $\Delta^\varepsilon$.
Lemma \ref{DP:basic:trivial} is proved by the monotonicity of $\Delta^\varepsilon$.
We here emphasize that the postprocessing property (Lemma \ref{DP:basic:postprocessing}) and both composition theorems (Lemmas \ref{DP:basic:sequential} and \ref{DP:basic:adaptive}) are derived from the reflexivity and composability of $\Delta^\varepsilon$.

The group privacy (Lemma \ref{DP:basic:group}) can be derived from the
specific properties of $\Delta^\varepsilon$ and the metric space $(\Nat^{|\mathcal{X}|},\| - \|_1)$:
\begin{description}
\item[``Transitivity''] Let $\mu_1, \mu_2,\mu_3 \in \prob(X)$.
If $\Delta^{\varepsilon_1}_X(\mu_1,\mu_2) \leq 0$ and $\Delta^{\varepsilon_2}_X(\mu_2,\mu_3) \leq 0$
then $\Delta^{\varepsilon_1+\varepsilon_2}_X(\mu_1,\mu_3) \leq 0$.
\end{description}
\begin{lemma}\label{lem:adj_k:group}
Let $k\in \Nat$, $D,D' \in \Nat^{|\mathcal{X}|}$. 
If $\|D - D' \|_1 \leq k$ then we obtain $(D,D') \in R^k$
where $R = \{(D,D') | D,D'\colon\text{adjacent}\}$.
\end{lemma}

The nonnegativity, reflexivity, monotonicity, and ``transitivity'' of $\Delta^\varepsilon$ are easy to prove.
The composability was proved in the discrete setting~\cite{BartheOlmedo2013,olmedo2014approximate}, then it was extended to the continuous setting~\cite{Sato2016MFPS,DBLP:conf/lics/SatoBGHK19,Sato_Katsumata_2023}.
Combining these existing proofs, we provide a simplified but detailed pencil and paper proof.

\begin{proof}[A simplified proof of composability]
We assume
$\Delta^{\varepsilon_1}_X (\mu,\nu) \leq \delta_1$
and 
$\Delta^{\varepsilon_2}_Y (f(x),g(x)) \leq \delta_2$ ($\forall x \in X$).
Let $S' \in \Sigma_Y$, and 
\begin{align*}
F_{S'} (x) &= \max(0,f(x)(S')- \delta_2) && (x \in X)\\
G_{S'}(x) &= \min(1,\exp(\varepsilon_2) g(x)(S')) && (x \in X)
\end{align*}
From the assumption, we obtain
\begin{align}
\label{DP:comp:condition1}
\mu(S) - \exp(\varepsilon_1) \nu(S) \leq \delta_1 & && (S \in \Sigma_X) \tag{a}\\
 \label{DP:comp:condition2} 
 0 \leq F_{S'} (x) \leq G_{S'} (x) \leq 1 &&& (x \in X) \tag{b}
\end{align}
Consider the measure $\pi = \mu+\nu$ on $X$ defined by $\pi(S) = \mu(S)+\nu(S)$.
It is finite, and dominates both $\mu$ and $\nu$ (i.e. $\mu$ and $\nu$ are absolutely continuous with respect to $\pi$).
Hence, by Radon-Nikod\'{y}m theorem, 
we can take the density functions (Radon-Nikod\'{y}m derivatives) ${d\mu}/{d\pi}$ and ${d\nu}/{d\pi}$ of $\mu$ and $\nu$, respectively.
%We define $B = \left\{ x ~\middle|~ 0 \leq \left(\frac{d\mu}{d\pi}(x) - e^{\varepsilon_1}\frac{d\nu}{d\pi}(x) \right) \right\}$.
Then, from (\ref{DP:comp:condition1}) and (\ref{DP:comp:condition2}), we evaluate as follows:
\begin{align*}
&(\mu \bind f)(S') - \exp(\varepsilon_1+\varepsilon_2)(\nu \bind g)(S')\\
&= \int f(x)(S') d\mu(x) - \exp(\varepsilon_1+\varepsilon_2) \int g(x)(S) d\nu(x)\\
%&= \int f(x)(S') \cdot \frac{d\mu}{d\pi}(x)~d\pi(x) - e^{\varepsilon_1+\varepsilon_2} \int g(x)(S') \cdot\frac{d\nu}{d\pi}(x) ~d\pi(x)\\
&= \int f(x)(S') \frac{d\mu}{d\pi}(x) - \exp(\varepsilon_1+\varepsilon_2) g(x)(S') \frac{d\nu}{d\pi}(x) ~d\pi\\
&\leq \int \left(F_{S'} (x) +\delta_2\right) \frac{d\mu}{d\pi}(x) - \exp(\varepsilon_1) G_{S'} (x) \frac{d\nu}{d\pi}(x) ~d\pi(x)\\
&= \int F_{S'} (x) \frac{d\mu}{d\pi}(x) - \exp(\varepsilon_1) G_{S'} (x) \frac{d\nu}{d\pi}(x) ~d\pi(x) + \delta_2 \\
&\leq \int  G_{S'} (x) \frac{d\mu}{d\pi}(x) - \exp(\varepsilon_1) G_{S'} (x) \frac{d\nu}{d\pi}(x) ~d\pi(x) + \delta_2 \\
&\leq \int_{x \in B} 
\left( \frac{d\mu}{d\pi}(x) - \exp(\varepsilon_1)\frac{d\nu}{d\pi}(x)\right)
\cdot G_{S'} (x) ~d\pi(x) + \delta_2\\
&\leq \mu(B) - \exp(\varepsilon_1)\nu(B)+ \delta_2\\
&\leq \delta_1 + \delta_2
\end{align*}
Here, $B = \left\{ x ~\middle|~ 0 \leq \left(\frac{d\mu}{d\pi}(x) - \exp(\varepsilon_1)\frac{d\nu}{d\pi}(x) \right) \right\}$.
Since $S' \in \Sigma_Y$ is arbitrary, we conclude $\Delta^{\varepsilon_1 + \varepsilon_2}_Y (\mu \bind f ,\nu \bind g) \leq \delta_1+\delta_2$.
\end{proof}

\subsection{Divergence $\Delta^\varepsilon$ in Isabelle HOL}

We formalize the divergence $\Delta^\varepsilon$ in Isabelle/HOL.
\begin{isabelle}
\isacommand{definition}\isamarkupfalse%
\ DP{\isacharunderscore}{\kern0pt}divergence\ {\isacharcolon}{\kern0pt}{\isacharcolon}{\kern0pt}\ {\isachardoublequoteopen}{\isacharprime}{\kern0pt}a\ measure\ {\isasymRightarrow}\ {\isacharprime}{\kern0pt}a\ measure\ {\isasymRightarrow}\ real\ {\isasymRightarrow}\ ereal\ {\isachardoublequoteclose}\ \isakeyword{where}\isanewline
\ \ {\isachardoublequoteopen}DP{\isacharunderscore}{\kern0pt}divergence\ M\ N\ {\isasymepsilon}\ {\isacharequal}{\kern0pt}\ Sup\ {\isacharbraceleft}{\kern0pt}ereal{\isacharparenleft}{\kern0pt}measure\ M\ A\ {\isacharminus}{\kern0pt}\ {\isacharparenleft}{\kern0pt}exp\ {\isasymepsilon}{\isacharparenright}{\kern0pt}\ {\isacharasterisk}{\kern0pt}\ measure\ N\ A{\isacharparenright}{\kern0pt}\ {\isacharbar}{\kern0pt}\ A{\isacharcolon}{\kern0pt}{\isacharcolon}{\kern0pt}{\isacharprime}{\kern0pt}a\ set{\isachardot}{\kern0pt}\ A\ {\isasymin}\ {\isacharparenleft}{\kern0pt}sets\ M{\isacharparenright}{\kern0pt}{\isacharbraceright}{\kern0pt}{\isachardoublequoteclose}
\end{isabelle}
We here remark that \isa{M} and \isa{N} are general measures of type \isa{{\isacharprime}{\kern0pt}a\ measure}, and \isa{\isasymepsilon} may be negative.
We restrict them later.

The equivalence between $\Delta^\varepsilon$ and the inequality of differential privacy (Lemma \ref{lem:divergenceDP:equivalence_DP}) is formalized as follows:
\begin{isabelle}
\isacommand{lemma}\isamarkupfalse%
\ DP{\isacharunderscore}{\kern0pt}divergence{\isacharunderscore}{\kern0pt}forall{\isacharcolon}{\kern0pt}\isanewline
\ \ \isakeyword{shows}\ {\isachardoublequoteopen}{\isacharparenleft}{\kern0pt}{\isasymforall}\ A\ {\isasymin}\ sets\ M{\isachardot}{\kern0pt}\ measure\ M\ A\ {\isacharminus}{\kern0pt}\ {\isacharparenleft}{\kern0pt}exp\ {\isasymepsilon}{\isacharparenright}{\kern0pt}\ {\isacharasterisk}{\kern0pt}\ measure\ N\ A\ {\isasymle}\ {\isacharparenleft}{\kern0pt}{\isasymdelta}\ {\isacharcolon}{\kern0pt}{\isacharcolon}{\kern0pt}\ real{\isacharparenright}{\kern0pt}{\isacharparenright}{\kern0pt}\ {\isasymlongleftrightarrow}\ DP{\isacharunderscore}{\kern0pt}divergence\ M\ N\ {\isasymepsilon}\ {\isasymle}\ {\isacharparenleft}{\kern0pt}{\isasymdelta}\ {\isacharcolon}{\kern0pt}{\isacharcolon}{\kern0pt}\ real{\isacharparenright}{\kern0pt}{\isachardoublequoteclose}
\end{isabelle}

For the simplicity of formalization, we choose \isa{ereal} (extended real) for the type of values of the divergence instead of \isa{real} and \isa{ennreal} (extended nonnegative real).

The first reason is that \isa{ereal} forms a complete linear order\footnote{linear order + complete lattice}.
We could formalize $\Delta^\varepsilon$ using \isa{real} instead of \isa{ereal}: 
\begin{isabelle}
\isacommand{definition}\isamarkupfalse%
\ DP{\isacharunderscore}{\kern0pt}divergence{\isacharunderscore}{\kern0pt}real\ {\isacharcolon}{\kern0pt}{\isacharcolon}{\kern0pt}\ {\isachardoublequoteopen}{\isacharprime}{\kern0pt}a\ measure\ {\isasymRightarrow}\ {\isacharprime}{\kern0pt}a\ measure\ {\isasymRightarrow}\ real\ {\isasymRightarrow}\ real\ {\isachardoublequoteclose}\ \isakeyword{where}\isanewline
\ \ {\isachardoublequoteopen}DP{\isacharunderscore}{\kern0pt}divergence{\isacharunderscore}{\kern0pt}real\ M\ N\ {\isasymepsilon}\ {\isacharequal}{\kern0pt}\ Sup\ {\isacharbraceleft}{\kern0pt}\ measure\ M\ A\ {\isacharminus}{\kern0pt}\ {\isacharparenleft}{\kern0pt}exp\ {\isasymepsilon}{\isacharparenright}{\kern0pt}\ {\isacharasterisk}{\kern0pt}\ measure\ N\ A\ {\isacharbar}{\kern0pt}\ A{\isacharcolon}{\kern0pt}{\isacharcolon}{\kern0pt}{\isacharprime}{\kern0pt}a\ set{\isachardot}{\kern0pt}\ A\ {\isasymin}\ {\isacharparenleft}{\kern0pt}sets\ M{\isacharparenright}{\kern0pt}{\isacharbraceright}{\kern0pt}{\isachardoublequoteclose}
\end{isabelle}
It is equal to \isa{DP{\isacharunderscore}{\kern0pt}divergence} for probability measures on \isa{L}.
\begin{isabelle}
\isacommand{lemma}\isamarkupfalse%
\ DP{\isacharunderscore}{\kern0pt}divergence{\isacharunderscore}{\kern0pt}is{\isacharunderscore}{\kern0pt}real{\isacharcolon}{\kern0pt}\isanewline
\ \ \isakeyword{assumes}\ M{\isacharcolon}{\kern0pt}\ {\isachardoublequoteopen}M\ {\isasymin}\ space\ {\isacharparenleft}{\kern0pt}prob{\isacharunderscore}{\kern0pt}algebra\ L{\isacharparenright}{\kern0pt}{\isachardoublequoteclose}\isanewline
\ \ \ \ \isakeyword{and}\ N{\isacharcolon}{\kern0pt}\ {\isachardoublequoteopen}N\ {\isasymin}\ space\ {\isacharparenleft}{\kern0pt}prob{\isacharunderscore}{\kern0pt}algebra\ L{\isacharparenright}{\kern0pt}{\isachardoublequoteclose}\ \isanewline
\ \ \isakeyword{shows}\ {\isachardoublequoteopen}DP{\isacharunderscore}{\kern0pt}divergence\ M\ N\ {\isasymepsilon}\ {\isacharequal}{\kern0pt}\ DP{\isacharunderscore}{\kern0pt}divergence{\isacharunderscore}{\kern0pt}real\ M\ N\ {\isasymepsilon}{\isachardoublequoteclose}
\end{isabelle}
However, \isa{DP{\isacharunderscore}{\kern0pt}divergence{\isacharunderscore}{\kern0pt}real\ M\ N} is not well-defined if \isa{M} and \isa{N} are not finite.
It is not convenient for formalizing properties of $\Delta^\varepsilon$. 
For example, we need to add extra assumptions to \isa{DP{\isacharunderscore}{\kern0pt}divergence{\isacharunderscore}{\kern0pt}forall} when we choose \isa{real}.
In contrast, it is convenient that \isa{DP{\isacharunderscore}{\kern0pt}divergence\ M\ N} is always defined.

The second reason is that \isa{ereal} is more convenient than another complete linear order \isa{ennreal},
while we later show the nonnegativity of \isa{DP{\isacharunderscore}{\kern0pt}divergence}.
It supports addition and subtraction in the usual sense\footnote{For example, $(1 - 2) + 2 = 1$ holds in \isa{ereal}, but it does not hold in \isa{ennreal}.}.
It is helpful for formalizing properties of $\Delta^\varepsilon$.
In particular, the proof of composability uses many transpositions of terms.

\subsubsection{A Locale for Two Probability Measures on a Common Space} \label{sec:LocaleComparableMeasures}

In the formal proofs of properties of $\Delta^\varepsilon$, 
we often need to use many mathematical properties about two probability distributions on a \emph{common} measurable space.
In particular, in the proof of composability of $\Delta^\varepsilon$, for given two probability distributions $\mu,\nu \in \prob (X)$, we take the density functions $d\mu/d\pi$ and $d\nu/d\pi$ with respect to the sum $\pi = \mu+\nu$.

For these reasons, we introduce the following locale for two probability distributions on a common measurable space, which contains the following features:
\begin{itemize}
\item lemmas for conversions of the underlying sets and $\sigma$-algebras among \isa{M}, \isa{N}, \isa{sum{\isacharunderscore}{\kern0pt}measure\ M\ N} and \isa{L}.
\item the density functions \isa{dM} and \isa{dN} of \isa{M} and \isa{N} with respect to \isa{sum{\isacharunderscore}{\kern0pt}measure\ M\ N}.
\item properties of \isa{dM} and \isa{dN}: being density functions, measurability, boundedness, and integrability. 
\end{itemize}
\begin{isabelle}
\isacommand{locale}\isamarkupfalse%
\ comparable{\isacharunderscore}{\kern0pt}probability{\isacharunderscore}{\kern0pt}measures\ {\isacharequal}{\kern0pt}\isanewline
\ \ \isakeyword{fixes}\ L\ M\ N\ {\isacharcolon}{\kern0pt}{\isacharcolon}{\kern0pt}\ {\isachardoublequoteopen}{\isacharprime}{\kern0pt}a\ measure{\isachardoublequoteclose}\isanewline
\ \ \isakeyword{assumes}\ M{\isacharcolon}{\kern0pt}\ {\isachardoublequoteopen}M\ {\isasymin}\ space\ {\isacharparenleft}{\kern0pt}prob{\isacharunderscore}{\kern0pt}algebra\ L{\isacharparenright}{\kern0pt}{\isachardoublequoteclose}\isanewline
\ \ \ \ \isakeyword{and}\ N{\isacharcolon}{\kern0pt}\ {\isachardoublequoteopen}N\ {\isasymin}\ space\ {\isacharparenleft}{\kern0pt}prob{\isacharunderscore}{\kern0pt}algebra\ L{\isacharparenright}{\kern0pt}{\isachardoublequoteclose}\isanewline
\isakeyword{begin}
\isanewline
\ \ ...
\isanewline
\isacommand{lemma}\isamarkupfalse%
\ spaceM{\isacharbrackleft}{\kern0pt}simp{\isacharbrackright}{\kern0pt}{\isacharcolon}{\kern0pt}\ {\isachardoublequoteopen}sets\ M\ {\isacharequal}{\kern0pt}\ sets\ L{\isachardoublequoteclose}
\isanewline
\isacommand{lemma}\isamarkupfalse%
\ spaceN{\isacharbrackleft}{\kern0pt}simp{\isacharbrackright}{\kern0pt}{\isacharcolon}{\kern0pt}\ {\isachardoublequoteopen}sets\ N\ {\isacharequal}{\kern0pt}\ sets\ L{\isachardoublequoteclose}
\isanewline
\ \ ...
\isanewline
\isacommand{lemma}\isamarkupfalse%
\ McontMN{\isacharbrackleft}{\kern0pt}simp{\isacharcomma}{\kern0pt}intro{\isacharbrackright}{\kern0pt}{\isacharcolon}\isanewline \ \ {\kern0pt}\ {\isachardoublequoteopen}absolutely{\isacharunderscore}{\kern0pt}continuous\ {\isacharparenleft}{\kern0pt}sum{\isacharunderscore}{\kern0pt}measure\ M\ N{\isacharparenright}{\kern0pt}\ M{\isachardoublequoteclose}
\isanewline
\isacommand{lemma}\isamarkupfalse%
\ NcontMN{\isacharbrackleft}{\kern0pt}simp{\isacharcomma}{\kern0pt}intro{\isacharbrackright}{\kern0pt}{\isacharcolon}{\kern0pt}
\isanewline\ \ {\isachardoublequoteopen}absolutely{\isacharunderscore}{\kern0pt}continuous\ {\isacharparenleft}{\kern0pt}sum{\isacharunderscore}{\kern0pt}measure\ M\ N{\isacharparenright}{\kern0pt}\ N{\isachardoublequoteclose}
\isanewline
\ \ ...
\isanewline
\isacommand{definition}\isamarkupfalse%
\ {\isachardoublequoteopen}dM\ {\isacharequal}{\kern0pt}\ real{\isacharunderscore}{\kern0pt}RN{\isacharunderscore}{\kern0pt}deriv\ {\isacharparenleft}{\kern0pt}sum{\isacharunderscore}{\kern0pt}measure\ M\ N{\isacharparenright}{\kern0pt}\ M{\isachardoublequoteclose}
\isanewline
\isacommand{definition}\isamarkupfalse%
\ {\isachardoublequoteopen}dN\ {\isacharequal}{\kern0pt}\ real{\isacharunderscore}{\kern0pt}RN{\isacharunderscore}{\kern0pt}deriv\ {\isacharparenleft}{\kern0pt}sum{\isacharunderscore}{\kern0pt}measure\ M\ N{\isacharparenright}{\kern0pt}\ N{\isachardoublequoteclose}
\isanewline
\ \ ...
\isanewline
\isacommand{lemma}\isamarkupfalse%
\ dM{\isacharunderscore}{\kern0pt}less{\isacharunderscore}{\kern0pt}{\isadigit{1}}{\isacharunderscore}{\kern0pt}AE{\isacharcolon}{\kern0pt}\isanewline
\ \ \isakeyword{shows}\ {\isachardoublequoteopen}AE\ x\ in\ {\isacharparenleft}{\kern0pt}sum{\isacharunderscore}{\kern0pt}measure\ M\ N{\isacharparenright}{\kern0pt}{\isachardot}{\kern0pt}\ dM\ x\ {\isasymle}\ {\isadigit{1}}{\isachardoublequoteclose}
\isanewline
\isacommand{lemma}\isamarkupfalse%
\ dN{\isacharunderscore}{\kern0pt}less{\isacharunderscore}{\kern0pt}{\isadigit{1}}{\isacharunderscore}{\kern0pt}AE{\isacharcolon}{\kern0pt}\isanewline
\ \ \isakeyword{shows}\ {\isachardoublequoteopen}AE\ x\ in\ {\isacharparenleft}{\kern0pt}sum{\isacharunderscore}{\kern0pt}measure\ M\ N{\isacharparenright}{\kern0pt}{\isachardot}{\kern0pt}\ dN\ x\ {\isasymle}\ {\isadigit{1}}{\isachardoublequoteclose}
\isanewline
%\ \ ...
%\isanewline
%\isacommand{lemma}\isamarkupfalse%
%\ dM{\isacharunderscore}{\kern0pt}dN{\isacharunderscore}{\kern0pt}partition{\isacharunderscore}{\kern0pt}{\isadigit{1}}{\isacharunderscore}{\kern0pt}AE{\isacharcolon}{\kern0pt}\isanewline
%\ \ \isakeyword{shows}\ {\isachardoublequoteopen}AE\ x\ in\ {\isacharparenleft}{\kern0pt}sum{\isacharunderscore}{\kern0pt}measure\ M\ N{\isacharparenright}{\kern0pt}{\isachardot}{\kern0pt}\ {\isacharparenleft}{\kern0pt}dM\ x\ {\isacharplus}{\kern0pt}\ dN\ x{\isacharparenright}{\kern0pt}\ {\isacharequal}{\kern0pt}\ {\isadigit{1}}{\isachardoublequoteclose}
%\isanewline
\isacommand{end}
\end{isabelle}
The assumptions \isa{{\isachardoublequoteopen}M\ {\isasymin}\ space\ {\isacharparenleft}{\kern0pt}prob{\isacharunderscore}{\kern0pt}algebra\ L{\isacharparenright}{\kern0pt}{\isachardoublequoteclose}} and \isa{{\isachardoublequoteopen}N\ {\isasymin}\ space\ {\isacharparenleft}{\kern0pt}prob{\isacharunderscore}{\kern0pt}algebra\ L{\isacharparenright}{\kern0pt}{\isachardoublequoteclose}}
in this locale corresponds to the mathematical condition $\mu,\nu \in \prob(X)$.
%Here we take the sum \isa{sum{\isacharunderscore}{\kern0pt}measure\ M\ N} of probability distribitions \isa{M} and \isa{N} for a ($\sigma$-finite)
%measure dominating \isa{M} and \isa{N} (see \isa{McontMN} and \isa{NcontMN}). 

We also give a constant \isa{real{\isacharunderscore}{\kern0pt}RN{\isacharunderscore}{\kern0pt}deriv} for the construction of real-valued Radon-Nikod\'{y}m derivatives 
instead of usual \isa{ennreal}-valued \isa{RN{\isacharunderscore}{\kern0pt}deriv}.
The existence is provided in the last of the theory \isa{Radon\_Nikodym} in the standard library of Isabelle/HOL.
The constant \isa{real{\isacharunderscore}{\kern0pt}RN{\isacharunderscore}{\kern0pt}deriv} can be defined in the same way as usual $\isa{RN\_deriv}$.
%\begin{isabelle}
%\isacommand{definition}\isamarkupfalse%
%\ sum{\isacharunderscore}{\kern0pt}measure\ \ {\isacharcolon}{\kern0pt}{\isacharcolon}{\kern0pt}\ {\isachardoublequoteopen}{\isacharprime}{\kern0pt}a\ measure\ {\isasymRightarrow}\ {\isacharprime}{\kern0pt}a\ measure\ {\isasymRightarrow}\ {\isacharprime}{\kern0pt}a\ measure{\isachardoublequoteclose}\ \isakeyword{where}\isanewline
%\ \ {\isachardoublequoteopen}sum{\isacharunderscore}{\kern0pt}measure\ M\ N\ {\isacharequal}{\kern0pt}\isanewline
%\ \ \ \ measure{\isacharunderscore}{\kern0pt}of\ {\isacharparenleft}{\kern0pt}space\ M{\isacharparenright}{\kern0pt}\ {\isacharparenleft}{\kern0pt}sets\ M{\isacharparenright}{\kern0pt}\ {\isacharparenleft}{\kern0pt}{\isasymlambda}\ A{\isachardot}{\kern0pt}\ emeasure\ M\ A\ {\isacharplus}{\kern0pt}\ emeasure\ N\ A{\isacharparenright}{\kern0pt}{\isachardoublequoteclose}
%\end{isabelle}

\subsubsection{Formalization of Properties of $\Delta^\varepsilon$}
We formalize the properties of $\Delta^\varepsilon$ as in Figure \ref{fig:basic:DP_divergence}.
\begin{figure*}
\begin{minipage}[t]{0.9\columnwidth}
\begin{isabelle}
\isacommand{lemma}\isamarkupfalse%
\ DP{\isacharunderscore}{\kern0pt}divergence{\isacharunderscore}{\kern0pt}nonnegativity{\isacharcolon}{\kern0pt}\isanewline
\ \ \isakeyword{shows}\ {\isachardoublequoteopen}{\isadigit{0}}\ {\isasymle}\ DP{\isacharunderscore}{\kern0pt}divergence\ M\ N\ {\isasymepsilon}{\isachardoublequoteclose}
\end{isabelle}

\begin{isabelle}
\isacommand{lemma}\isamarkupfalse%
\ DP{\isacharunderscore}{\kern0pt}divergence{\isacharunderscore}{\kern0pt}monotonicity{\isacharcolon}{\kern0pt}\isanewline
\ \ \isakeyword{assumes}\ M{\isacharcolon}{\kern0pt}\ {\isachardoublequoteopen}M\ {\isasymin}\ space\ {\isacharparenleft}{\kern0pt}prob{\isacharunderscore}{\kern0pt}algebra\ L{\isacharparenright}{\kern0pt}{\isachardoublequoteclose}\isanewline
\ \ \ \ \isakeyword{and}\ N{\isacharcolon}{\kern0pt}\ {\isachardoublequoteopen}N\ {\isasymin}\ space\ {\isacharparenleft}{\kern0pt}prob{\isacharunderscore}{\kern0pt}algebra\ L{\isacharparenright}{\kern0pt}{\isachardoublequoteclose}\ \isakeyword{and}\ {\isachardoublequoteopen}{\isasymepsilon}{\isadigit{1}}\ {\isasymle}\ {\isasymepsilon}{\isadigit{2}}{\isachardoublequoteclose}\isanewline
\ \ \isakeyword{shows}\ {\isachardoublequoteopen}DP{\isacharunderscore}{\kern0pt}divergence\ M\ N\ {\isasymepsilon}{\isadigit{2}}\ {\isasymle}\ DP{\isacharunderscore}{\kern0pt}divergence\ M\ N\ {\isasymepsilon}{\isadigit{1}}{\isachardoublequoteclose}
\end{isabelle}

\begin{isabelle}
\isacommand{lemma}\isamarkupfalse%
\ DP{\isacharunderscore}{\kern0pt}divergence{\isacharunderscore}{\kern0pt}transitivity{\isacharcolon}{\kern0pt}\isanewline
\ \ \isakeyword{assumes}\ M{\isadigit{1}}{\isacharcolon}{\kern0pt}\ {\isachardoublequoteopen}M{\isadigit{1}}\ {\isasymin}\ space\ {\isacharparenleft}{\kern0pt}prob{\isacharunderscore}{\kern0pt}algebra\ L{\isacharparenright}{\kern0pt}{\isachardoublequoteclose}\isanewline
\ \ \ \ \isakeyword{and}\ M{\isadigit{2}}{\isacharcolon}{\kern0pt}\ {\isachardoublequoteopen}M{\isadigit{2}}\ {\isasymin}\ space\ {\isacharparenleft}{\kern0pt}prob{\isacharunderscore}{\kern0pt}algebra\ L{\isacharparenright}{\kern0pt}{\isachardoublequoteclose}\isanewline
\ \ \ \ \isakeyword{and}\ DP{\isadigit{1}}{\isacharcolon}{\kern0pt}\ {\isachardoublequoteopen}DP{\isacharunderscore}{\kern0pt}divergence\ M{\isadigit{1}}\ M{\isadigit{2}}\ {\isasymepsilon}{\isadigit{1}}\ {\isasymle}\ {\isadigit{0}}{\isachardoublequoteclose}\isanewline
\ \ \ \ \isakeyword{and}\ DP{\isadigit{2}}{\isacharcolon}{\kern0pt}\ {\isachardoublequoteopen}DP{\isacharunderscore}{\kern0pt}divergence\ M{\isadigit{2}}\ M{\isadigit{3}}\ {\isasymepsilon}{\isadigit{2}}\ {\isasymle}\ {\isadigit{0}}{\isachardoublequoteclose}\isanewline
\ \ \isakeyword{shows}\ {\isachardoublequoteopen}DP{\isacharunderscore}{\kern0pt}divergence\ M{\isadigit{1}}\ M{\isadigit{3}}\ {\isacharparenleft}{\kern0pt}{\isasymepsilon}{\isadigit{1}}{\isacharplus}{\kern0pt}{\isasymepsilon}{\isadigit{2}}{\isacharparenright}{\kern0pt}\ {\isasymle}\ {\isadigit{0}}{\isachardoublequoteclose}
\end{isabelle}
\end{minipage}
\hspace{0.04\columnwidth}
 \begin{minipage}[t]{1.1\columnwidth}
 \begin{isabelle}
\isacommand{lemma}\isamarkupfalse%
\ DP{\isacharunderscore}{\kern0pt}divergence{\isacharunderscore}{\kern0pt}reflexivity{\isacharcolon}{\kern0pt}\isanewline
\ \ \isakeyword{shows}\ {\isachardoublequoteopen}DP{\isacharunderscore}{\kern0pt}divergence\ M\ M\ {\isadigit{0}}\ {\isacharequal}{\kern0pt}\ {\isadigit{0}}{\isachardoublequoteclose}
\end{isabelle}
\vspace{0.5em}

\begin{isabelle}
\isacommand{proposition}\isamarkupfalse%
\ DP{\isacharunderscore}{\kern0pt}divergence{\isacharunderscore}{\kern0pt}composability{\isacharcolon}{\kern0pt}\isanewline
\ \ \isakeyword{assumes}\ M{\isacharcolon}{\kern0pt}\ {\isachardoublequoteopen}M\ {\isasymin}\ space\ {\isacharparenleft}{\kern0pt}prob{\isacharunderscore}{\kern0pt}algebra\ L{\isacharparenright}{\kern0pt}{\isachardoublequoteclose}\isanewline
\ \ \ \ \isakeyword{and}\ N{\isacharcolon}{\kern0pt}\ {\isachardoublequoteopen}N\ {\isasymin}\ space\ {\isacharparenleft}{\kern0pt}prob{\isacharunderscore}{\kern0pt}algebra\ L{\isacharparenright}{\kern0pt}{\isachardoublequoteclose}\isanewline
\ \ \ \ \isakeyword{and}\ f{\isacharcolon}{\kern0pt}\ {\isachardoublequoteopen}f\ {\isasymin}\ L\ {\isasymrightarrow}\isactrlsub M\ prob{\isacharunderscore}{\kern0pt}algebra\ K{\isachardoublequoteclose}\isanewline
\ \ \ \ \isakeyword{and}\ g{\isacharcolon}{\kern0pt}\ {\isachardoublequoteopen}g\ {\isasymin}\ L\ {\isasymrightarrow}\isactrlsub M\ prob{\isacharunderscore}{\kern0pt}algebra\ K{\isachardoublequoteclose}\isanewline
\ \ \ \ \isakeyword{and}\ div{\isadigit{1}}{\isacharcolon}{\kern0pt}\ {\isachardoublequoteopen}DP{\isacharunderscore}{\kern0pt}divergence\ M\ N\ {\isasymepsilon}{\isadigit{1}}\ {\isasymle}\ {\isacharparenleft}{\kern0pt}{\isasymdelta}{\isadigit{1}}\ {\isacharcolon}{\kern0pt}{\isacharcolon}{\kern0pt}\ real{\isacharparenright}{\kern0pt}{\isachardoublequoteclose}\isanewline
\ \ \ \ \isakeyword{and}\ div{\isadigit{2}}{\isacharcolon}{\kern0pt}\ {\isachardoublequoteopen}{\isasymforall}x\ {\isasymin}\ {\isacharparenleft}{\kern0pt}space\ L{\isacharparenright}{\kern0pt}{\isachardot}{\kern0pt}\ DP{\isacharunderscore}{\kern0pt}divergence\ {\isacharparenleft}{\kern0pt}f\ x{\isacharparenright}{\kern0pt}\ {\isacharparenleft}{\kern0pt}g\ x{\isacharparenright}{\kern0pt}\ {\isasymepsilon}{\isadigit{2}}\ {\isasymle}\ {\isacharparenleft}{\kern0pt}{\isasymdelta}{\isadigit{2}}\ {\isacharcolon}{\kern0pt}{\isacharcolon}{\kern0pt}\ real{\isacharparenright}{\kern0pt}{\isachardoublequoteclose}\isanewline
\ \ \ \ \isakeyword{and}\ {\isachardoublequoteopen}{\isadigit{0}}\ {\isasymle}\ {\isasymepsilon}{\isadigit{1}}{\isachardoublequoteclose}\ \isakeyword{and}\ {\isachardoublequoteopen}{\isadigit{0}}\ {\isasymle}\ {\isasymepsilon}{\isadigit{2}}{\isachardoublequoteclose}\isanewline
\ \ \isakeyword{shows}\ {\isachardoublequoteopen}DP{\isacharunderscore}{\kern0pt}divergence\ {\isacharparenleft}{\kern0pt}M\ {\isasymbind}\ f{\isacharparenright}{\kern0pt}\ {\isacharparenleft}{\kern0pt}N\ {\isasymbind}\ g{\isacharparenright}{\kern0pt}\ {\isacharparenleft}{\kern0pt}{\isasymepsilon}{\isadigit{1}}\ {\isacharplus}{\kern0pt}\ {\isasymepsilon}{\isadigit{2}}{\isacharparenright}{\kern0pt}\ {\isasymle}\ {\isasymdelta}{\isadigit{1}}\ {\isacharplus}{\kern0pt}\ {\isasymdelta}{\isadigit{2}}{\isachardoublequoteclose}
\end{isabelle}
\end{minipage} 
\caption{Formalization of the nonnegativity, reflexivity, monotonicity, composability and ``transitivity'' of $\Delta^\varepsilon$.}
\label{fig:basic:DP_divergence}
\end{figure*}
To prove them, we often use the locale \isa{comparable{\isacharunderscore}{\kern0pt}probability{\isacharunderscore}{\kern0pt}measures}. 

It is easy to give formal proofs of the first four lemmas.
In the formal proof of the composability, we follow the simplified proof given in this section.
We use \isa{sum{\isacharunderscore}{\kern0pt}measure\ M\ N} for the base measure $\pi$.
We also use \isa{dM} and \isa{dN} for ${d\mu}/{d\pi}$ and ${d\nu}/{d\pi}$. 
The evaluations of integrals in the composability proof contain many subtractions, which nonnegative integrals do not support well.
We thus choose the Lebesgue integral (\isa{lebesgue\_integral})
instead of the nonnegative integral (\isa{nn\_integral}).
They impose many integrability proofs, but almost all of them are automated.
We remark that we have chosen the finite base measure \isa{sum{\isacharunderscore}{\kern0pt}measure\ M\ N}, and density functions of \isa{dM} and \isa{dN} of \isa{M} and \isa{N} which are (essentially) bounded by $1$.
Hence, all integrations in the proof are ones of (essentially) bounded functions under the finite measure, which are easy to automate.

\section{Differential Privacy in Isabelle/HOL}
In this section, we formalize differential privacy and its properties shown in Section \ref{sec:DifferentialPrivacy}.
Later, in Section \ref{sec:formal_RNM}, we formalize the differential privacy of the report noisy max mechanism.

For simplicity, we write $n \in \Nat$ for the number $|\mathcal{X}|$ of data types, and assume $\mathcal{X} = \{0,\ldots,n-1\}$.
Throughout our formalization, we implement a tuple $(x_0,\ldots,x_{k-1})$ of $k$ as a list with length $k$.
Our library contains the construction of measurable spaces of finite lists and the measurability of basic list operations,
but we omit their details because of space limitations.

\subsection{Differential Privacy for General Adjacency}
To formalize the notion of differential privacy, 
we consider a general domain $X$ of datasets and a general adjacency relation $R^\mathrm{adj}$.
Later, we instantiate $X = \Nat^{|\mathcal{X}|}$ and $R^\mathrm{adj} = \{ (D,D') ~|~ \|D - D' \| \leq 1\}$.

We formalize differential privacy \emph{with respect to} general $X$ and $R^\mathrm{adj}$.
First, we give the below Isabelle/HOL-term corresponding to the following condition ($\mu,\nu \in \prob (Y)$):
\[
\forall S \in \Sigma_Y.~ \Pr[\mu \in S] \leq \exp(\varepsilon) \Pr[\nu \in S] + \delta.
\]
\begin{isabelle}
\isacommand{definition}\isamarkupfalse%
\ DP{\isacharunderscore}{\kern0pt}inequality{\isacharcolon}{\kern0pt}{\isacharcolon}{\kern0pt}\ {\isachardoublequoteopen}{\isacharprime}{\kern0pt}a\ measure\ {\isasymRightarrow}\ {\isacharprime}{\kern0pt}a\ measure\ {\isasymRightarrow}\ real\ {\isasymRightarrow}\ real\ {\isasymRightarrow}\ bool{\isachardoublequoteclose}\ \isakeyword{where}\isanewline
\ \ {\isachardoublequoteopen}DP{\isacharunderscore}{\kern0pt}inequality\ M\ N\ {\isasymepsilon}\ {\isasymdelta}\ {\isasymequiv}\ {\isacharparenleft}{\kern0pt}{\isasymforall}\ A\ {\isasymin}\ sets\ M{\isachardot}{\kern0pt}\ measure\ M\ A\ {\isasymle}\ {\isacharparenleft}{\kern0pt}exp\ {\isasymepsilon}{\isacharparenright}{\kern0pt}\ {\isacharasterisk}{\kern0pt}\ measure\ N\ A\ {\isacharplus}{\kern0pt}\ {\isasymdelta}{\isacharparenright}{\kern0pt}{\isachardoublequoteclose}
\end{isabelle}

Then, we give a formal definition of differential privacy:
\begin{isabelle}
\isacommand{definition}\isamarkupfalse%
\ differential{\isacharunderscore}{\kern0pt}privacy\ {\isacharcolon}{\kern0pt}{\isacharcolon}{\kern0pt}\ {\isachardoublequoteopen}{\isacharparenleft}{\kern0pt}{\isacharprime}{\kern0pt}a\ {\isasymRightarrow}\ {\isacharprime}{\kern0pt}b\ measure{\isacharparenright}{\kern0pt}\ {\isasymRightarrow}\ {\isacharparenleft}{\kern0pt}{\isacharprime}{\kern0pt}a\ rel{\isacharparenright}{\kern0pt}\ {\isasymRightarrow}\ real\ {\isasymRightarrow}\ real\ {\isasymRightarrow}\ bool\ {\isachardoublequoteclose}\isakeyword{where}\isanewline
\ \ {\isachardoublequoteopen}differential{\isacharunderscore}{\kern0pt}privacy\ M\ adj\ {\isasymepsilon}\ {\isasymdelta}\ {\isasymequiv}\ {\isasymforall}{\isacharparenleft}{\kern0pt}d{\isadigit{1}}{\isacharcomma}{\kern0pt}d{\isadigit{2}}{\isacharparenright}{\kern0pt}{\isasymin}adj{\isachardot}{\kern0pt}\ DP{\isacharunderscore}{\kern0pt}inequality\ {\isacharparenleft}{\kern0pt}M\ d{\isadigit{1}}{\isacharparenright}{\kern0pt}\ {\isacharparenleft}{\kern0pt}M\ d{\isadigit{2}}{\isacharparenright}{\kern0pt}\ {\isasymepsilon}\ {\isasymdelta}\ {\isasymand}\ DP{\isacharunderscore}{\kern0pt}inequality\ {\isacharparenleft}{\kern0pt}M\ d{\isadigit{2}}{\isacharparenright}{\kern0pt}\ {\isacharparenleft}{\kern0pt}M\ d{\isadigit{1}}{\isacharparenright}{\kern0pt}\ {\isasymepsilon}\ {\isasymdelta}{\isachardoublequoteclose}
\end{isabelle}
Here \isa{M} is a randomized algorithm, and \isa{adj} is for the adjacency relation $R^{\mathrm{adj}}$.
Since \isa{adj} may not be symmetric\footnote{A binary relation $R$ is symmetric if $R = R^{-1}$.
We later apply an asymmetric \isa{adj} in the formalization of report noisy max mechanism.}, both inequalities are needed in the formal definition.
If \isa{adj} is symmetric, it corresponds to the original definition (Def. \ref{def:DP}).
%
%\begin{isabelle}
%\isacommand{lemma}\isamarkupfalse%
%\ differential{\isacharunderscore}{\kern0pt}privacy{\isacharunderscore}{\kern0pt}adj{\isacharunderscore}{\kern0pt}sym{\isacharcolon}{\kern0pt}\isanewline
%\ \ \isakeyword{assumes}\ {\isachardoublequoteopen}sym\ adj{\isachardoublequoteclose}\isanewline
%\ \ \isakeyword{shows}\ {\isachardoublequoteopen}differential{\isacharunderscore}{\kern0pt}privacy\ M\ adj\ {\isasymepsilon}\ {\isasymdelta}\ {\isasymlongleftrightarrow}\ {\isacharparenleft}{\kern0pt}{\isasymforall}\ {\isacharparenleft}{\kern0pt}d{\isadigit{1}}{\isacharcomma}{\kern0pt}d{\isadigit{2}}{\isacharparenright}{\kern0pt}\ {\isasymin}\ adj{\isachardot}{\kern0pt}\ DP{\isacharunderscore}{\kern0pt}inequality\ {\isacharparenleft}{\kern0pt}M\ d{\isadigit{1}}{\isacharparenright}{\kern0pt}\ {\isacharparenleft}{\kern0pt}M\ d{\isadigit{2}}{\isacharparenright}{\kern0pt}\ {\isasymepsilon}\ {\isasymdelta}{\isacharparenright}{\kern0pt}{\isachardoublequoteclose}
%\end{isabelle}

We give the formal proof of basic properties of differential privacy using the formalization \isa{DP\_Divergence} of $\Delta^\varepsilon$.
For example, 
the postprocessing property and
the composition theorem (Lemmas \ref{DP:basic:postprocessing} and \ref{DP:basic:sequential}) are formalized as follows:
\begin{isabelle}
\isacommand{proposition}\isamarkupfalse%
\ differential{\isacharunderscore}{\kern0pt}privacy{\isacharunderscore}{\kern0pt}postprocessing{\isacharcolon}{\kern0pt}\isanewline
\ \ \isakeyword{assumes}\ {\isachardoublequoteopen}{\isasymepsilon}\ {\isasymge}\ {\isadigit{0}}{\isachardoublequoteclose}\ \isakeyword{and}\ {\isachardoublequoteopen}{\isasymdelta}\ {\isasymge}\ {\isadigit{0}}{\isachardoublequoteclose}\isanewline
\ \ \ \ \isakeyword{and}\ {\isachardoublequoteopen}differential{\isacharunderscore}{\kern0pt}privacy\ M\ adj\ {\isasymepsilon}\ {\isasymdelta}{\isachardoublequoteclose}\isanewline
\ \ \ \ \isakeyword{and}\ M{\isacharcolon}{\kern0pt}\ {\isachardoublequoteopen}M\ {\isasymin}\ X\ {\isasymrightarrow}\isactrlsub M\ prob{\isacharunderscore}{\kern0pt}algebra\ R{\isachardoublequoteclose}\isanewline
\ \ \ \ \isakeyword{and}\ f{\isacharcolon}{\kern0pt}\ {\isachardoublequoteopen}f\ {\isasymin}\ R\ {\isasymrightarrow}\isactrlsub M\ prob{\isacharunderscore}{\kern0pt}algebra\ R{\isacharprime}{\kern0pt}{\isachardoublequoteclose}\ \isanewline
\ \ \ \ \isakeyword{and}\ {\isachardoublequoteopen}adj\ {\isasymsubseteq}\ space\ X\ {\isasymtimes}\ space\ X{\isachardoublequoteclose}\isanewline
\ \ \isakeyword{shows}\ {\isachardoublequoteopen}differential{\isacharunderscore}{\kern0pt}privacy\ {\isacharparenleft}{\kern0pt}{\isasymlambda}x{\isachardot}{\kern0pt}\ do{\isacharbraceleft}{\kern0pt}y\ {\isasymleftarrow}\ M\ x{\isacharsemicolon}{\kern0pt}\ f\ y{\isacharbraceright}{\kern0pt}{\isacharparenright}{\kern0pt}\ adj\ {\isasymepsilon}\ {\isasymdelta}{\isachardoublequoteclose}
\end{isabelle}
\begin{isabelle}
\isacommand{proposition}\isamarkupfalse%
\ differential{\isacharunderscore}{\kern0pt}privacy{\isacharunderscore}{\kern0pt}composition{\isacharunderscore}{\kern0pt}pair{\isacharcolon}{\kern0pt}\isanewline
\ \ \isakeyword{assumes}\ {\isachardoublequoteopen}{\isasymepsilon}\ {\isasymge}\ {\isadigit{0}}{\isachardoublequoteclose}\ \isakeyword{and}\ {\isachardoublequoteopen}{\isasymdelta}\ {\isasymge}\ {\isadigit{0}}{\isachardoublequoteclose}\isanewline
\ \ \ \ \isakeyword{and}\ {\isachardoublequoteopen}{\isasymepsilon}{\isacharprime}{\kern0pt}\ {\isasymge}\ {\isadigit{0}}{\isachardoublequoteclose}\ \isakeyword{and}\ {\isachardoublequoteopen}{\isasymdelta}{\isacharprime}{\kern0pt}\ {\isasymge}\ {\isadigit{0}}{\isachardoublequoteclose}\isanewline
\ \ \ \ \isakeyword{and}\ DPM{\isacharcolon}{\kern0pt}\ {\isachardoublequoteopen}differential{\isacharunderscore}{\kern0pt}privacy\ M\ adj\ {\isasymepsilon}\ {\isasymdelta}{\isachardoublequoteclose}\isanewline
\ \ \ \ \isakeyword{and}\ M{\isacharbrackleft}{\kern0pt}measurable{\isacharbrackright}{\kern0pt}{\isacharcolon}{\kern0pt}\ {\isachardoublequoteopen}M\ {\isasymin}\ X\ {\isasymrightarrow}\isactrlsub M\ prob{\isacharunderscore}{\kern0pt}algebra\ Y{\isachardoublequoteclose}\isanewline
\ \ \ \ \isakeyword{and}\ DPN{\isacharcolon}{\kern0pt}\ {\isachardoublequoteopen}differential{\isacharunderscore}{\kern0pt}privacy\ N\ adj\ {\isasymepsilon}{\isacharprime}{\kern0pt}\ {\isasymdelta}{\isacharprime}{\kern0pt}{\isachardoublequoteclose}\isanewline
\ \ \ \ \isakeyword{and}\ N{\isacharbrackleft}{\kern0pt}measurable{\isacharbrackright}{\kern0pt}{\isacharcolon}{\kern0pt}\ {\isachardoublequoteopen}N\ {\isasymin}\ X\ {\isasymrightarrow}\isactrlsub M\ prob{\isacharunderscore}{\kern0pt}algebra\ Z{\isachardoublequoteclose}\isanewline
\ \ \ \ \isakeyword{and}\ {\isachardoublequoteopen}adj\ {\isasymsubseteq}\ space\ X\ {\isasymtimes}\ space\ X{\isachardoublequoteclose}\isanewline
\ \ \isakeyword{shows}\ {\isachardoublequoteopen}differential{\isacharunderscore}{\kern0pt}privacy\ {\isacharparenleft}{\kern0pt}{\isasymlambda}x{\isachardot}{\kern0pt}\ do{\isacharbraceleft}{\kern0pt}y\ {\isasymleftarrow}\ M\ x{\isacharsemicolon}{\kern0pt}\ z\ {\isasymleftarrow}\ N\ x{\isacharsemicolon}{\kern0pt}\ return\ {\isacharparenleft}{\kern0pt}Y\ {\isasymOtimes}\isactrlsub M\ Z{\isacharparenright}{\kern0pt}\ {\isacharparenleft}{\kern0pt}y{\isacharcomma}{\kern0pt}z{\isacharparenright}{\kern0pt}{\isacharbraceright}{\kern0pt}\ {\isacharparenright}{\kern0pt}\ adj\ {\isacharparenleft}{\kern0pt}{\isasymepsilon}\ {\isacharplus}{\kern0pt}\ {\isasymepsilon}{\isacharprime}{\kern0pt}{\isacharparenright}{\kern0pt}\ {\isacharparenleft}{\kern0pt}{\isasymdelta}\ {\isacharplus}{\kern0pt}\ {\isasymdelta}{\isacharprime}{\kern0pt}{\isacharparenright}{\kern0pt}{\isachardoublequoteclose}
\end{isabelle}

Borrowing ideas from relational program logics reasoning about differential privacy, 
we also give the following \emph{preprocessing} lemma which transfers DP with respect to some $X$ and $R^\mathrm{adj}$ to DP with respect to another $X'$ and $R^\mathrm{adj'}$.
It is convenient to formalize Laplace mechanism, and to split the formalization of report noisy max mechanism into the main part and counting query.

\begin{isabelle}
\isacommand{lemma}\isamarkupfalse%
\ differential{\isacharunderscore}{\kern0pt}privacy{\isacharunderscore}{\kern0pt}preprocessing{\isacharcolon}{\kern0pt}\isanewline
\ \ \isakeyword{assumes}\ {\isachardoublequoteopen}{\isasymepsilon}\ {\isasymge}\ {\isadigit{0}}{\isachardoublequoteclose}\ \isakeyword{and}\ {\isachardoublequoteopen}{\isasymdelta}\ {\isasymge}\ {\isadigit{0}}{\isachardoublequoteclose}\isanewline
\ \ \ \ \isakeyword{and}\ {\isachardoublequoteopen}differential{\isacharunderscore}{\kern0pt}privacy\ M\ adj\ {\isasymepsilon}\ {\isasymdelta}{\isachardoublequoteclose}\isanewline
\ \ \ \ \isakeyword{and}\ f{\isacharcolon}{\kern0pt}\ {\isachardoublequoteopen}f\ {\isasymin}\ X{\isacharprime}{\kern0pt}\ {\isasymrightarrow}\isactrlsub M\ X{\isachardoublequoteclose}\ \isanewline
\ \ \ \ \isakeyword{and}\ ftr{\isacharcolon}{\kern0pt}\ {\isachardoublequoteopen}{\isasymforall}{\isacharparenleft}{\kern0pt}x{\isacharcomma}{\kern0pt}y{\isacharparenright}{\kern0pt}\ {\isasymin}\ adj{\isacharprime}{\kern0pt}{\isachardot}{\kern0pt}\ {\isacharparenleft}{\kern0pt}f\ x{\isacharcomma}{\kern0pt}\ f\ y{\isacharparenright}{\kern0pt}\ {\isasymin}\ adj{\isachardoublequoteclose}\isanewline
\ \ \ \ \isakeyword{and}\ {\isachardoublequoteopen}adj\ {\isasymsubseteq}\ space\ X\ {\isasymtimes}\ space\ X{\isachardoublequoteclose}\isanewline
\ \ \ \ \isakeyword{and}\ {\isachardoublequoteopen}adj{\isacharprime}{\kern0pt}\ {\isasymsubseteq}\ space\ X{\isacharprime}{\kern0pt}\ {\isasymtimes}\ space\ X{\isacharprime}{\kern0pt}{\isachardoublequoteclose}\isanewline
\ \ \isakeyword{shows}\ {\isachardoublequoteopen}differential{\isacharunderscore}{\kern0pt}privacy\ {\isacharparenleft}{\kern0pt}M\ o\ f{\isacharparenright}{\kern0pt}\ adj{\isacharprime}{\kern0pt}\ {\isasymepsilon}\ {\isasymdelta}{\isachardoublequoteclose}
\end{isabelle}

\subsection{Laplace Mechanism}

We next formalize the Laplace mechanism.

\subsubsection{Laplace Distribution}

To formalize the Laplace mechanism, we first implement the density function and cumulative distribution function of the Laplace distribution.

\begin{isabelle}
\isacommand{definition}\isamarkupfalse%
\ laplace{\isacharunderscore}{\kern0pt}density\ {\isacharcolon}{\kern0pt}{\isacharcolon}{\kern0pt}\ {\isachardoublequoteopen}real\ {\isasymRightarrow}\ real\ {\isasymRightarrow}\ real\ {\isasymRightarrow}\ real{\isachardoublequoteclose}\ \isakeyword{where}\isanewline
\ \ {\isachardoublequoteopen}laplace{\isacharunderscore}{\kern0pt}density\ l\ m\ x\ {\isacharequal}{\kern0pt}\ {\isacharparenleft}{\kern0pt}if\ l\ {\isachargreater}{\kern0pt}\ {\isadigit{0}}\ then\ exp{\isacharparenleft}{\kern0pt}{\isacharminus}{\kern0pt}{\isasymbar}x\ {\isacharminus}{\kern0pt}\ m{\isasymbar}\ {\isacharslash}{\kern0pt}\ l{\isacharparenright}{\kern0pt}\ {\isacharslash}{\kern0pt}\ {\isacharparenleft}{\kern0pt}{\isadigit{2}}\ {\isacharasterisk}{\kern0pt}\ l{\isacharparenright}{\kern0pt}\ else\ {\isadigit{0}}{\isacharparenright}{\kern0pt}{\isachardoublequoteclose}
\end{isabelle}
\begin{isabelle}
\isacommand{definition}\isamarkupfalse%
\ laplace{\isacharunderscore}{\kern0pt}CDF\ {\isacharcolon}{\kern0pt}{\isacharcolon}{\kern0pt}\ {\isachardoublequoteopen}real\ {\isasymRightarrow}\ real\ {\isasymRightarrow}\ real\ {\isasymRightarrow}\ real{\isachardoublequoteclose}\ \isakeyword{where}\isanewline
\ \ {\isachardoublequoteopen}laplace{\isacharunderscore}{\kern0pt}CDF\ l\ m\ x\ {\isacharequal}{\kern0pt}\ {\isacharparenleft}{\kern0pt}if\ {\isadigit{0}}\ {\isacharless}{\kern0pt}\ l\ then\ if\ x\ {\isacharless}{\kern0pt}\ m\ then\ exp\ {\isacharparenleft}{\kern0pt}{\isacharparenleft}{\kern0pt}x\ {\isacharminus}{\kern0pt}\ m{\isacharparenright}{\kern0pt}\ {\isacharslash}{\kern0pt}\ l{\isacharparenright}{\kern0pt}\ {\isacharslash}{\kern0pt}\ {\isadigit{2}}\ else\ {\isadigit{1}}\ {\isacharminus}{\kern0pt}\ exp\ {\isacharparenleft}{\kern0pt}{\isacharminus}{\kern0pt}\ {\isacharparenleft}{\kern0pt}x\ {\isacharminus}{\kern0pt}\ m{\isacharparenright}{\kern0pt}\ {\isacharslash}{\kern0pt}\ l{\isacharparenright}{\kern0pt}\ {\isacharslash}{\kern0pt}\ {\isadigit{2}}\ else\ {\isadigit{0}}{\isacharparenright}{\kern0pt}{\isachardoublequoteclose}
\end{isabelle}
Then, \isa{density\ lborel\ {\isacharparenleft}{\kern0pt}laplace{\isacharunderscore}{\kern0pt}density\ b\ z{\isacharparenright}{\kern0pt}{\isacharparenright}}
is an implementation of the Laplace distribution $\mathrm{Lap}(b,z)$ in Isabelle/HOL.

We then formalize the properties of the Laplace distribution. 
For example, the lemma below shows that \isa{laplace{\isacharunderscore}{\kern0pt}CDF\ l\ m} is actually the cumulative distribution function.
\begin{isabelle}
\isacommand{lemma}\isamarkupfalse%
\ emeasure{\isacharunderscore}{\kern0pt}laplace{\isacharunderscore}{\kern0pt}density{\isacharcolon}{\kern0pt}\isanewline
\ \ \isakeyword{assumes}\ {\isachardoublequoteopen}{\isadigit{0}}\ {\isacharless}{\kern0pt}\ l{\isachardoublequoteclose}\isanewline
\ \ \isakeyword{shows}\ {\isachardoublequoteopen}emeasure\ {\isacharparenleft}{\kern0pt}density\ lborel\ {\isacharparenleft}{\kern0pt}laplace{\isacharunderscore}{\kern0pt}density\ l\ m{\isacharparenright}{\kern0pt}{\isacharparenright}{\kern0pt}\ {\isacharbraceleft}{\kern0pt}{\isachardot}{\kern0pt}{\isachardot}{\kern0pt}\ a{\isacharbraceright}{\kern0pt}\ {\isacharequal}{\kern0pt}\ laplace{\isacharunderscore}{\kern0pt}CDF\ l\ m\ a{\isachardoublequoteclose}
\end{isabelle}

We used the formalization of Gaussian distribution in the standard library of Isabelle/HOL as a reference, and 
we also applied the fundamental theorem of calculus for improper integrals in the standard library.
Our current formalization is rather straightforward. 
It might be shortened by applying lemmas in the AFP entry \isa{Laplace Transform}~\cite{Laplace_Transform-AFP}.

\subsubsection{Single Laplace Noise}
Next, we formalize the mappings $(b,m) \mapsto \mathrm{Lap}(b,m)$ and $b \mapsto \mathrm{Lap}(b)$ ($ = \mathrm{Lap}(b,0)$).
We suppose $\mathrm{Lap}(b,m) = \return 0$ for $b \leq 0$.
We give formal proofs of their measurability (for fixed \isa{b}) and Lemma \ref{lem:Lap:shifting}.

\begin{isabelle}
\isacommand{definition}\isamarkupfalse%
\ Lap{\isacharunderscore}{\kern0pt}dist\ {\isacharcolon}{\kern0pt}{\isacharcolon}{\kern0pt}\ {\isachardoublequoteopen}real\ {\isasymRightarrow}\ real\ {\isasymRightarrow}\ real\ measure{\isachardoublequoteclose}\ \isakeyword{where}\isanewline
\ \ {\isachardoublequoteopen}Lap{\isacharunderscore}{\kern0pt}dist\ b\ {\isasymmu}\ {\isacharequal}{\kern0pt}\ {\isacharparenleft}{\kern0pt}if\ b\ {\isasymle}\ {\isadigit{0}}\ then\ return\ borel\ {\isasymmu}\ else\ density\ lborel\ {\isacharparenleft}{\kern0pt}laplace{\isacharunderscore}{\kern0pt}density\ b\ {\isasymmu}{\isacharparenright}{\kern0pt}{\isacharparenright}{\kern0pt}{\isachardoublequoteclose}
\end{isabelle}
\begin{isabelle}
\isacommand{lemma}\isamarkupfalse%
\ measurable{\isacharunderscore}{\kern0pt}Lap{\isacharunderscore}{\kern0pt}dist{\isacharbrackleft}{\kern0pt}measurable{\isacharbrackright}{\kern0pt}{\isacharcolon}{\kern0pt}\isanewline
\ \ \isakeyword{shows}\ {\isachardoublequoteopen}Lap{\isacharunderscore}{\kern0pt}dist\ b\ {\isasymin}\ borel\ {\isasymrightarrow}\isactrlsub M\ prob{\isacharunderscore}{\kern0pt}algebra\ borel{\isachardoublequoteclose}
\end{isabelle}
\begin{isabelle}
\isacommand{definition}\isamarkupfalse%
\ {\isachardoublequoteopen}Lap{\isacharunderscore}{\kern0pt}dist{\isadigit{0}}\ b\ {\isasymequiv}\ Lap{\isacharunderscore}{\kern0pt}dist\ b\ {\isadigit{0}}\ {\isachardoublequoteclose}
\end{isabelle}
\begin{isabelle}
\isacommand{lemma}\isamarkupfalse%
\ Lap{\isacharunderscore}{\kern0pt}dist{\isacharunderscore}{\kern0pt}def{\isadigit{2}}{\isacharcolon}{\kern0pt}\isanewline
\ \ \isakeyword{shows}\ {\isachardoublequoteopen}Lap{\isacharunderscore}{\kern0pt}dist\ b\ x\ {\isacharequal}{\kern0pt}\ do{\isacharbraceleft}{\kern0pt}r\ {\isasymleftarrow}\ Lap{\isacharunderscore}{\kern0pt}dist{\isadigit{0}}\ b{\isacharsemicolon}{\kern0pt}\ return\ borel\ {\isacharparenleft}{\kern0pt}x\ {\isacharplus}{\kern0pt}\ r{\isacharparenright}{\kern0pt}{\isacharbraceright}{\kern0pt}{\isachardoublequoteclose}
\end{isabelle}
%\begin{isabelle}
%\isacommand{corollary}\isamarkupfalse%
%\ Lap{\isacharunderscore}{\kern0pt}dist{\isacharunderscore}{\kern0pt}shift{\isacharcolon}{\kern0pt}\isanewline
%\ \ \isakeyword{shows}\ {\isachardoublequoteopen}Lap{\isacharunderscore}{\kern0pt}dist\ b\ {\isacharparenleft}{\kern0pt}x\ {\isacharplus}{\kern0pt}\ y{\isacharparenright}{\kern0pt}\ {\isacharequal}{\kern0pt}\ do{\isacharbraceleft}{\kern0pt}r\ {\isasymleftarrow}\ Lap{\isacharunderscore}{\kern0pt}dist\ b\ x{\isacharsemicolon}{\kern0pt}\ return\ borel\ {\isacharparenleft}{\kern0pt}y\ {\isacharplus}{\kern0pt}\ r{\isacharparenright}{\kern0pt}{\isacharbraceright}{\kern0pt}{\isachardoublequoteclose}
%\end{isabelle}

We now formalize Lemma \ref{lem:Lap:DP:finer} in the divergence form.
%The formal proof follows from the pencil and paper proof. 
\begin{isabelle}
\isacommand{proposition}\isamarkupfalse%
\ DP{\isacharunderscore}{\kern0pt}divergence{\isacharunderscore}{\kern0pt}Lap{\isacharunderscore}{\kern0pt}dist{\isacharprime}{\kern0pt}{\isacharcolon}{\kern0pt}\isanewline
\ \ \isakeyword{assumes}\ {\isachardoublequoteopen}b\ {\isachargreater}{\kern0pt}\ {\isadigit{0}}{\isachardoublequoteclose}\ \isakeyword{and}\ {\isachardoublequoteopen}{\isasymbar}\ x\ {\isacharminus}{\kern0pt}\ y\ {\isasymbar}\ {\isasymle}\ r{\isachardoublequoteclose}\isanewline
\ \ \isakeyword{shows}\ {\isachardoublequoteopen}DP{\isacharunderscore}{\kern0pt}divergence\ {\isacharparenleft}{\kern0pt}Lap{\isacharunderscore}{\kern0pt}dist\ b\ x{\isacharparenright}{\kern0pt}\ {\isacharparenleft}{\kern0pt}Lap{\isacharunderscore}{\kern0pt}dist\ b\ y{\isacharparenright}{\kern0pt}\ {\isacharparenleft}{\kern0pt}r\ {\isacharslash}{\kern0pt}\ b{\isacharparenright}{\kern0pt}\ {\isasymle}\ {\isacharparenleft}{\kern0pt}{\isadigit{0}}\ {\isacharcolon}{\kern0pt}{\isacharcolon}{\kern0pt}\ real{\isacharparenright}{\kern0pt}{\isachardoublequoteclose}
\end{isabelle}

\subsubsection{Bundled Laplace Noise}\label{sec:Lap_dist_list}
Next, we implement $\mathrm{Lap}^m (b)$ and $\mathrm{Lap}^m(b,\vec{x})$ for a fixed $b$ as the following premitive recursive functions:
\begin{isabelle}
\isacommand{context}\isamarkupfalse%
\isanewline
\ \ \isakeyword{fixes}\ b{\isacharcolon}{\kern0pt}{\isacharcolon}{\kern0pt}real\ \isanewline
\isakeyword{begin}
\end{isabelle}
\begin{isabelle}
\isacommand{primrec}\isamarkupfalse%
\ Lap{\isacharunderscore}{\kern0pt}dist{\isadigit{0}}{\isacharunderscore}{\kern0pt}list\ {\isacharcolon}{\kern0pt}{\isacharcolon}{\kern0pt}\ {\isachardoublequoteopen}nat\ {\isasymRightarrow}\ {\isacharparenleft}{\kern0pt}real\ list{\isacharparenright}{\kern0pt}\ measure{\isachardoublequoteclose}\ \isakeyword{where}\isanewline
\ \ {\isachardoublequoteopen}Lap{\isacharunderscore}{\kern0pt}dist{\isadigit{0}}{\isacharunderscore}{\kern0pt}list\ {\isadigit{0}}\ {\isacharequal}{\kern0pt}\ return\ {\isacharparenleft}{\kern0pt}listM\ borel{\isacharparenright}{\kern0pt}\ {\isacharbrackleft}{\kern0pt}{\isacharbrackright}{\kern0pt}{\isachardoublequoteclose}\ {\isacharbar}{\kern0pt}\isanewline
\ \ {\isachardoublequoteopen}Lap{\isacharunderscore}{\kern0pt}dist{\isadigit{0}}{\isacharunderscore}{\kern0pt}list\ {\isacharparenleft}{\kern0pt}Suc\ n{\isacharparenright}{\kern0pt}\ {\isacharequal}{\kern0pt}\ do{\isacharbraceleft}{\kern0pt}x{\isadigit{1}}\ {\isasymleftarrow}\ {\isacharparenleft}{\kern0pt}Lap{\isacharunderscore}{\kern0pt}dist{\isadigit{0}}\ b{\isacharparenright}{\kern0pt}{\isacharsemicolon}{\kern0pt}\ x{\isadigit{2}}\ {\isasymleftarrow}\ {\isacharparenleft}{\kern0pt}Lap{\isacharunderscore}{\kern0pt}dist{\isadigit{0}}{\isacharunderscore}{\kern0pt}list\ n{\isacharparenright}{\kern0pt}{\isacharsemicolon}{\kern0pt}\ return\ {\isacharparenleft}{\kern0pt}listM\ borel{\isacharparenright}{\kern0pt}\ {\isacharparenleft}{\kern0pt}x{\isadigit{1}}\ {\isacharhash}{\kern0pt}\ x{\isadigit{2}}{\isacharparenright}{\kern0pt}{\isacharbraceright}{\kern0pt}{\isachardoublequoteclose}%
\end{isabelle}
%\tetsuya{write explicitly}
%Then we implement a procedure adding the noise sampled from $\mathrm{Lap}^m (b)$ to a given $(x_0,\ldots,x_{m-1})$.
%We implement it as the following premitive recursive function, and then show that it is equal to what we want.
\begin{isabelle}
\isacommand{primrec}\isamarkupfalse%
\ Lap{\isacharunderscore}{\kern0pt}dist{\isacharunderscore}{\kern0pt}list\ {\isacharcolon}{\kern0pt}{\isacharcolon}{\kern0pt}\ {\isachardoublequoteopen}real\ list\ {\isasymRightarrow}\ {\isacharparenleft}{\kern0pt}real\ list{\isacharparenright}{\kern0pt}\ measure{\isachardoublequoteclose}\ \isakeyword{where}\isanewline
\ \ {\isachardoublequoteopen}Lap{\isacharunderscore}{\kern0pt}dist{\isacharunderscore}{\kern0pt}list\ {\isacharbrackleft}{\kern0pt}{\isacharbrackright}{\kern0pt}\ {\isacharequal}{\kern0pt}\ return\ {\isacharparenleft}{\kern0pt}listM\ borel{\isacharparenright}{\kern0pt}\ {\isacharbrackleft}{\kern0pt}{\isacharbrackright}{\kern0pt}{\isachardoublequoteclose}{\isacharbar}{\kern0pt}\isanewline
\ \ {\isachardoublequoteopen}Lap{\isacharunderscore}{\kern0pt}dist{\isacharunderscore}{\kern0pt}list\ {\isacharparenleft}{\kern0pt}x\ {\isacharhash}{\kern0pt}\ xs{\isacharparenright}{\kern0pt}\ {\isacharequal}{\kern0pt}\ do{\isacharbraceleft}{\kern0pt}x{\isadigit{1}}\ {\isasymleftarrow}\ {\isacharparenleft}{\kern0pt}Lap{\isacharunderscore}{\kern0pt}dist\ b\ x{\isacharparenright}{\kern0pt}{\isacharsemicolon}{\kern0pt}\ x{\isadigit{2}}\ {\isasymleftarrow}\ {\isacharparenleft}{\kern0pt}Lap{\isacharunderscore}{\kern0pt}dist{\isacharunderscore}{\kern0pt}list\ xs{\isacharparenright}{\kern0pt}{\isacharsemicolon}{\kern0pt}\ return\ {\isacharparenleft}{\kern0pt}listM\ borel{\isacharparenright}{\kern0pt}\ {\isacharparenleft}{\kern0pt}x{\isadigit{1}}\ {\isacharhash}{\kern0pt}\ x{\isadigit{2}}{\isacharparenright}{\kern0pt}{\isacharbraceright}{\kern0pt}{\isachardoublequoteclose}
\end{isabelle}
We formalize the equation (\ref{eq:LapMech_def3}) in Section \ref{sec:laplace_mechanism_sensitivity}. 
\begin{isabelle}
\isacommand{lemma}\isamarkupfalse%
\ Lap{\isacharunderscore}{\kern0pt}dist{\isacharunderscore}{\kern0pt}list{\isacharunderscore}{\kern0pt}def{\isadigit{2}}{\isacharcolon}{\kern0pt}\isanewline
\ \ \isakeyword{shows}\ {\isachardoublequoteopen}Lap{\isacharunderscore}{\kern0pt}dist{\isacharunderscore}{\kern0pt}list\ xs\ {\isacharequal}{\kern0pt}\ do{\isacharbraceleft}{\kern0pt}ys\ {\isasymleftarrow}\ {\isacharparenleft}{\kern0pt}Lap{\isacharunderscore}{\kern0pt}dist{\isadigit{0}}{\isacharunderscore}{\kern0pt}list\ {\isacharparenleft}{\kern0pt}length\ xs{\isacharparenright}{\kern0pt}{\isacharparenright}{\kern0pt}{\isacharsemicolon}{\kern0pt}\ return\ {\isacharparenleft}{\kern0pt}listM\ borel{\isacharparenright}{\kern0pt}\ {\isacharparenleft}{\kern0pt}map{\isadigit{2}}\ {\isacharparenleft}{\kern0pt}{\isacharplus}{\kern0pt}{\isacharparenright}{\kern0pt}\ xs\ ys{\isacharparenright}{\kern0pt}{\isacharbraceright}{\kern0pt}{\isachardoublequoteclose}
\end{isabelle}
In the proof of the measurability of \isa{Lap{\isacharunderscore}{\kern0pt}dist{\isacharunderscore}{\kern0pt}list}, we use our library of measurable spaces of lists (Section \ref{sec:list_space}), and apply the measurability of \isa{rec\_list}.

By applying \isa{DP{\isacharunderscore}{\kern0pt}divergence{\isacharunderscore}{\kern0pt}Lap{\isacharunderscore}{\kern0pt}dist{\isacharprime}} repeatedly, we obtain the following lemma, an essential part of Lemma \ref{prop.DP.LapMech}:
\begin{isabelle}
\isacommand{lemma}\isamarkupfalse%
\ DP{\isacharunderscore}{\kern0pt}Lap{\isacharunderscore}{\kern0pt}dist{\isacharunderscore}{\kern0pt}list{\isacharcolon}{\kern0pt}\isanewline
\ \ \isakeyword{fixes}\ xs\ ys\ {\isacharcolon}{\kern0pt}{\isacharcolon}{\kern0pt}\ {\isachardoublequoteopen}real\ list{\isachardoublequoteclose}\ \isakeyword{and}\ n\ {\isacharcolon}{\kern0pt}{\isacharcolon}{\kern0pt}\ nat\ \isakeyword{and}\ r\ {\isacharcolon}{\kern0pt}{\isacharcolon}{\kern0pt}\ real\ \isakeyword{and}\ b{\isacharcolon}{\kern0pt}{\isacharcolon}{\kern0pt}real\isanewline
\ \ \isakeyword{assumes}\ posb{\isacharcolon}{\kern0pt}\ {\isachardoublequoteopen}b\ {\isachargreater}{\kern0pt}\ {\isacharparenleft}{\kern0pt}{\isadigit{0}}\ {\isacharcolon}{\kern0pt}{\isacharcolon}{\kern0pt}\ real{\isacharparenright}{\kern0pt}{\isachardoublequoteclose}\isanewline
\ \ \ \ \isakeyword{and}\ {\isachardoublequoteopen}length\ xs\ {\isacharequal}{\kern0pt}\ n{\isachardoublequoteclose}\ \isakeyword{and}\ {\isachardoublequoteopen}length\ ys\ {\isacharequal}{\kern0pt}\ n{\isachardoublequoteclose}\isanewline
\ \ \ \ \isakeyword{and}\ adj{\isacharcolon}{\kern0pt}\ {\isachardoublequoteopen}{\isacharparenleft}{\kern0pt}{\isasymSum}\ i{\isasymin}{\isacharbraceleft}{\kern0pt}{\isadigit{1}}{\isachardot}{\kern0pt}{\isachardot}{\kern0pt}n{\isacharbraceright}{\kern0pt}{\isachardot}{\kern0pt}\ {\isasymbar}\ nth\ xs\ {\isacharparenleft}{\kern0pt}i{\isacharminus}{\kern0pt}{\isadigit{1}}{\isacharparenright}{\kern0pt}\ {\isacharminus}{\kern0pt}\ nth\ ys\ {\isacharparenleft}{\kern0pt}i{\isacharminus}{\kern0pt}{\isadigit{1}}{\isacharparenright}{\kern0pt}\ {\isasymbar}{\isacharparenright}{\kern0pt}\ {\isasymle}\ r{\isachardoublequoteclose}\isanewline
\ \ \ \ \isakeyword{and}\ posr{\isacharcolon}{\kern0pt}\ {\isachardoublequoteopen}r\ {\isasymge}\ {\isadigit{0}}{\isachardoublequoteclose}\isanewline
\ \ \isakeyword{shows}\ {\isachardoublequoteopen}DP{\isacharunderscore}{\kern0pt}divergence\ {\isacharparenleft}{\kern0pt}Lap{\isacharunderscore}{\kern0pt}dist{\isacharunderscore}{\kern0pt}list\ b\ xs{\isacharparenright}{\kern0pt}\ {\isacharparenleft}{\kern0pt}Lap{\isacharunderscore}{\kern0pt}dist{\isacharunderscore}{\kern0pt}list\ b\ ys{\isacharparenright}{\kern0pt}\ {\isacharparenleft}{\kern0pt}r\ {\isacharslash}{\kern0pt}\ b{\isacharparenright}{\kern0pt}\ {\isasymle}\ {\isadigit{0}}{\isachardoublequoteclose}
\end{isabelle}

\subsubsection{Laplace Mechanism}\label{sec:Lap_mechanism}

We finally formalize the Laplace mechanism in Isabelle/HOL.
Again, we first consider general $X$ and $R^\mathrm{adj}$.
Later we instantiate $X = \Nat^{|\mathcal{X}|}$ and $R^\mathrm{adj} = \{ (D,D') ~|~ \|D - D' \| \leq 1\}$.
We introduce a locale for the proof.
\begin{isabelle}
\isacommand{locale}\isamarkupfalse%
\ Lap{\isacharunderscore}{\kern0pt}Mechanism{\isacharunderscore}{\kern0pt}list\ {\isacharequal}{\kern0pt}\isanewline
\ \ \isakeyword{fixes}\ X{\isacharcolon}{\kern0pt}{\isacharcolon}{\kern0pt}{\isachardoublequoteopen}{\isacharprime}{\kern0pt}a\ measure{\isachardoublequoteclose}\isanewline
\ \ \ \ \isakeyword{and}\ f{\isacharcolon}{\kern0pt}{\isacharcolon}{\kern0pt}{\isachardoublequoteopen}{\isacharprime}{\kern0pt}a\ {\isasymRightarrow}\ real\ list{\isachardoublequoteclose}\isanewline
\ \ \ \ \isakeyword{and}\ adj{\isacharcolon}{\kern0pt}{\isacharcolon}{\kern0pt}{\isachardoublequoteopen}{\isacharprime}{\kern0pt}a\ rel{\isachardoublequoteclose}\isanewline
\ \ \ \ \isakeyword{and}\ m{\isacharcolon}{\kern0pt}{\isacharcolon}{\kern0pt}nat\ \isanewline
\ \ \isakeyword{assumes}\ {\isacharbrackleft}{\kern0pt}measurable{\isacharbrackright}{\kern0pt}{\isacharcolon}{\kern0pt}\ {\isachardoublequoteopen}f\ {\isasymin}\ X\ {\isasymrightarrow}\isactrlsub M\ listM\ borel{\isachardoublequoteclose}\isanewline
\ \ \ \ \isakeyword{and}\ len{\isacharcolon}{\kern0pt}\ {\isachardoublequoteopen}{\isasymAnd}\ x{\isachardot}{\kern0pt}\ x\ {\isasymin}\ space\ X\ {\isasymLongrightarrow}\ length\ {\isacharparenleft}{\kern0pt}f\ x{\isacharparenright}{\kern0pt}\ {\isacharequal}{\kern0pt}\ m{\isachardoublequoteclose}\isanewline
\ \ \ \ \isakeyword{and}\ adj{\isacharcolon}{\kern0pt}\ {\isachardoublequoteopen}adj\ {\isasymsubseteq}\ space\ X\ {\isasymtimes}\ space\ X{\isachardoublequoteclose}\isanewline
\isakeyword{begin}
\end{isabelle}
\begin{isabelle}
\isacommand{definition}\isamarkupfalse%
\ sensitivity{\isacharcolon}{\kern0pt}{\isacharcolon}{\kern0pt}\ ereal\ \isakeyword{where}\isanewline
\ \ {\isachardoublequoteopen}sensitivity\ {\isacharequal}{\kern0pt}\ Sup{\isacharbraceleft}{\kern0pt}\ ereal\ {\isacharparenleft}{\kern0pt}\ {\isasymSum}\ i{\isasymin}{\isacharbraceleft}{\kern0pt}{\isadigit{1}}{\isachardot}{\kern0pt}{\isachardot}{\kern0pt}m{\isacharbraceright}{\kern0pt}{\isachardot}{\kern0pt}\ {\isasymbar}\ nth\ {\isacharparenleft}{\kern0pt}f\ x{\isacharparenright}{\kern0pt}\ {\isacharparenleft}{\kern0pt}i{\isacharminus}{\kern0pt}{\isadigit{1}}{\isacharparenright}{\kern0pt}\ {\isacharminus}{\kern0pt}\ nth\ {\isacharparenleft}{\kern0pt}f\ y{\isacharparenright}{\kern0pt}\ {\isacharparenleft}{\kern0pt}i{\isacharminus}{\kern0pt}{\isadigit{1}}{\isacharparenright}{\kern0pt}\ {\isasymbar}{\isacharparenright}{\kern0pt}\ {\isacharbar}{\kern0pt}\ x\ y{\isacharcolon}{\kern0pt}{\isacharcolon}{\kern0pt}{\isacharprime}{\kern0pt}a{\isachardot}{\kern0pt}\ x\ {\isasymin}\ space\ X\ {\isasymand}\ y\ {\isasymin}\ space\ X\ {\isasymand}\ {\isacharparenleft}{\kern0pt}x{\isacharcomma}{\kern0pt}y{\isacharparenright}{\kern0pt}\ {\isasymin}\ adj{\isacharbraceright}{\kern0pt}{\isachardoublequoteclose}
\end{isabelle}
\begin{isabelle}
\isacommand{definition}\isamarkupfalse%
\ LapMech{\isacharunderscore}{\kern0pt}list{\isacharcolon}{\kern0pt}{\isacharcolon}{\kern0pt}{\isachardoublequoteopen}real\ {\isasymRightarrow}\ {\isacharprime}{\kern0pt}a\ {\isasymRightarrow}\ {\isacharparenleft}{\kern0pt}real\ list{\isacharparenright}{\kern0pt}\ measure{\isachardoublequoteclose}\ \isakeyword{where}\isanewline
{\isachardoublequoteopen}LapMech{\isacharunderscore}{\kern0pt}list\ {\isasymepsilon}\ x\ {\isacharequal}{\kern0pt}\ Lap{\isacharunderscore}{\kern0pt}dist{\isacharunderscore}{\kern0pt}list\ {\isacharparenleft}{\kern0pt}{\isacharparenleft}{\kern0pt}real{\isacharunderscore}{\kern0pt}of{\isacharunderscore}{\kern0pt}ereal\ sensitivity{\isacharparenright}{\kern0pt}\ {\isacharslash}{\kern0pt}\ {\isasymepsilon}{\isacharparenright}{\kern0pt}\ {\isacharparenleft}{\kern0pt}f\ x{\isacharparenright}{\kern0pt}{\isachardoublequoteclose}
\end{isabelle}
\begin{isabelle}
\isacommand{lemma}\isamarkupfalse%
\ LapMech{\isacharunderscore}{\kern0pt}list{\isacharunderscore}{\kern0pt}def{\isadigit{2}}{\isacharcolon}{\kern0pt}\isanewline
\ \ \isakeyword{assumes}\ {\isachardoublequoteopen}x\ {\isasymin}\ space\ X{\isachardoublequoteclose}\isanewline
\ \ \isakeyword{shows}\ {\isachardoublequoteopen}LapMech{\isacharunderscore}{\kern0pt}list\ {\isasymepsilon}\ x\ {\isacharequal}{\kern0pt}\ do{\isacharbraceleft}{\kern0pt}\ xs\ {\isasymleftarrow}\ Lap{\isacharunderscore}{\kern0pt}dist{\isadigit{0}}{\isacharunderscore}{\kern0pt}list\ {\isacharparenleft}{\kern0pt}real{\isacharunderscore}{\kern0pt}of{\isacharunderscore}{\kern0pt}ereal\ sensitivity\ {\isacharslash}{\kern0pt}\ {\isasymepsilon}{\isacharparenright}{\kern0pt}\ m{\isacharsemicolon}{\kern0pt}\ return\ {\isacharparenleft}{\kern0pt}listM\ borel{\isacharparenright}{\kern0pt}\ {\isacharparenleft}{\kern0pt}map{\isadigit{2}}\ {\isacharparenleft}{\kern0pt}{\isacharplus}{\kern0pt}{\isacharparenright}{\kern0pt}\ {\isacharparenleft}{\kern0pt}f\ x{\isacharparenright}{\kern0pt}\ xs{\isacharparenright}{\kern0pt}{\isacharbraceright}{\kern0pt}{\isachardoublequoteclose}
\end{isabelle}
\begin{isabelle}
\isacommand{proposition}\isamarkupfalse%
\ differential{\isacharunderscore}{\kern0pt}privacy{\isacharunderscore}{\kern0pt}LapMech{\isacharunderscore}{\kern0pt}list{\isacharcolon}{\kern0pt}\isanewline
\ \ \isakeyword{assumes}\ pose{\isacharcolon}{\kern0pt}\ {\isachardoublequoteopen}{\isasymepsilon}\ {\isachargreater}{\kern0pt}\ {\isadigit{0}}{\isachardoublequoteclose}\ \isakeyword{and}\ {\isachardoublequoteopen}sensitivity\ {\isachargreater}{\kern0pt}\ {\isadigit{0}}{\isachardoublequoteclose}\ \isakeyword{and}\ {\isachardoublequoteopen}sensitivity\ {\isacharless}{\kern0pt}\ {\isasyminfinity}{\isachardoublequoteclose}\isanewline
\ \ \isakeyword{shows}\ {\isachardoublequoteopen}differential{\isacharunderscore}{\kern0pt}privacy\ {\isacharparenleft}{\kern0pt}LapMech{\isacharunderscore}{\kern0pt}list\ {\isasymepsilon}{\isacharparenright}{\kern0pt}\ adj\ {\isasymepsilon}\ {\isadigit{0}}{\isachardoublequoteclose}
\isanewline
\isacommand{end}
\end{isabelle}
Here, the formal definition of $\mathrm{LapMech}_{f,m,b}(D) $ follows equation (\ref{eq:LapMech_def2_rec}).
The equation (\ref{eq:LapMech_def2}) is formalized as \isa{LapMech{\isacharunderscore}{\kern0pt}list{\isacharunderscore}{\kern0pt}def{\isadigit{2}}}.
The formal proof of \isa{differential{\isacharunderscore}{\kern0pt}privacy{\isacharunderscore}{\kern0pt}LapMech{\isacharunderscore}{\kern0pt}list} is done by combining the lemmas
\isa{differential{\isacharunderscore}{\kern0pt}privacy{\isacharunderscore}{\kern0pt}preprocessing}
and
\isa{DP{\isacharunderscore}{\kern0pt}Lap{\isacharunderscore}{\kern0pt}dist{\isacharunderscore}{\kern0pt}list}.

\subsection{Instantiatiation of the Datasets and Adjacency}
%The above proposition \isa{differential{\isacharunderscore}{\kern0pt}privacy{\isacharunderscore}{\kern0pt}LapMech{\isacharunderscore}{\kern0pt}list} is for general \isa{X} and \isa{adj}.
To formalize differential privacy in the sense of \cite{DworkRothTCS-042}, 
we need to instantiate \isa{X} and \isa{adj} according to the situation in \cite{DworkRothTCS-042}.
To do this, we introduce the following locale.
\begin{isabelle}
\isacommand{locale}\isamarkupfalse%
\ results{\isacharunderscore}{\kern0pt}AFDP\ {\isacharequal}{\kern0pt}\isanewline
\ \ \isakeyword{fixes}\ n\ {\isacharcolon}{\kern0pt}{\isacharcolon}{\kern0pt}nat\ \isanewline
\isakeyword{begin}
\end{isabelle}
\begin{isabelle}
\isacommand{interpretation}\isamarkupfalse%
\ L{\isadigit{1}}{\isacharunderscore}{\kern0pt}norm{\isacharunderscore}{\kern0pt}list\ {\isachardoublequoteopen}{\isacharparenleft}{\kern0pt}UNIV{\isacharcolon}{\kern0pt}{\isacharcolon}{\kern0pt}nat\ set{\isacharparenright}{\kern0pt}{\isachardoublequoteclose}{\isachardoublequoteopen}{\isacharparenleft}{\kern0pt}{\isasymlambda}\ x\ y{\isachardot}{\kern0pt}\ {\isasymbar}int\ x\ {\isacharminus}{\kern0pt}\ int\ y{\isasymbar}{\isacharparenright}{\kern0pt}{\isachardoublequoteclose}n
\end{isabelle}
\begin{isabelle}
\isacommand{definition}\isamarkupfalse%
\ sp{\isacharunderscore}{\kern0pt}Dataset\ {\isacharcolon}{\kern0pt}{\isacharcolon}{\kern0pt}\ {\isachardoublequoteopen}nat\ list\ measure{\isachardoublequoteclose}\ \isakeyword{where}\isanewline
\ \ {\isachardoublequoteopen}sp{\isacharunderscore}{\kern0pt}Dataset\ {\isasymequiv}\ restrict{\isacharunderscore}{\kern0pt}space\ {\isacharparenleft}{\kern0pt}listM\ {\isacharparenleft}{\kern0pt}count{\isacharunderscore}{\kern0pt}space\ UNIV{\isacharparenright}{\kern0pt}{\isacharparenright}{\kern0pt}\ space{\isacharunderscore}{\kern0pt}L{\isadigit{1}}{\isacharunderscore}{\kern0pt}norm{\isacharunderscore}{\kern0pt}list{\isachardoublequoteclose}
\end{isabelle}
\begin{isabelle}
\isacommand{definition}\isamarkupfalse%
\ adj{\isacharunderscore}{\kern0pt}L{\isadigit{1}}{\isacharunderscore}{\kern0pt}norm\ {\isacharcolon}{\kern0pt}{\isacharcolon}{\kern0pt}\ {\isachardoublequoteopen}{\isacharparenleft}{\kern0pt}nat\ list\ {\isasymtimes}\ nat\ list{\isacharparenright}{\kern0pt}\ set{\isachardoublequoteclose}\ \isakeyword{where}\isanewline
\ \ {\isachardoublequoteopen}adj{\isacharunderscore}{\kern0pt}L{\isadigit{1}}{\isacharunderscore}{\kern0pt}norm\ {\isasymequiv}\ {\isacharbraceleft}{\kern0pt}{\isacharparenleft}{\kern0pt}xs{\isacharcomma}{\kern0pt}ys{\isacharparenright}{\kern0pt}{\isacharbar}{\kern0pt}\ xs\ ys{\isachardot}{\kern0pt}\ xs\ {\isasymin}\ space\ sp{\isacharunderscore}{\kern0pt}Dataset\ {\isasymand}\ ys\ {\isasymin}\ space\ sp{\isacharunderscore}{\kern0pt}Dataset\ {\isasymand}\ dist{\isacharunderscore}{\kern0pt}L{\isadigit{1}}{\isacharunderscore}{\kern0pt}norm{\isacharunderscore}{\kern0pt}list\ xs\ ys\ {\isasymle}\ {\isadigit{1}}{\isacharbraceright}{\kern0pt}{\isachardoublequoteclose}
\end{isabelle}
\begin{isabelle}
\isacommand{abbreviation}\isamarkupfalse
\ {\isachardoublequoteopen}differential{\isacharunderscore}{\kern0pt}privacy{\isacharunderscore}{\kern0pt}AFDP\ M\ {\isasymepsilon}\ {\isasymdelta}\ {\isasymequiv}\ differential{\isacharunderscore}{\kern0pt}privacy\ M\ adj{\isacharunderscore}{\kern0pt}L{\isadigit{1}}{\isacharunderscore}{\kern0pt}norm\ {\isasymepsilon}\ {\isasymdelta}{\isachardoublequoteclose}
\end{isabelle}
Here, \isa{n} is the number $|\mathcal{X}|$ of data types (i.e. the length of datasets);
\isa{L{\isadigit{1}}{\isacharunderscore}{\kern0pt}norm{\isacharunderscore}{\kern0pt}list} is the locale 
for $L_1$-norm of lists introduced in Section \ref{sec:list_metric};
\isa{sp{\isacharunderscore}{\kern0pt}Dataset} is the space $\Nat^{|\mathcal{X}|}$;
%\isa{dist{\isacharunderscore}{\kern0pt}L{\isadigit{1}}{\isacharunderscore}{\kern0pt}norm{\isacharunderscore}{\kern0pt}list} is the metric $\|-\|_1$;
\isa{adj{\isacharunderscore}{\kern0pt}L{\isadigit{1}}{\isacharunderscore}{\kern0pt}norm} is the 
adjacency relation $R^\mathrm{adj} = \{ (D,D') ~|~ \|D - D' \| \leq 1\}$.

In this locale, we formalize Lemmas \ref{DP:basic:trivial}, \ref{DP:basic:postprocessing}, \ref{DP:basic:group}, \ref{DP:basic:sequential} and \ref{DP:basic:adaptive}.
For Lemmas \ref{DP:basic:trivial}, \ref{DP:basic:postprocessing}, \ref{DP:basic:sequential} and \ref{DP:basic:adaptive}, 
we just instantiate \isa{sp{\isacharunderscore}{\kern0pt}Dataset} and \isa{adj{\isacharunderscore}{\kern0pt}L{\isadigit{1}}{\isacharunderscore}{\kern0pt}norm} to their general versions. 
Lemma \ref{DP:basic:group} is formalized as:
\begin{isabelle}
\isacommand{lemma}\isamarkupfalse%
\ group{\isacharunderscore}{\kern0pt}privacy{\isacharunderscore}{\kern0pt}AFDP{\isacharcolon}{\kern0pt}\isanewline
\ \ \isakeyword{assumes}\ M{\isacharcolon}{\kern0pt}\ {\isachardoublequoteopen}M\ {\isasymin}\ sp{\isacharunderscore}{\kern0pt}Dataset\ {\isasymrightarrow}\isactrlsub M\ prob{\isacharunderscore}{\kern0pt}algebra\ Y{\isachardoublequoteclose}\isanewline
\ \ \ \ \isakeyword{and}\ DP{\isacharcolon}{\kern0pt}\ {\isachardoublequoteopen}\ differential{\isacharunderscore}{\kern0pt}privacy{\isacharunderscore}{\kern0pt}AFDP\ M\ {\isasymepsilon}\ {\isadigit{0}}{\isachardoublequoteclose}\isanewline
\ \ \isakeyword{shows}\ {\isachardoublequoteopen}differential{\isacharunderscore}{\kern0pt}privacy\ M\ {\isacharparenleft}{\kern0pt}dist{\isacharunderscore}{\kern0pt}L{\isadigit{1}}{\isacharunderscore}{\kern0pt}norm\ k{\isacharparenright}{\kern0pt}\ {\isacharparenleft}{\kern0pt}real\ k\ {\isacharasterisk}{\kern0pt}\ {\isasymepsilon}{\isacharparenright}{\kern0pt}\ {\isadigit{0}}{\isachardoublequoteclose}
\end{isabelle}
Here, \isa{dist{\isacharunderscore}{\kern0pt}L{\isadigit{1}}{\isacharunderscore}{\kern0pt}norm\ k} is an implementation of the binary relation $R_k = \{ (D,D') ~|~ \|D - D' \| \leq k\}$. 
In the formal proof, we use a formal version of Lemma \ref{lem:adj_k:group} (see, also Secion \ref{sec:list_metric}).

To complete the formalization of the Laplace mechanism (e.g. formal proof of Proposition \ref{prop:DP:Lapmech}), we set the following context. 
We interpret the locale \isa{Lap{\isacharunderscore}{\kern0pt}Mechanism{\isacharunderscore}{\kern0pt}list} with \isa{sp{\isacharunderscore}{\kern0pt}Dataset} and \isa{adj{\isacharunderscore}{\kern0pt}L{\isadigit{1}}{\isacharunderscore}{\kern0pt}norm}.
That interpretation provides the differential privacy of the instance of \isa{LapMech{\isacharunderscore}{\kern0pt}list},
which is an inductive version of \isa{LapMech{\isacharunderscore}{\kern0pt}list{\isacharunderscore}{\kern0pt}AFDP}.

Finally, the formal proof of Proposition \ref{prop:DP:Lapmech} is done by proving that \isa{ LapMech{\isacharunderscore}{\kern0pt}list{\isacharunderscore}{\kern0pt}AFDP}
and the instance of \isa{LapMech{\isacharunderscore}{\kern0pt}list} are equal (Lemma \isa{LapMech{\isacharunderscore}{\kern0pt}list{\isacharunderscore}{\kern0pt}AFDP{\isacharprime}}).
%It seems obvious, but we need it due to Isabelle's type system.

\begin{isabelle}
\isacommand{context}\isamarkupfalse%
\isanewline
\ \ \isakeyword{fixes}\ f{\isacharcolon}{\kern0pt}{\isacharcolon}{\kern0pt}{\isachardoublequoteopen}nat\ list\ {\isasymRightarrow}\ real\ list{\isachardoublequoteclose}\isanewline
\ \ \ \ \isakeyword{and}\ m{\isacharcolon}{\kern0pt}{\isacharcolon}{\kern0pt}nat\ \isanewline
\ \ \isakeyword{assumes}\ {\isacharbrackleft}{\kern0pt}measurable{\isacharbrackright}{\kern0pt}{\isacharcolon}{\kern0pt}\ {\isachardoublequoteopen}f\ {\isasymin}\ sp{\isacharunderscore}{\kern0pt}Dataset\ {\isasymrightarrow}\isactrlsub M\ {\isacharparenleft}{\kern0pt}listM\ borel{\isacharparenright}{\kern0pt}{\isachardoublequoteclose}\isanewline
\ \ \ \ \isakeyword{and}\ len{\isacharcolon}{\kern0pt}\ {\isachardoublequoteopen}{\isasymAnd}\ x{\isachardot}{\kern0pt}\ x\ {\isasymin}\ space\ X\ {\isasymLongrightarrow}\ length\ {\isacharparenleft}{\kern0pt}f\ x{\isacharparenright}{\kern0pt}\ {\isacharequal}{\kern0pt}\ m{\isachardoublequoteclose}\isanewline
\isakeyword{begin}
\end{isabelle}
\begin{isabelle}
\isacommand{interpretation}\isamarkupfalse%
\ Lap{\isacharunderscore}{\kern0pt}Mechanism{\isacharunderscore}{\kern0pt}list\ {\isachardoublequoteopen}sp{\isacharunderscore}{\kern0pt}Dataset{\isachardoublequoteclose}\ f\ {\isachardoublequoteopen}adj{\isacharunderscore}{\kern0pt}L{\isadigit{1}}{\isacharunderscore}{\kern0pt}norm{\isachardoublequoteclose}\ m
\end{isabelle}
\begin{isabelle}
\isacommand{definition}\isamarkupfalse%
\ LapMech{\isacharunderscore}{\kern0pt}list{\isacharunderscore}{\kern0pt}AFDP\ {\isacharcolon}{\kern0pt}{\isacharcolon}{\kern0pt}\ {\isachardoublequoteopen}real\ {\isasymRightarrow}\ nat\ list\ {\isasymRightarrow}\ real\ list\ measure{\isachardoublequoteclose}\ \isakeyword{where}\isanewline
\ \ {\isachardoublequoteopen}LapMech{\isacharunderscore}{\kern0pt}list{\isacharunderscore}{\kern0pt}AFDP\ {\isasymepsilon}\ x\ {\isacharequal}{\kern0pt}\ do{\isacharbraceleft}{\kern0pt}\ ys\ {\isasymleftarrow}\ {\isacharparenleft}{\kern0pt}Lap{\isacharunderscore}{\kern0pt}dist{\isadigit{0}}{\isacharunderscore}{\kern0pt}list{\isacharparenleft}{\kern0pt}real{\isacharunderscore}{\kern0pt}of{\isacharunderscore}{\kern0pt}ereal\ sensitivity\ {\isacharslash}{\kern0pt}\ {\isasymepsilon}{\isacharparenright}{\kern0pt}\ m{\isacharparenright}{\kern0pt}{\isacharsemicolon}{\kern0pt}\ return\ {\isacharparenleft}{\kern0pt}listM\ borel{\isacharparenright}{\kern0pt}\ {\isacharparenleft}{\kern0pt}map{\isadigit{2}}\ {\isacharparenleft}{\kern0pt}{\isacharplus}{\kern0pt}{\isacharparenright}{\kern0pt}\ {\isacharparenleft}{\kern0pt}f\ x{\isacharparenright}{\kern0pt}\ ys{\isacharparenright}{\kern0pt}\ {\isacharbraceright}{\kern0pt}{\isachardoublequoteclose}
\end{isabelle}
\begin{isabelle}
\isacommand{lemma}\isamarkupfalse%
\ LapMech{\isacharunderscore}{\kern0pt}list{\isacharunderscore}{\kern0pt}AFDP{\isacharprime}{\kern0pt}{\isacharcolon}{\kern0pt}\isanewline
\ \ \isakeyword{assumes}\ {\isachardoublequoteopen}x\ {\isasymin}\ space\ sp{\isacharunderscore}{\kern0pt}Dataset{\isachardoublequoteclose}\isanewline
\ \ \isakeyword{shows}\ {\isachardoublequoteopen}LapMech{\isacharunderscore}{\kern0pt}list{\isacharunderscore}{\kern0pt}AFDP\ {\isasymepsilon}\ x\ {\isacharequal}{\kern0pt}\ LapMech{\isacharunderscore}{\kern0pt}list\ {\isasymepsilon}\ x{\isachardoublequoteclose}
\end{isabelle}
\begin{isabelle}
\isacommand{lemma}\isamarkupfalse%
\ differential{\isacharunderscore}{\kern0pt}privacy{\isacharunderscore}{\kern0pt}Lap{\isacharunderscore}{\kern0pt}Mechanism{\isacharunderscore}{\kern0pt}list{\isacharunderscore}{\kern0pt}AFDP{\isacharcolon}{\kern0pt}\isanewline
\ \ \isakeyword{assumes}\ {\isachardoublequoteopen}{\isadigit{0}}\ {\isacharless}{\kern0pt}\ {\isasymepsilon}{\isachardoublequoteclose}\ \isakeyword{and}\ {\isachardoublequoteopen}{\isadigit{0}}\ {\isacharless}{\kern0pt}\ sensitivity{\isachardoublequoteclose}\ \isakeyword{and}\ {\isachardoublequoteopen}sensitivity\ {\isacharless}{\kern0pt}\ {\isasyminfinity}{\isachardoublequoteclose}\isanewline
\ \ \isakeyword{shows}\ {\isachardoublequoteopen}differential{\isacharunderscore}{\kern0pt}privacy{\isacharunderscore}{\kern0pt}AFDP\ {\isacharparenleft}{\kern0pt}LapMech{\isacharunderscore}{\kern0pt}list{\isacharunderscore}{\kern0pt}AFDP\ {\isasymepsilon}{\isacharparenright}{\kern0pt}\ {\isasymepsilon}\ {\isadigit{0}}{\isachardoublequoteclose}
\end{isabelle}
\begin{isabelle}
\isacommand{end}
\end{isabelle}

\section{Report Noisy Max Mechanism}\label{sec:report_noisy_max_DP}

The \emph{report noisy max mechanism} is a randomized algorithm that returns the index of maximum values in tuples sampled from a Laplace mechanism.
\begin{align*}
\mathrm{RNM}_{f,m,\varepsilon}(D) = \{
&(y_0,\ldots,y_{m-1}) \leftarrow \mathrm{LapMech}_{f,m,1/\varepsilon}(D); \\ 
&\return {\textstyle \argmax_{0\leq i < m}}y_i \}\}.
\end{align*}
Here, we assume that $\argmax_{0\leq i < m}y_i$ returns the least $i$ such that $y_i$ is the maximum value of $\{y_0,\ldots,y_{m-1}\}$.

We can prove that it is $(\Delta f \cdot \varepsilon,0)$-DP by Lemma \ref{DP:basic:postprocessing} for the post-processing $(y_0,\ldots,y_{m-1}) \mapsto \return \argmax_{0\leq i < m}y_i$ and Lemma \ref{prop.DP.LapMech}.
The strength $\Delta f \cdot \varepsilon$ of privacy guarantee may depend on the length $m$ of outputs.
We suppose that $f$ is a tuple of $m$ functions defined by $f(D) = (f_0(D),\ldots,f_{m-1}(D))$ for each $D \in \Nat^{|\mathcal{X}|}$ and $\Delta f_i = 1$ holds for each $0 \leq i < m$.
Then, $\Delta f$ is at most $m$\footnote{Moreover, we have $\Delta f = m$ if we choose $\Delta f_0 = 1$ and $f_i = f_0$ for $0 < i < m$.}.
Thus, we naively conclude that $\mathrm{RNM}_{f,m,\varepsilon}$ is $(m \cdot \varepsilon,0)$-DP.

However, when $f$ is a tuple of $m$ \emph{counting queries} described below, 
$\mathrm{RNM}_{f,m,\varepsilon}$ is actually $(\varepsilon,0)$-DP regardless of $m$.

\subsection{Counting Queries}
A counting query is a function that counts the number of elements in certain types in a given dataset.

Since each dataset $D$ is a histogram in which each $D[i]$ represents the number of elements of type $i$,
a counting query is formulated as a function $q \colon \Nat^{|\mathcal{X}|} \to \Nat$ defined by
\[
q(D) = \sum_{i \in \mathcal{X}_q} D[i].
\]
Here, $\mathcal{X}_q \subseteq \mathcal{X}$ is the set of types in which $q$ counts the number of elements that belong. 
We regard each $i \in X_q$ as their indexes $0 \leq i < |\mathcal{X}|$.

\subsection{Differential Privacy}
We consider a tuple $\vec{q} \colon \Nat^{|\mathcal{X}|} \to \Nat^m$ of $m$ counting queries $q_i \colon \Nat^{|\mathcal{X}|} \to \Nat$ ($0 \leq i < m$).
We regard it as a function $\vec{q} \colon \Nat^{|\mathcal{X}|} \to \Real^m$ to apply the Laplace mechanism.

First, we evaluate naively the differential privacy of $\mathrm{RNM}_{\vec{q},m,\varepsilon}$.
The sensitivity of any counting query $q$ is $1$ because $| q(D) - q(D') | \leq 1$ for any adjacent datasets $D,D' \in \Nat^{|\mathcal{X}|}$.
%Hence, the sensitivity of a tuple $\vec{q} \colon \Nat^{|\mathcal{X}|} \to \Nat^m$ of $m$ counting queries $q_i \colon \Nat^{|\mathcal{X}|} \to \Nat$ ($0 \leq i < m$) is $m$.
Hence, the sensitivity of the $\vec{q}$ is at most $m$. 
Thus, 
\begin{proposition}[Naive]\label{prop:DP:RNM_naive}
$\mathrm{RNM}_{\vec{q},m,\varepsilon}$ is $(m\cdot \varepsilon,0)$-DP.
\end{proposition}

Next, by using specific properties of counting queries and $\argmax$ operations,
regardless of $m$, we conclude, 
\begin{proposition}[{\cite[Claim 3.9]{DworkRothTCS-042}}]\label{prop:DP:RNM_true}
$\mathrm{RNM}_{\vec{q},m,\varepsilon}$ is $(\varepsilon,0)$-DP.
\end{proposition}

To give a formal proof of this proposition, we reconstruct the proof in \cite[Claim 3.9]{DworkRothTCS-042}. 
Instead of the minimum value $r^\ast = \min\{r_i | \forall j \neq i.  c_i + r_i > c_j + r_j \}$,
we use its underlying set $\{r | c_i + r \geq \max_{j \neq i} {c_j + r_j }\}$.
For example, the probability $\Pr_{r_i \sim \mathrm{Lap}(1/\varepsilon)}[r_i \geq r^\ast]$ in the original proof 
is rewritten by $\Pr_{r \sim \mathrm{Lap}(1/\varepsilon)}[ c_i + r \geq \max_{j \neq i} {c_j + r_j }]$.
In addition, for easy formalization, we also partition the proof into several lemmas. 

\begin{proof}[A reconstructed proof of Proposition \ref{prop:DP:RNM_true}]
We first define the main body $\mathrm{RNM'}_{m,\varepsilon}$ of $\mathrm{RNM}_{\vec{q},m,\varepsilon}$ as follows:
\begin{align*}
\mathrm{RNM'}_{m,\varepsilon}(\vec{c}) &= \{ (z_0,\ldots, z_{m-1}) \leftarrow \mathrm{Lap}^{m}(1/\varepsilon,\vec{c}); \\
& \qquad \return {\textstyle \argmax_{0 \leq j < m}} (z_j) \}.
\end{align*}
By the definition of $\mathrm{RNM}_{\vec{q},m,\varepsilon}$ and equation (\ref{eq:LapMech_def2_rec}), we obtain,
\begin{equation}\label{eq:LapMech_RNM_RNM'_c}
\mathrm{RNM}_{\vec{q},m,\varepsilon} = \mathrm{RNM'}_{m,\varepsilon} \circ \vec{q}.
\end{equation}

We here fix adjacent datasets $D, D' \in \Nat^{|\mathcal{X}|}$, and write
\begin{align*}
\vec{q}(D) &= (c_0,\ldots,c_{m-1}) = \vec{c},&
\vec{q}(D') &= (c'_0,\ldots,c'_{m-1}) = \vec{c'}.
\end{align*}
Here, $| D[i] - D'[i] | \leq1$ and $D[j] = D'[j]$ ($j \neq i$) holds for some $i$.
Since $\vec{q}$ is a tuple of counting queries, we obtain,
\begin{lemma}[Lipschitz \& Monotonicity]\label{lem:adjacent:counting_query}
We have one of the two following two cases:
\begin{itemize}
\item[(A)] $c_j \geq c'_j$ and $c_j\leq c'_j + 1$ for all $0 \leq j < m$,
\item[(B)] $c'_j \geq c_j$ and $c'_j\leq c_j + 1$ for all $0 \leq j < m$.
\end{itemize}
\end{lemma}

We here fix $c_j, c'_j \in \Real$ ($0 \leq j < m$) and assume (A) without loss of generality.
%Before proving this lemma, we provide several lemmas.
We have the following lemmas.
\begin{lemma}\label{lem:fst_max_argmax_adj}
Fix arbitrary values $r_j \in \Real$ ($0 \leq j < m$).
We write $d = \max_{0\leq j < m} (c_j + r_j) $ and $d' = \max_{0\leq j < m} (c'_j + r_j)$.
Then, we have $d \geq d'$ and $d \leq d' + 1$.
\end{lemma}

\begin{lemma}\label{lem:argmax_insert_i_i}
For all $d_j \in \Real$ ($0 \leq j < m$) and $0 \leq i < m$, 
\begin{align*}
\argmax_{0 \leq j < m} d_j = i
%&\iff \max_{0 \leq j < i} d_j < z \land \max_{i < j < m} d_j \leq z\\
&\iff \max_{0 \leq j < m, j \neq i} d_j \leq d_i \land \max_{0 \leq j < i} d_j \neq d_i.
\end{align*}
\end{lemma}

We here define for each index $0 \leq i < m$ and fixed values of noise $r_j \in \Real$ ($0 \leq j < m, j \neq i$),
\begin{align}
p_i &= \!\!\Pr_{r_i \sim \mathrm{Lap}(1/\varepsilon)}\left[{\textstyle \argmax_{0 \leq j < m}} (c_j + r_j) = i \right],\label{eq:p_i}\\
p'_i &= \!\!\Pr_{r_i \sim \mathrm{Lap}(1/\varepsilon)}\left[{\textstyle \argmax_{0 \leq j < m}} (c'_j + r_j) = i \right].\notag
\end{align}
Then we obtain, 
\begin{lemma}\label{lem:DP_RNM_M_i}
For each $0 \leq i < m$, $r_j \in \Real$ ($0 \leq j < m$ and $j \neq i$),
$p_i \leq \exp(\varepsilon) p'_i$ and $p'_i \leq \exp(\varepsilon) p_i$.
\end{lemma}
\begin{proof}[Proof sketch]
%By Lemma \ref{lem:argmax_insert_i_i}, and the fact $\mathrm{Lap}(1/\varepsilon)\{r\} = 0$ for any $r \in \Real$,
%we obtain for each $0 \leq i < m$, $r_j \in \Real$ ($0 \leq j < m$ and $j \neq i$), 
%\begin{lemma}
%\begin{align}\label{eq:RNM_M_probability}
%p_i = \Pr_{r_i \sim \mathrm{Lap}(1/\varepsilon)}\left[c_i + r_i \geq {\textstyle \max_{j \neq i}} {c_j + r_j } \right],\\
%\label{eq:RNM_M_probability2}
%p'_i = \Pr_{r_i \sim \mathrm{Lap}(1/\varepsilon)}\left[c'_i + r_i \geq {\textstyle \max_{j \neq i}} {c'_j + r_j } \right].
%\end{align}
%By Lemmas \ref{lem:Lap:DP:finer} and \ref{lem:fst_max_argmax_adj} and equations (\ref{eq:RNM_M_probability}) and (\ref{eq:RNM_M_probability2}), we have,
%By Lemma \ref{lem:fst_max_argmax_adj} and assumption (A),
%\begin{align*}
%\{ r_i | c_i + r_i \geq {\textstyle \max_{j \neq i}} {c_j + r_j }\} &\subseteq \{ r_i | c'_i + (r_i + 1) \geq {\textstyle \max_{j \neq i}} {c'_j + r_j }\},\\
%\{ r_i | c'_i + r_i \geq {\textstyle \max_{j \neq i}} {c'_j + r_j }\} &\subseteq \{ r_i | c_i + (r_i + 1) \geq {\textstyle \max_{j \neq i}} {c_j + r_j }\}.\/\
%\end{align*}
By applying Lemmas \ref{lem:argmax_insert_i_i}, \ref{lem:fst_max_argmax_adj}, \ref{lem:Lap:DP:finer} and \ref{lem:Lap:shifting} in this order, 
we prove $p_i \leq \exp(\varepsilon) p'_i$.
%From these inclusions, Lemma \ref{lem:Lap:DP:finer} and equation (\ref{eq:RNM_M_probability}), we have 
\begin{align*}
p_i & = \Pr_{r_i \sim \mathrm{Lap}(1/\varepsilon,0)}\left[c_i + r_i \geq {\textstyle \max_{j \neq i}} {c_j + r_j } \right] \\% && \text{Lemma \ref{lem:argmax_insert_i_i} }\\
	& \leq \Pr_{r_i \sim \mathrm{Lap}(1/\varepsilon,0)}\left[ c'_i + (r_i + 1) \geq {\textstyle \max_{j \neq i}} {c'_j + r_j } \right] \\% &&\text{Lemma \ref{lem:fst_max_argmax_adj}} \\
	& \leq \exp(\varepsilon) \Pr_{r_i \sim \mathrm{Lap}(1/\varepsilon, -1)}\left[ c'_i + (r_i + 1) \geq {\textstyle \max_{j \neq i}} {c'_j + r_j } \right] \\
	& = \exp(\varepsilon) \Pr_{r_i \sim \mathrm{Lap}(1/\varepsilon,0)}\left[ c'_i + r_i\geq {\textstyle \max_{j \neq i}} {c'_j + r_j } \right] = p'_i.
\end{align*}
Similarly, we also have $p'_i \leq \exp(\varepsilon) p_i$.
\end{proof}
%By the symmetry of (A) and (B),
%if $c_j, c'_j$ satisfy the case {\rm (B)}, 
%$p_i \leq \exp(\varepsilon) p'_i$ and $p'_i \leq \exp(\varepsilon) p_i$ also holds.
%
From the definition of $\mathrm{Lap}^{m}(1/\varepsilon)$, equation (\ref{eq:LapMech_def3}), and
the monad laws and commutativity of Giry monad,
for each $0 \leq i < m$ , we obtain (we write $\vec{r}_{-i} = (r_0,\ldots,r_{i-1},r_{i+1},\ldots,r_{m-1})$):
\begin{equation}\label{eq:RNM_expand}
\begin{aligned}
&\mathrm{RNM'}_{m,\varepsilon}(x_0,\ldots,x_{m-1})\\
&\quad = \{ \vec{r}_{-i} \leftarrow \mathrm{Lap}^{m-1}(1/\varepsilon);\\
&\quad \qquad r_i \leftarrow \mathrm{Lap}(1/\varepsilon) ; \return {\textstyle \argmax_{0 \leq j < m}} (x_j + r_j) \}.
\end{aligned}
\end{equation}
We hence prove core inequations for Proposition \ref{prop:DP:RNM_true}.
\begin{lemma}\label{DP_RNM'_M_i}
Assume that $c_j, c'_j \in \Real$ ($0 \leq j < m$) satisfy (A).
Then for any $0 \leq i < m$,
\begin{align*}
\Pr[\mathrm{RNM'}_{m,\varepsilon}(\vec{c}) = i] &\leq \exp(\varepsilon) \Pr[\mathrm{RNM'}_{m,\varepsilon}(\vec{c'}) = i],\\
\Pr[\mathrm{RNM'}_{m,\varepsilon}(\vec{c'}) = i] &\leq \exp(\varepsilon) \Pr[\mathrm{RNM'}_{m,\varepsilon}(\vec{c}) = i].
\end{align*}
\end{lemma}
%\begin{proof}[Proof sketch]
%From the definition of $\mathrm{Lap}^{m}(1/\varepsilon)$ and the monad laws and commutativity of Giry monad,
%for each $0 \leq i < m$, we obtain (we let $\vec{r}_{-i} = (r_0,\ldots,r_{i-1},r_{i+1},\ldots,r_{m-1})$):
%\begin{equation}\label{eq:RNM_expand}
%\begin{aligned}
%&\mathrm{RNM'}_{m,\varepsilon}(x_0,\ldots,x_{m-1})\\
%&\quad = \{ \vec{r}_{-i} \leftarrow \mathrm{Lap}^{m-1}(1/\varepsilon);\\
%&\quad \qquad r_i \leftarrow \mathrm{Lap}(1/\varepsilon) ; \return {\textstyle \argmax_{0 \leq j < m}} (x_j + r_j) \}.
%\end{aligned}
%\end{equation}
%From this equation and the definitions of $p_i$ and $p_i$, 
%we can calculate: 
%\[
%\Pr[\mathrm{RNM'}_{m,\varepsilon}(\vec{c}) = i] 
%= \int p_i (\vec{r}_{-i})~d(\mathrm{Lap}^{m-1}(1/\varepsilon))(\vec{r}_{-i}).
%\]
%We can expand $\Pr[\mathrm{RNM'}_{m,\varepsilon}(\vec{c'}) = i]$ similary.
%Therefore, by Lemma \ref{lem:DP_RNM_M_i} and the monoinicity of integrations, we conclude the desired inequaltions.
%\end{proof}
By the symmetry of (A) and (B) and the symmetry of the statement of this lemma, we also have the same inequations in the case (B).
Hence, in both the cases (A) and (B),
for each $S \subseteq \{0,\ldots,m-1\}$, we conclude (by taking sums over $S$),
\[
\Pr[\mathrm{RNM'}_{m,\varepsilon}(\vec{c}) \in S] \leq \exp(\varepsilon) \Pr[\mathrm{RNM'}_{m,\varepsilon}(\vec{c'}) \in S].
\]

Hence, by Lemma \ref{lem:adjacent:counting_query} and equation (\ref{eq:LapMech_RNM_RNM'_c}), 
we conclude for any adjacent datasets $D, D' \in \Nat^{|\mathcal{X}|}$, 
\[
\Pr[\mathrm{RNM}_{\vec{q},m,\varepsilon}(D) \in S] \leq \exp(\varepsilon) \Pr[\mathrm{RNM}_{\vec{q},m,\varepsilon}(D') \in S].
\]
That is, $\mathrm{RNM}_{\vec{q},m,\varepsilon}$ is $(\varepsilon,0)$-DP.
\end{proof}

\subsubsection{Remark on Formalizing Proposition \ref{prop:DP:RNM_true}}
The formalization of %the finar evaluation for differential privacy of $\mathrm{RNM}_{\vec{q},m,\varepsilon}$
Proposition \ref{prop:DP:RNM_true} is based on the above pencil and paper proof.
The formal proof is quite long (about 1,000 lines) due to the following reasons.
First, we often need to deal with \emph{insertions} of elements to lists. 
In the proof, we mainly use the value $p_i$ (defined in equation (\ref{eq:p_i})).
%\begin{equation}\label{eq:p_i}
%p_i = \!\!\Pr_{r_i \sim \mathrm{Lap}(1/\varepsilon)}\left[{\textstyle \argmax_{0 \leq j < m}} (c_j + r_j) = i \right]
%\end{equation}
%for each $0 \leq i < m$ and fixed values $r_j \in \Real$ ($0 \leq j < m$ and $j \neq i$).
To implement $p_i$, we need to implement the function for fixed $r_j$ ($j \neq i $):
\[
	\lambda r_i.~{\textstyle \argmax_{0 \leq j < m}} (c_j + r_j) .
\]
It contains implicitly the insertion of $c_i + r_i$ to the list $(c_0+r_0,\ldots,c_{i-1}+r_{i-1},c_{i+1}+r_{i+1},c_{m-1}+r_{m-1})$ at $i$-th position.

Second, the position $i$ in such insertions of elements must be arbitrary, that is, we need to consider all $0 \leq i < m$.
Thus, we often need inductions on $i$.
We also need to check corner cases of $i$ and $m$, such as $i = 0$, $i = m-1$, $i \geq m$ and $m = 1$.

Third, due to the insertions, we apply monad laws and commutativity of Giry monad many times.

%In Isabelle/HOL, they always require measurability and hence we need to prove them,
%though almost all of them are automated or solved by 1-line proofs. 

\section{Formalization of the Report Noisy Max Mechanism in Isabelle/HOL}
\label{sec:formal_RNM}
\begin{figure*}
\begin{isabelle}
\isacommand{lemma}\isamarkupfalse%
\ finer{\isacharunderscore}{\kern0pt}sensitivity{\isacharunderscore}{\kern0pt}counting{\isacharunderscore}{\kern0pt}query{\isacharcolon}{\kern0pt}\isanewline
\ \ \isakeyword{assumes}\ {\isachardoublequoteopen}{\isacharparenleft}{\kern0pt}xs{\isacharcomma}{\kern0pt}ys{\isacharparenright}{\kern0pt}\ {\isasymin}\ adj{\isacharunderscore}{\kern0pt}L{\isadigit{1}}{\isacharunderscore}{\kern0pt}norm{\isachardoublequoteclose}\isanewline
\ \ \isakeyword{shows}\ {\isachardoublequoteopen}list{\isacharunderscore}{\kern0pt}all{\isadigit{2}}\ {\isacharparenleft}{\kern0pt}{\isasymlambda}\ x\ y{\isachardot}{\kern0pt}\ x\ {\isasymge}\ y\ {\isasymand}\ x\ {\isasymle}\ y\ {\isacharplus}{\kern0pt}\ {\isadigit{1}}{\isacharparenright}{\kern0pt}\ {\isacharparenleft}{\kern0pt}counting{\isacharunderscore}{\kern0pt}query\ xs{\isacharparenright}{\kern0pt}\ {\isacharparenleft}{\kern0pt}counting{\isacharunderscore}{\kern0pt}query\ ys{\isacharparenright}{\kern0pt}
\isanewline \ \ \ \ \ \ \ \ \ \ \ \ \ \ \ {\isasymor}\ list{\isacharunderscore}{\kern0pt}all{\isadigit{2}}\ {\isacharparenleft}{\kern0pt}{\isasymlambda}\ x\ y{\isachardot}{\kern0pt}\ x\ {\isasymge}\ y\ {\isasymand}\ x\ {\isasymle}\ y\ {\isacharplus}{\kern0pt}\ {\isadigit{1}}{\isacharparenright}{\kern0pt}\ {\isacharparenleft}{\kern0pt}counting{\isacharunderscore}{\kern0pt}query\ ys{\isacharparenright}{\kern0pt}\ {\isacharparenleft}{\kern0pt}counting{\isacharunderscore}{\kern0pt}query\ xs{\isacharparenright}{\kern0pt}{\isachardoublequoteclose}
\end{isabelle}
\begin{isabelle}
\isacommand{lemma}\isamarkupfalse%
\ fst{\isacharunderscore}{\kern0pt}max{\isacharunderscore}{\kern0pt}argmax{\isacharunderscore}{\kern0pt}adj{\isacharcolon}{\kern0pt}\isanewline
\ \ \isakeyword{fixes}\ xs\ ys\ rs\ {\isacharcolon}{\kern0pt}{\isacharcolon}{\kern0pt}\ {\isachardoublequoteopen}real\ list{\isachardoublequoteclose}\ \isakeyword{and}\ n\ {\isacharcolon}{\kern0pt}{\isacharcolon}{\kern0pt}\ nat\isanewline
\ \ \isakeyword{assumes}\ {\isachardoublequoteopen}length\ xs\ {\isacharequal}{\kern0pt}\ n{\isachardoublequoteclose}\ \isakeyword{and}\ {\isachardoublequoteopen}length\ ys\ {\isacharequal}{\kern0pt}\ n{\isachardoublequoteclose}\  \isakeyword{and}\ {\isachardoublequoteopen}length\ rs\ {\isacharequal}{\kern0pt}\ n{\isachardoublequoteclose}\ \isakeyword{and}\ adj{\isacharcolon}{\kern0pt}\ {\isachardoublequoteopen}list{\isacharunderscore}{\kern0pt}all{\isadigit{2}}\ {\isacharparenleft}{\kern0pt}{\isasymlambda}\ x\ y{\isachardot}{\kern0pt}\ x\ {\isasymge}\ y\ {\isasymand}\ x\ {\isasymle}\ y\ {\isacharplus}{\kern0pt}\ {\isadigit{1}}{\isacharparenright}{\kern0pt}\ xs\ ys{\isachardoublequoteclose}\isanewline
\ \ \isakeyword{shows}\ {\isachardoublequoteopen}{\isacharparenleft}{\kern0pt}fst\ {\isacharparenleft}{\kern0pt}max{\isacharunderscore}{\kern0pt}argmax\ {\isacharparenleft}{\kern0pt}map{\isadigit{2}}\ {\isacharparenleft}{\kern0pt}{\isacharplus}{\kern0pt}{\isacharparenright}{\kern0pt}\ xs\ rs{\isacharparenright}{\kern0pt}{\isacharparenright}{\kern0pt}{\isacharparenright}{\kern0pt}\ {\isasymge}\ {\isacharparenleft}{\kern0pt}fst\ {\isacharparenleft}{\kern0pt}max{\isacharunderscore}{\kern0pt}argmax\ {\isacharparenleft}{\kern0pt}map{\isadigit{2}}\ {\isacharparenleft}{\kern0pt}{\isacharplus}{\kern0pt}{\isacharparenright}{\kern0pt}\ ys\ rs{\isacharparenright}{\kern0pt}{\isacharparenright}{\kern0pt}{\isacharparenright}{\kern0pt}\ {\isasymand}\ {\isacharparenleft}{\kern0pt}fst\ {\isacharparenleft}{\kern0pt}max{\isacharunderscore}{\kern0pt}argmax\ {\isacharparenleft}{\kern0pt}map{\isadigit{2}}\ {\isacharparenleft}{\kern0pt}{\isacharplus}{\kern0pt}{\isacharparenright}{\kern0pt}\ xs\ rs{\isacharparenright}{\kern0pt}{\isacharparenright}{\kern0pt}{\isacharparenright}{\kern0pt}\ {\isasymle}\ {\isacharparenleft}{\kern0pt}fst\ {\isacharparenleft}{\kern0pt}max{\isacharunderscore}{\kern0pt}argmax\ {\isacharparenleft}{\kern0pt}map{\isadigit{2}}\ {\isacharparenleft}{\kern0pt}{\isacharplus}{\kern0pt}{\isacharparenright}{\kern0pt}\ ys\ rs{\isacharparenright}{\kern0pt}{\isacharparenright}{\kern0pt}{\isacharparenright}{\kern0pt}\ {\isacharplus}{\kern0pt}\ {\isadigit{1}}{\isachardoublequoteclose}
\end{isabelle}
\begin{isabelle}
\isacommand{lemma}\isamarkupfalse%
\ argmax{\isacharunderscore}{\kern0pt}insert{\isacharunderscore}{\kern0pt}i{\isacharunderscore}{\kern0pt}i{\isacharprime}{\kern0pt}{\isacharcolon}{\kern0pt}\isanewline
\ \ \isakeyword{assumes}\ {\isachardoublequoteopen}m\ {\isasymle}\ n{\isachardoublequoteclose}\ \isakeyword{and}\ {\isachardoublequoteopen}length\ xs\ {\isacharequal}{\kern0pt}\ n{\isachardoublequoteclose}\isanewline
\ \ \isakeyword{shows}\ {\isachardoublequoteopen}{\isacharparenleft}{\kern0pt}argmax{\isacharunderscore}{\kern0pt}insert\ k\ xs\ m\ {\isacharequal}{\kern0pt}\ m{\isacharparenright}{\kern0pt}\ {\isasymlongleftrightarrow}\ {\isacharparenleft}{\kern0pt}ereal\ k\ {\isasymge}\ {\isacharparenleft}{\kern0pt}fst\ {\isacharparenleft}{\kern0pt}max{\isacharunderscore}{\kern0pt}argmax\ xs{\isacharparenright}{\kern0pt}{\isacharparenright}{\kern0pt}\ {\isasymand}\ {\isacharparenleft}{\kern0pt}ereal\ k\ {\isasymnoteq}\ {\isacharparenleft}{\kern0pt}fst\ {\isacharparenleft}{\kern0pt}max{\isacharunderscore}{\kern0pt}argmax\ {\isacharparenleft}{\kern0pt}drop\ m\ xs{\isacharparenright}{\kern0pt}{\isacharparenright}{\kern0pt}{\isacharparenright}{\kern0pt}{\isacharparenright}{\kern0pt}{\isacharparenright}{\kern0pt}{\isachardoublequoteclose}
\end{isabelle}
\begin{isabelle}
\isacommand{lemma}\isamarkupfalse%
\ DP{\isacharunderscore}{\kern0pt}RNM{\isacharunderscore}{\kern0pt}M{\isacharunderscore}{\kern0pt}i{\isacharcolon}{\kern0pt}\isanewline
\ \ \isakeyword{fixes}\ xs\ ys\ rs\ {\isacharcolon}{\kern0pt}{\isacharcolon}{\kern0pt}\ {\isachardoublequoteopen}real\ list{\isachardoublequoteclose}\ \isakeyword{and}\ x\ y\ {\isacharcolon}{\kern0pt}{\isacharcolon}{\kern0pt}\ real\ \isakeyword{and}\ i\ n\ {\isacharcolon}{\kern0pt}{\isacharcolon}{\kern0pt}\ nat\isanewline
\ \ \isakeyword{assumes}\ {\isachardoublequoteopen}length\ xs\ {\isacharequal}{\kern0pt}\ n{\isachardoublequoteclose}\ \isakeyword{and}\ {\isachardoublequoteopen}length\ ys\ {\isacharequal}{\kern0pt}\ n{\isachardoublequoteclose}\ \isakeyword{and}\ {\isachardoublequoteopen}length\ rs\ {\isacharequal}{\kern0pt}\ n{\isachardoublequoteclose}\isanewline
\ \ \ \ \isakeyword{and}\ adj{\isacharprime}{\kern0pt}{\isacharcolon}{\kern0pt}\ {\isachardoublequoteopen}x\ {\isasymge}\ y\ {\isasymand}\ x\ {\isasymle}\ y\ {\isacharplus}{\kern0pt}\ {\isadigit{1}}{\isachardoublequoteclose}\  \isakeyword{and}\ adj{\isacharcolon}{\kern0pt}\ {\isachardoublequoteopen}list{\isacharunderscore}{\kern0pt}all{\isadigit{2}}\ {\isacharparenleft}{\kern0pt}{\isasymlambda}\ x\ y{\isachardot}{\kern0pt}\ x\ {\isasymge}\ y\ {\isasymand}\ x\ {\isasymle}\ y\ {\isacharplus}{\kern0pt}\ {\isadigit{1}}{\isacharparenright}{\kern0pt}\ xs\ ys{\isachardoublequoteclose}\ \isakeyword{and}\ {\isachardoublequoteopen}i\ {\isasymle}\ n{\isachardoublequoteclose}\isanewline
\ \ \isakeyword{shows}\ {\isachardoublequoteopen}{\isasymP}{\isacharparenleft}{\kern0pt}j\ in\ {\isacharparenleft}{\kern0pt}RNM{\isacharunderscore}{\kern0pt}M\ xs\ rs\ x\ i{\isacharparenright}{\kern0pt}{\isachardot}{\kern0pt}\ j\ {\isacharequal}{\kern0pt}\ i{\isacharparenright}{\kern0pt}\ {\isasymle}\ {\isacharparenleft}{\kern0pt}exp\ {\isasymepsilon}{\isacharparenright}{\kern0pt}\ {\isacharasterisk}{\kern0pt}\ {\isasymP}{\isacharparenleft}{\kern0pt}j\ in\ {\isacharparenleft}{\kern0pt}RNM{\isacharunderscore}{\kern0pt}M\ ys\ rs\ y\ i{\isacharparenright}{\kern0pt}{\isachardot}{\kern0pt}\ j\ {\isacharequal}{\kern0pt}\ i{\isacharparenright}{\kern0pt}
\isanewline \ \ \ \ \ \ \ \ \ \ \ \ \ \ \  {\isasymand}\ {\isasymP}{\isacharparenleft}{\kern0pt}j\ in\ {\isacharparenleft}{\kern0pt}RNM{\isacharunderscore}{\kern0pt}M\ ys\ rs\ y\ i{\isacharparenright}{\kern0pt}{\isachardot}{\kern0pt}\ j\ {\isacharequal}{\kern0pt}\ i{\isacharparenright}{\kern0pt}\ {\isasymle}\ {\isacharparenleft}{\kern0pt}exp\ {\isasymepsilon}{\isacharparenright}{\kern0pt}\ {\isacharasterisk}{\kern0pt}\ {\isasymP}{\isacharparenleft}{\kern0pt}j\ in\ {\isacharparenleft}{\kern0pt}RNM{\isacharunderscore}{\kern0pt}M\ xs\ rs\ x\ i{\isacharparenright}{\kern0pt}{\isachardot}{\kern0pt}\ j\ {\isacharequal}{\kern0pt}\ i{\isacharparenright}{\kern0pt}{\isachardoublequoteclose}
\end{isabelle}
\begin{isabelle}
\isacommand{lemma}\isamarkupfalse%
\ DP{\isacharunderscore}{\kern0pt}RNM{\isacharprime}{\kern0pt}{\isacharunderscore}{\kern0pt}M{\isacharunderscore}{\kern0pt}i{\isacharcolon}{\kern0pt}\isanewline
\ \ \isakeyword{fixes}\ xs\ ys\ {\isacharcolon}{\kern0pt}{\isacharcolon}{\kern0pt}\ {\isachardoublequoteopen}real\ list{\isachardoublequoteclose}\ \isakeyword{and}\ i\ n\ {\isacharcolon}{\kern0pt}{\isacharcolon}{\kern0pt}\ nat\isanewline
\ \ \isakeyword{assumes}\ lxs{\isacharcolon}{\kern0pt}\ {\isachardoublequoteopen}length\ xs\ {\isacharequal}{\kern0pt}\ n{\isachardoublequoteclose}
\ \isakeyword{and}\ lys{\isacharcolon}{\kern0pt}\ {\isachardoublequoteopen}length\ ys\ {\isacharequal}{\kern0pt}\ n{\isachardoublequoteclose}
\ \isakeyword{and}\ adj{\isacharcolon}{\kern0pt}\ {\isachardoublequoteopen}list{\isacharunderscore}{\kern0pt}all{\isadigit{2}}\ {\isacharparenleft}{\kern0pt}{\isasymlambda}\ x\ y{\isachardot}{\kern0pt}\ x\ {\isasymge}\ y\ {\isasymand}\ x\ {\isasymle}\ y\ {\isacharplus}{\kern0pt}\ {\isadigit{1}}{\isacharparenright}{\kern0pt}\ xs\ ys{\isachardoublequoteclose}\isanewline
\ \ \isakeyword{shows}\ {\isachardoublequoteopen}{\isasymP}{\isacharparenleft}{\kern0pt}j\ in\ {\isacharparenleft}{\kern0pt}RNM{\isacharprime}{\kern0pt}\ xs{\isacharparenright}{\kern0pt}{\isachardot}{\kern0pt}\ j\ {\isacharequal}{\kern0pt}\ i{\isacharparenright}{\kern0pt}\ {\isasymle}\ {\isacharparenleft}{\kern0pt}exp\ {\isasymepsilon}{\isacharparenright}{\kern0pt}\ {\isacharasterisk}{\kern0pt}\ {\isasymP}{\isacharparenleft}{\kern0pt}j\ in\ {\isacharparenleft}{\kern0pt}RNM{\isacharprime}{\kern0pt}\ ys{\isacharparenright}{\kern0pt}{\isachardot}{\kern0pt}\ j\ {\isacharequal}{\kern0pt}\ i{\isacharparenright}{\kern0pt}\ {\isasymand}\ {\isasymP}{\isacharparenleft}{\kern0pt}j\ in\ {\isacharparenleft}{\kern0pt}RNM{\isacharprime}{\kern0pt}\ ys{\isacharparenright}{\kern0pt}{\isachardot}{\kern0pt}\ j\ {\isacharequal}{\kern0pt}\ i{\isacharparenright}{\kern0pt}\ {\isasymle}\ {\isacharparenleft}{\kern0pt}exp\ {\isasymepsilon}{\isacharparenright}{\kern0pt}\ {\isacharasterisk}{\kern0pt}\ {\isasymP}{\isacharparenleft}{\kern0pt}j\ in\ {\isacharparenleft}{\kern0pt}RNM{\isacharprime}{\kern0pt}\ xs{\isacharparenright}{\kern0pt}{\isachardot}{\kern0pt}\ j\ {\isacharequal}{\kern0pt}\ i{\isacharparenright}{\kern0pt}{\isachardoublequoteclose}
\end{isabelle}
\caption{Formalizations of Lemmas \ref{lem:adjacent:counting_query}, \ref{lem:fst_max_argmax_adj}, \ref{lem:argmax_insert_i_i}, \ref{lem:DP_RNM_M_i} and \ref{DP_RNM'_M_i}.}\label{Fig:Lemmas_RNM}
\end{figure*}
In this section, we formalize the differential privacy of the report noisy max mechanism $\mathrm{RNM}_{\vec{q},m,\varepsilon}$, namely Propositions \ref{prop:DP:RNM_naive} and \ref{prop:DP:RNM_true}.
The report noisy max mechanism $\mathrm{RNM}_{\vec{q},m,\varepsilon}$ is implemented in Isabelle/HOL as follows:
\begin{isabelle}
\isacommand{definition}\isamarkupfalse%
\ RNM{\isacharunderscore}{\kern0pt}counting\ {\isacharcolon}{\kern0pt}{\isacharcolon}{\kern0pt}\ {\isachardoublequoteopen}real\ {\isasymRightarrow}\ nat\ list\ {\isasymRightarrow}\ nat\ measure{\isachardoublequoteclose}\ \isakeyword{where}\isanewline
\ \ {\isachardoublequoteopen}RNM{\isacharunderscore}{\kern0pt}counting\ {\isasymepsilon}\ x\ {\isacharequal}{\kern0pt}\ do\ {\isacharbraceleft}{\kern0pt}\isanewline
\ y\ {\isasymleftarrow}\ Lap{\isacharunderscore}{\kern0pt}dist{\isacharunderscore}{\kern0pt}list\ {\isacharparenleft}{\kern0pt}{\isadigit{1}}\ {\isacharslash}{\kern0pt}\ {\isasymepsilon}{\isacharparenright}{\kern0pt}\ {\isacharparenleft}{\kern0pt}counting{\isacharunderscore}{\kern0pt}query\ x{\isacharparenright}{\kern0pt}{\isacharsemicolon}{\kern0pt}\isanewline
\ return\ {\isacharparenleft}{\kern0pt}count{\isacharunderscore}{\kern0pt}space\ UNIV{\isacharparenright}{\kern0pt}\ {\isacharparenleft}{\kern0pt}argmax{\isacharunderscore}{\kern0pt}list\ y{\isacharparenright}{\kern0pt}\isanewline
\ {\isacharbraceright}{\kern0pt}{\isachardoublequoteclose}
\end{isabelle}
Here, \isa{counting{\isacharunderscore}{\kern0pt}query} is an implementation of a tuple of counting queries;
\isa{argmax{\isacharunderscore}{\kern0pt}list} is an implementation of the mapping $(y_0,\ldots,y_{m-1}) \mapsto \argmax_{0\leq i < m} y_i$;
\isa{Lap{\isacharunderscore}{\kern0pt}dist{\isacharunderscore}{\kern0pt}list} is the procedure adding noise sampled from $\mathrm{Lap}(1/\varepsilon)$ given in Section \ref{sec:Lap_dist_list}.
We give formal proofs of its differential privacy.

%\subsection{Argmax of lists}
We check the details of \isa{argmax{\isacharunderscore}{\kern0pt}list}, and formalize Lemmas \ref{lem:fst_max_argmax_adj} and \ref{lem:argmax_insert_i_i}.
First, we define the function \isa{max{\isacharunderscore}{\kern0pt}argmax} returning the pair of max and  argmax of a list\footnote{we assume that the maximum of the empty list is $-\infty$.}.
We then define \isa{argmax{\isacharunderscore}{\kern0pt}list} by taking the second components of pairs.
\begin{isabelle}
\isacommand{primrec}\isamarkupfalse%
\ max{\isacharunderscore}{\kern0pt}argmax\ {\isacharcolon}{\kern0pt}{\isacharcolon}{\kern0pt}\ {\isachardoublequoteopen}real\ list\ {\isasymRightarrow}\ {\isacharparenleft}{\kern0pt}ereal\ {\isasymtimes}\ nat{\isacharparenright}{\kern0pt}{\isachardoublequoteclose}\ \isakeyword{where}\isanewline
\ \ {\isachardoublequoteopen}max{\isacharunderscore}{\kern0pt}argmax\ {\isacharbrackleft}{\kern0pt}{\isacharbrackright}{\kern0pt}\ {\isacharequal}{\kern0pt}\ {\isacharparenleft}{\kern0pt}{\isacharminus}{\kern0pt}{\isasyminfinity}{\isacharcomma}{\kern0pt}{\isadigit{0}}{\isacharparenright}{\kern0pt}{\isachardoublequoteclose}{\isacharbar}{\kern0pt}\isanewline
\ \ {\isachardoublequoteopen}max{\isacharunderscore}{\kern0pt}argmax\ {\isacharparenleft}{\kern0pt}x{\isacharhash}{\kern0pt}xs{\isacharparenright}{\kern0pt}\ {\isacharequal}{\kern0pt}\ {\isacharparenleft}{\kern0pt}let\ {\isacharparenleft}{\kern0pt}m{\isacharcomma}{\kern0pt}\ i{\isacharparenright}{\kern0pt}\ {\isacharequal}{\kern0pt}\ max{\isacharunderscore}{\kern0pt}argmax\ xs\ in\ if\ x\ {\isachargreater}{\kern0pt}\ m\ then\ {\isacharparenleft}{\kern0pt}x{\isacharcomma}{\kern0pt}{\isadigit{0}}{\isacharparenright}{\kern0pt}\ else\ {\isacharparenleft}{\kern0pt}m{\isacharcomma}{\kern0pt}\ Suc\ i{\isacharparenright}{\kern0pt}{\isacharparenright}{\kern0pt}{\isachardoublequoteclose}
\end{isabelle}
\begin{isabelle}
\isacommand{definition}\isamarkupfalse%
\ argmax{\isacharunderscore}{\kern0pt}list\ {\isacharcolon}{\kern0pt}{\isacharcolon}{\kern0pt}\ {\isachardoublequoteopen}real\ list\ {\isasymRightarrow}\ nat{\isachardoublequoteclose}\ \isakeyword{where}\isanewline
\ \ {\isachardoublequoteopen}argmax{\isacharunderscore}{\kern0pt}list\ {\isacharequal}{\kern0pt}\ snd\ o\ max{\isacharunderscore}{\kern0pt}argmax\ {\isachardoublequoteclose}
\end{isabelle}

\subsection{Differential Privacy of $\mathrm{RNM}_{\vec{q},m,\varepsilon}$}
\subsubsection{Naive Evaluation}
We formalize Proposition \ref{prop:DP:RNM_naive}.
% We define the sensitivity of \isa{counting{\isacharunderscore}{\kern0pt}query} using the locale \isa{Lap{\isacharunderscore}{\kern0pt}Mechanism{\isacharunderscore}{\kern0pt}list}.
% \begin{isabelle}
% \isacommand{interpretation}\isamarkupfalse%
% \ Lap{\isacharunderscore}{\kern0pt}Mechanism{\isacharunderscore}{\kern0pt}list\ {\isachardoublequoteopen}{\isacharparenleft}{\kern0pt}listM\ {\isacharparenleft}{\kern0pt}count{\isacharunderscore}{\kern0pt}space\ UNIV{\isacharparenright}{\kern0pt}{\isacharparenright}{\kern0pt}{\isachardoublequoteclose}\ {\isachardoublequoteopen}counting{\isacharunderscore}{\kern0pt}query{\isachardoublequoteclose}\ {\isachardoublequoteopen}adj{\isacharunderscore}{\kern0pt}L{\isadigit{1}}{\isacharunderscore}{\kern0pt}norm{\isachardoublequoteclose}\ m
% \end{isabelle}
% We then check the sensitivity of \isa{counting{\isacharunderscore}{\kern0pt}query}:
% \begin{isabelle}
% \isacommand{lemma}\isamarkupfalse%
% \ sensitvity{\isacharunderscore}{\kern0pt}counting{\isacharunderscore}{\kern0pt}query{\isacharcolon}{\kern0pt}\isanewline
% \ \ \isakeyword{shows}\ {\isachardoublequoteopen}sensitivity\ {\isasymle}\ m{\isachardoublequoteclose}
% \end{isabelle}
Finally, by applying \isa{DP{\isacharunderscore}{\kern0pt}Lap{\isacharunderscore}{\kern0pt}dist{\isacharunderscore}{\kern0pt}list} and the basic properties of differential privacy (formalized in Isabelle/HOL),
we conclude the following formal version of Proposition \ref{prop:DP:RNM_naive}:
\begin{isabelle}
\isacommand{theorem}\isamarkupfalse%
\ Naive{\isacharunderscore}{\kern0pt}differential{\isacharunderscore}{\kern0pt}privacy{\isacharunderscore}{\kern0pt}LapMech{\isacharunderscore}{\kern0pt}RNM{\isacharunderscore}{\kern0pt}AFDP{\isacharcolon}{\kern0pt}\isanewline
\isakeyword{assumes}\ pose{\isacharcolon}{\kern0pt}\ {\isachardoublequoteopen}{\isacharparenleft}{\kern0pt}{\isasymepsilon}{\isacharcolon}{\kern0pt}{\isacharcolon}{\kern0pt}real{\isacharparenright}{\kern0pt}\ {\isachargreater}{\kern0pt}\ {\isadigit{0}}{\isachardoublequoteclose}\isanewline
\isakeyword{shows}\ {\isachardoublequoteopen}differential{\isacharunderscore}{\kern0pt}privacy{\isacharunderscore}{\kern0pt}AFDP\ {\isacharparenleft}{\kern0pt}RNM{\isacharunderscore}{\kern0pt}counting\ {\isasymepsilon}{\isacharparenright}{\kern0pt}{\isacharparenleft}{\kern0pt}real\ {\isacharparenleft}{\kern0pt}m\ {\isacharasterisk}{\kern0pt}\ {\isasymepsilon}{\isacharparenright}{\kern0pt}{\isacharparenright}{\kern0pt}\ {\isadigit{0}}{\isachardoublequoteclose}
\end{isabelle}

\subsubsection{Finar Evaluation}
We formalize Proposition \ref{prop:DP:RNM_true}.
It suffices to interpret the locale \isa{Lap{\isacharunderscore}{\kern0pt}Mechanism{\isacharunderscore}{\kern0pt}RNM{\isacharunderscore}{\kern0pt}mainpart} with \isa{counting{\isacharunderscore}{\kern0pt}query}. 
%We instantiate $X = \Nat^{|\mathcal{X}|}$, $R^\mathrm{adj} = \{ (D,D') ~|~ \|D  - D' \|  \leq 1\}$ and $c = \vec{q}$ to the proof of differential privacy of $\mathrm{RNM'}_{m,\varepsilon} \circ c$.
%We go back to the locale \isa{ Lap{\isacharunderscore}{\kern0pt}Mechanism{\isacharunderscore}{\kern0pt}RNM{\isacharunderscore}{\kern0pt}counting},
%and prove the differential privacy of \isa{RNM{\isacharunderscore}{\kern0pt}counting}.
The remaining task is formalizing Lemma \ref{lem:adjacent:counting_query}.
It is given as \isa{ finer{\isacharunderscore}{\kern0pt}sensitivity{\isacharunderscore}{\kern0pt}counting{\isacharunderscore}{\kern0pt}query} in Figure \ref{Fig:Lemmas_RNM}.

%\begin{isabelle}
%\isacommand{lemma}\isamarkupfalse%
%\ finer{\isacharunderscore}{\kern0pt}sensitivity{\isacharunderscore}{\kern0pt}counting{\isacharunderscore}{\kern0pt}query{\isacharcolon}{\kern0pt}\isanewline
%\ \ \isakeyword{assumes}\ {\isachardoublequoteopen}{\isacharparenleft}{\kern0pt}xs{\isacharcomma}{\kern0pt}ys{\isacharparenright}{\kern0pt}\ {\isasymin}\ adj{\isacharunderscore}{\kern0pt}L{\isadigit{1}}{\isacharunderscore}{\kern0pt}norm{\isachardoublequoteclose}\isanewline
%\ \ \isakeyword{shows}\ {\isachardoublequoteopen}list{\isacharunderscore}{\kern0pt}all{\isadigit{2}}\ {\isacharparenleft}{\kern0pt}{\isasymlambda}\ x\ y{\isachardot}{\kern0pt}\ x\ {\isasymge}\ y\ {\isasymand}\ x\ {\isasymle}\ y\ {\isacharplus}{\kern0pt}\ {\isadigit{1}}{\isacharparenright}{\kern0pt}\ {\isacharparenleft}{\kern0pt}counting{\isacharunderscore}{\kern0pt}query\ xs{\isacharparenright}{\kern0pt}\ {\isacharparenleft}{\kern0pt}counting{\isacharunderscore}{\kern0pt}query\ ys{\isacharparenright}{\kern0pt}\ {\isasymor}\ list{\isacharunderscore}{\kern0pt}all{\isadigit{2}}\ {\isacharparenleft}{\kern0pt}{\isasymlambda}\ x\ y{\isachardot}{\kern0pt}\ x\ {\isasymge}\ y\ {\isasymand}\ x\ {\isasymle}\ y\ {\isacharplus}{\kern0pt}\ {\isadigit{1}}{\isacharparenright}{\kern0pt}\ {\isacharparenleft}{\kern0pt}counting{\isacharunderscore}{\kern0pt}query\ ys{\isacharparenright}{\kern0pt}\ {\isacharparenleft}{\kern0pt}counting{\isacharunderscore}{\kern0pt}query\ xs{\isacharparenright}{\kern0pt}{\isachardoublequoteclose}
%\end{isabelle}
Combining the proofs on the main body, we conclude the formal version of Proposition \ref{prop:DP:RNM_true}:
\begin{isabelle}
\isacommand{theorem}\isamarkupfalse%
\ differential{\isacharunderscore}{\kern0pt}privacy{\isacharunderscore}{\kern0pt}LapMech{\isacharunderscore}{\kern0pt}RNM{\isacharunderscore}{\kern0pt}AFDP{\isacharcolon}{\kern0pt}\isanewline
\ \ \isakeyword{assumes}\ pose{\isacharcolon}{\kern0pt}\ {\isachardoublequoteopen}{\isacharparenleft}{\kern0pt}{\isasymepsilon}{\isacharcolon}{\kern0pt}{\isacharcolon}{\kern0pt}real{\isacharparenright}{\kern0pt}\ {\isachargreater}{\kern0pt}\ {\isadigit{0}}{\isachardoublequoteclose}\isanewline
\ \ \isakeyword{shows}\ {\isachardoublequoteopen}differential{\isacharunderscore}{\kern0pt}privacy{\isacharunderscore}{\kern0pt}AFDP\ {\isacharparenleft}{\kern0pt}RNM{\isacharunderscore}{\kern0pt}counting\ {\isasymepsilon}{\isacharparenright}{\kern0pt}\ {\isasymepsilon}\ {\isadigit{0}}{\isachardoublequoteclose}
\end{isabelle}

\subsection{Differential Privacy of the Main Body of $\mathrm{RNM}_{\vec{q},m,\varepsilon}$}
We recall $\mathrm{RNM}_{\vec{q},m,\varepsilon} = \mathrm{RNM'}_{m,\varepsilon} \circ \vec{q}$
(equation (\ref{eq:LapMech_RNM_RNM'_c}) in Section \ref{sec:report_noisy_max_DP}).
Based on this, we first prove the differential privacy of $\mathrm{RNM'}_{m,\varepsilon} \circ c$ for general $X$, $R^{\mathrm{adj}}$ and $c \colon X \to \Real^m$
satisfying similar statement as Lemma \ref{lem:adjacent:counting_query} (intuitively ``(A) or (B) holds'').
Later, we instantiate them as expected.

For proof, we introduce the following locale.
\begin{isabelle}
\isacommand{locale}\isamarkupfalse%
\ Lap{\isacharunderscore}{\kern0pt}Mechanism{\isacharunderscore}{\kern0pt}RNM{\isacharunderscore}{\kern0pt}mainpart\ {\isacharequal}{\kern0pt}\isanewline
\ \ \isakeyword{fixes}\ M{\isacharcolon}{\kern0pt}{\isacharcolon}{\kern0pt}{\isachardoublequoteopen}{\isacharprime}{\kern0pt}a\ measure{\isachardoublequoteclose}\isanewline
\ \ \ \ \isakeyword{and}\ adj{\isacharcolon}{\kern0pt}{\isacharcolon}{\kern0pt}{\isachardoublequoteopen}{\isacharprime}{\kern0pt}a\ rel{\isachardoublequoteclose}\isanewline
\ \ \ \ \isakeyword{and}\ c{\isacharcolon}{\kern0pt}{\isacharcolon}{\kern0pt}{\isachardoublequoteopen}{\isacharprime}{\kern0pt}a\ {\isasymRightarrow}\ real\ list{\isachardoublequoteclose}\isanewline
\ \ \isakeyword{assumes}\ c{\isacharcolon}{\kern0pt}\ {\isachardoublequoteopen}c\ {\isasymin}\ M\ {\isasymrightarrow}\isactrlsub M\ listM\ borel{\isachardoublequoteclose}\isanewline
\ \ \ \ \isakeyword{and}\ cond{\isacharcolon}{\kern0pt}\ {\isachardoublequoteopen}{\isasymforall}\ {\isacharparenleft}{\kern0pt}x{\isacharcomma}{\kern0pt}y{\isacharparenright}{\kern0pt}\ {\isasymin}\ adj{\isachardot}{\kern0pt}\ list{\isacharunderscore}{\kern0pt}all{\isadigit{2}}\ {\isacharparenleft}{\kern0pt}{\isasymlambda}x\ y{\isachardot}{\kern0pt}\ y\ {\isasymle}\ x\ {\isasymand}\ x\ {\isasymle}\ y\ {\isacharplus}{\kern0pt}\ {\isadigit{1}}{\isacharparenright}{\kern0pt}\ {\isacharparenleft}{\kern0pt}c\ x{\isacharparenright}{\kern0pt}\ {\isacharparenleft}{\kern0pt}c\ y{\isacharparenright}{\kern0pt}\ {\isasymor}\ list{\isacharunderscore}{\kern0pt}all{\isadigit{2}}\ {\isacharparenleft}{\kern0pt}{\isasymlambda}x\ y{\isachardot}{\kern0pt}\ y\ {\isasymle}\ x\ {\isasymand}\ x\ {\isasymle}\ y\ {\isacharplus}{\kern0pt}\ {\isadigit{1}}{\isacharparenright}{\kern0pt}\ {\isacharparenleft}{\kern0pt}c\ y{\isacharparenright}{\kern0pt}\ {\isacharparenleft}{\kern0pt}c\ x{\isacharparenright}{\kern0pt}{\isachardoublequoteclose}\isanewline
\ \ \ \ \isakeyword{and}\ adj{\isacharcolon}{\kern0pt}\ {\isachardoublequoteopen}adj\ {\isasymsubseteq}\ space\ M\ {\isasymtimes}\ space\ M{\isachardoublequoteclose}\isanewline
\isakeyword{begin}%
\end{isabelle}
\begin{isabelle}
\isacommand{definition}\isamarkupfalse%
\ LapMech{\isacharunderscore}{\kern0pt}RNM\ {\isacharcolon}{\kern0pt}{\isacharcolon}{\kern0pt}\ {\isachardoublequoteopen}real\ {\isasymRightarrow}\ {\isacharprime}{\kern0pt}a\ {\isasymRightarrow}\ nat\ measure{\isachardoublequoteclose}\ \isakeyword{where}\isanewline
\ \ {\isachardoublequoteopen}LapMech{\isacharunderscore}{\kern0pt}RNM\ {\isasymepsilon}\ x\ {\isacharequal}{\kern0pt}\ do\ {\isacharbraceleft}{\kern0pt}y\ {\isasymleftarrow}\ Lap{\isacharunderscore}{\kern0pt}dist{\isacharunderscore}{\kern0pt}list\ {\isacharparenleft}{\kern0pt}{\isadigit{1}}\ {\isacharslash}{\kern0pt}\ {\isasymepsilon}{\isacharparenright}{\kern0pt}\ {\isacharparenleft}{\kern0pt}c\ x{\isacharparenright}{\kern0pt}{\isacharsemicolon}{\kern0pt}\ return\ {\isacharparenleft}{\kern0pt}count{\isacharunderscore}{\kern0pt}space\ UNIV{\isacharparenright}{\kern0pt}\ {\isacharparenleft}{\kern0pt}argmax{\isacharunderscore}{\kern0pt}list\ y{\isacharparenright}{\kern0pt}{\isacharbraceright}{\kern0pt}{\isachardoublequoteclose}
\end{isabelle}
%
%\begin{isabelle}
%\isacommand{lemma}\isamarkupfalse%
%\ measurable{\isacharunderscore}{\kern0pt}LapMech{\isacharunderscore}{\kern0pt}RNM{\isacharbrackleft}{\kern0pt}measurable{\isacharbrackright}{\kern0pt}{\isacharcolon}{\kern0pt}\isanewline
%\ \ \isakeyword{shows}\ {\isachardoublequoteopen}LapMech{\isacharunderscore}{\kern0pt}RNM\ {\isasymepsilon}\ {\isasymin}\ M\ {\isasymrightarrow}\isactrlsub M\ prob{\isacharunderscore}{\kern0pt}algebra{\isacharparenleft}{\kern0pt}count{\isacharunderscore}{\kern0pt}space\ UNIV{\isacharparenright}{\kern0pt}{\isachardoublequoteclose}
%\end{isabelle}

We formalize Lemmas \ref{lem:DP_RNM_M_i} and \ref{DP_RNM'_M_i}, and
prove the differential privacy of \isa{LapMech{\isacharunderscore}{\kern0pt}RNM}. 
We set the following context, and focus on the main body \isa{RNM{\isacharprime}}.

\begin{isabelle}
\isacommand{context}\isamarkupfalse%
\isanewline
\ \ \isakeyword{fixes}\ {\isasymepsilon}{\isacharcolon}{\kern0pt}{\isacharcolon}{\kern0pt}real
\isanewline
\ \ \isakeyword{assumes}\ pose{\isacharcolon}{\kern0pt}{\isachardoublequoteopen}{\isadigit{0}}\ {\isacharless}{\kern0pt}\ {\isasymepsilon}{\isachardoublequoteclose}\isanewline
\isakeyword{begin}%
\end{isabelle}
%\tetsuya{TODO: rebuild the definition of RNM'}
\begin{isabelle}
\isacommand{definition}\isamarkupfalse%
\ RNM{\isacharprime}{\kern0pt}\ {\isacharcolon}{\kern0pt}{\isacharcolon}{\kern0pt}\ {\isachardoublequoteopen}real\ list\ {\isasymRightarrow}\ nat\ measure{\isachardoublequoteclose}\ \isakeyword{where}\isanewline
\ \ {\isachardoublequoteopen}RNM{\isacharprime}{\kern0pt}\ zs\ {\isacharequal}{\kern0pt}\ do\ {\isacharbraceleft}{\kern0pt}y\ {\isasymleftarrow}\ Lap{\isacharunderscore}{\kern0pt}dist{\isacharunderscore}{\kern0pt}list\ {\isacharparenleft}{\kern0pt}{\isadigit{1}}\ {\isacharslash}{\kern0pt}\ {\isasymepsilon}{\isacharparenright}{\kern0pt}\ {\isacharparenleft}{\kern0pt}zs{\isacharparenright}{\kern0pt}{\isacharsemicolon}{\kern0pt}\ return\ {\isacharparenleft}{\kern0pt}count{\isacharunderscore}{\kern0pt}space\ UNIV{\isacharparenright}{\kern0pt}\ {\isacharparenleft}{\kern0pt}argmax{\isacharunderscore}{\kern0pt}list\ y{\isacharparenright}{\kern0pt}{\isacharbraceright}{\kern0pt}{\isachardoublequoteclose}
\end{isabelle}

%\begin{isabelle}
%\isacommand{lemma}\isamarkupfalse%
%\ measurable{\isacharunderscore}{\kern0pt}RNM{\isacharprime}{\kern0pt}{\isacharbrackleft}{\kern0pt}measurable{\isacharbrackright}{\kern0pt}{\isacharcolon}{\kern0pt}\isanewline
%\ \ \isakeyword{shows}\ {\isachardoublequoteopen}RNM{\isacharprime}{\kern0pt}\ {\isasymin}\ listM\ borel\ {\isasymrightarrow}\isactrlsub M\ prob{\isacharunderscore}{\kern0pt}algebra{\isacharparenleft}{\kern0pt}count{\isacharunderscore}{\kern0pt}space\ UNIV{\isacharparenright}{\kern0pt}{\isachardoublequoteclose}
%\end{isabelle}

We formalize $p_i$ (see equation (\ref{eq:p_i})) as follows:
\begin{isabelle}
\isacommand{definition}\isamarkupfalse%
\ RNM{\isacharunderscore}{\kern0pt}M\ {\isacharcolon}{\kern0pt}{\isacharcolon}{\kern0pt}\ {\isachardoublequoteopen}real\ list\ {\isasymRightarrow}\ real\ list\ {\isasymRightarrow}\ real\ {\isasymRightarrow}\ nat\ {\isasymRightarrow}\ nat\ measure{\isachardoublequoteclose}\ \isakeyword{where}\isanewline
\ \ {\isachardoublequoteopen}RNM{\isacharunderscore}{\kern0pt}M\ cs\ rs\ d\ i\ {\isacharequal}{\kern0pt}\ do{\isacharbraceleft}{\kern0pt}r\ {\isasymleftarrow}\ {\isacharparenleft}{\kern0pt}Lap{\isacharunderscore}{\kern0pt}dist{\isadigit{0}}\ {\isacharparenleft}{\kern0pt}{\isadigit{1}}{\isacharslash}{\kern0pt}{\isasymepsilon}{\isacharparenright}{\kern0pt}{\isacharparenright}{\kern0pt}{\isacharsemicolon}{\kern0pt}\ return\ {\isacharparenleft}{\kern0pt}count{\isacharunderscore}{\kern0pt}space\ UNIV{\isacharparenright}{\kern0pt}\ {\isacharparenleft}{\kern0pt}argmax{\isacharunderscore}{\kern0pt}insert\ {\isacharparenleft}{\kern0pt}r{\isacharplus}{\kern0pt}d{\isacharparenright}{\kern0pt}\ {\isacharparenleft}{\kern0pt}{\isacharparenleft}{\kern0pt}{\isasymlambda}\ {\isacharparenleft}{\kern0pt}xs{\isacharcomma}{\kern0pt}ys{\isacharparenright}{\kern0pt}{\isachardot}{\kern0pt}\ {\isacharparenleft}{\kern0pt}map{\isadigit{2}}\ {\isacharparenleft}{\kern0pt}{\isacharplus}{\kern0pt}{\isacharparenright}{\kern0pt}\ xs\ ys{\isacharparenright}{\kern0pt}{\isacharparenright}{\kern0pt}{\isacharparenleft}{\kern0pt}cs{\isacharcomma}{\kern0pt}\ rs{\isacharparenright}{\kern0pt}{\isacharparenright}{\kern0pt}\ i{\isacharparenright}{\kern0pt}{\isacharbraceright}{\kern0pt}{\isachardoublequoteclose}
\end{isabelle}
We here relate variables \isa{r}, \isa{d}, \isa{rs} and \isa{cs} with
$r_i$, $c_i$, $\vec{r}_{- i}$ and $\vec{c}_{- i}$ respectively.

Next, we define the function that combining argmax and insertion of elements.
\begin{isabelle}
\isacommand{definition}\isamarkupfalse%
\ argmax{\isacharunderscore}{\kern0pt}insert\ {\isacharcolon}{\kern0pt}{\isacharcolon}{\kern0pt}\ {\isachardoublequoteopen}real\ {\isasymRightarrow}\ real\ list\ {\isasymRightarrow}\ nat\ {\isasymRightarrow}\ nat{\isachardoublequoteclose}\ \isakeyword{where}\isanewline
\ \ {\isachardoublequoteopen}argmax{\isacharunderscore}{\kern0pt}insert\ k\ ks\ i\ {\isacharequal}{\kern0pt}\ argmax{\isacharunderscore}{\kern0pt}list\ {\isacharparenleft}{\kern0pt}list{\isacharunderscore}{\kern0pt}insert\ k\ ks\ i{\isacharparenright}{\kern0pt}{\isachardoublequoteclose}
%\isacommand{definition}\isamarkupfalse%
%\ argmax{\isacharunderscore}{\kern0pt}insert\ {\isacharcolon}{\kern0pt}{\isacharcolon}{\kern0pt}\ {\isachardoublequoteopen}real\ {\isasymRightarrow}\ real\ list\ {\isasymRightarrow}\ nat\ {\isasymRightarrow}\ nat{\isachardoublequoteclose}\ \isakeyword{where}\isanewline
%\ \ {\isachardoublequoteopen}argmax{\isacharunderscore}{\kern0pt}insert\ k\ ks\ i\ {\isacharequal}{\kern0pt}\ argmax{\isacharunderscore}{\kern0pt}list\ {\isacharparenleft}{\kern0pt}list{\isacharunderscore}{\kern0pt}insert\ k\ ks\ i{\isacharparenright}{\kern0pt}{\isachardoublequoteclose}
\end{isabelle}
%\ \ {\isachardoublequoteopen}argmax{\isacharunderscore}{\kern0pt}insert\ k\ ks\ i\ {\isacharequal}{\kern0pt}\ argmax{\isacharunderscore}{\kern0pt}list\ {\isacharparenleft}{\kern0pt}take\ i\ ks\ {\isacharat}{\kern0pt}\ {\isacharbrackleft}{\kern0pt}k{\isacharbrackright}{\kern0pt}\ {\isacharat}{\kern0pt}\ drop\ i\ ks{\isacharparenright}{\kern0pt}{\isachardoublequoteclose}
Here, \isa{list{\isacharunderscore}{\kern0pt}insert\ k\ ks\ i = {\isacharparenleft}{\kern0pt}take\ i\ ks\ {\isacharat}{\kern0pt}\ {\isacharbrackleft}{\kern0pt}k{\isacharbrackright}{\kern0pt}\ {\isacharat}{\kern0pt}\ drop\ i\ ks{\isacharparenright}} 
is the list made by inserting the element \isa{k} to the list \isa{ks} at the \isa{i}-th position.
When \isa{k} and \isa{xs} correspond to $d_i$ and $\vec{d}_{- i}$ respectively, 
\isa{argmax{\isacharunderscore}{\kern0pt}insert\ k\ ks\ i} corresponds to ${\textstyle \argmax_{0 \leq j < m}} d_j$.
We then formalize Lemma \ref{lem:argmax_insert_i_i} as \isa{argmax{\isacharunderscore}{\kern0pt}insert{\isacharunderscore}{\kern0pt}i{\isacharunderscore}{\kern0pt}i{\isacharprime}} in Figure \ref{Fig:Lemmas_RNM}. 

Then, \isa{RNM{\isacharunderscore}{\kern0pt}M\ cs\ rs\ d\ i} is an implementation of
the following probability distribution:
\[
\{r_i \leftarrow \mathrm{Lap}(1/\varepsilon);\return {\textstyle \argmax_{0 \leq j < m}} (c_j + r_j) \} \in \prob(\Nat).
\]
Then, $p_i$ is implemented as \isa{{\isasymP}{\isacharparenleft}{\kern0pt}j\ in\ {\isacharparenleft}{\kern0pt}RNM{\isacharunderscore}{\kern0pt}M\ xs\ rs\ x\ i{\isacharparenright}{\kern0pt}{\isachardot}{\kern0pt}\ j\ {\isacharequal}{\kern0pt}\ i{\isacharparenright}{\kern0pt}}.
We then formalize Lemma \ref{lem:DP_RNM_M_i}  as \isa{DP{\isacharunderscore}{\kern0pt}RNM{\isacharunderscore}{\kern0pt}M{\isacharunderscore}{\kern0pt}i} in Figure \ref{Fig:Lemmas_RNM}.

Next, we formalize Lemma \ref{DP_RNM'_M_i} as \isa{DP{\isacharunderscore}{\kern0pt}RNM{\isacharprime}{\kern0pt}{\isacharunderscore}{\kern0pt}M{\isacharunderscore}{\kern0pt}i} in Figure \ref{Fig:Lemmas_RNM}.
If $1 < m$, we formalize equation (\ref{eq:RNM_expand}):
\begin{isabelle}
\isacommand{lemma}\isamarkupfalse%
\ RNM{\isacharprime}{\kern0pt}{\isacharunderscore}{\kern0pt}expand{\isacharcolon}{\kern0pt}\isanewline
\ \ \isakeyword{fixes}\ n\ {\isacharcolon}{\kern0pt}{\isacharcolon}{\kern0pt}\ nat\isanewline
\ \ \isakeyword{assumes}\ {\isachardoublequoteopen}length\ xs\ {\isacharequal}{\kern0pt}\ n{\isachardoublequoteclose}\ \isakeyword{and}\ {\isachardoublequoteopen}i\ {\isasymle}\ n{\isachardoublequoteclose}\isanewline
\ \ \isakeyword{shows}\ {\isachardoublequoteopen}{\isacharparenleft}{\kern0pt}RNM{\isacharprime}{\kern0pt}\ {\isacharparenleft}{\kern0pt}list{\isacharunderscore}{\kern0pt}insert\ x\ xs\ i{\isacharparenright}{\kern0pt}{\isacharparenright}{\kern0pt}\ {\isacharequal}{\kern0pt}\ do{\isacharbraceleft}{\kern0pt}rs\ {\isasymleftarrow}\ {\isacharparenleft}{\kern0pt}Lap{\isacharunderscore}{\kern0pt}dist{\isadigit{0}}{\isacharunderscore}{\kern0pt}list\ {\isacharparenleft}{\kern0pt}{\isadigit{1}}\ {\isacharslash}{\kern0pt}\ {\isasymepsilon}{\isacharparenright}{\kern0pt}\ {\isacharparenleft}{\kern0pt}length\ xs{\isacharparenright}{\kern0pt}{\isacharparenright}{\kern0pt}{\isacharsemicolon}{\kern0pt}\ {\isacharparenleft}{\kern0pt}RNM{\isacharunderscore}{\kern0pt}M\ xs\ rs\ x\ i{\isacharparenright}{\kern0pt}{\isacharbraceright}{\kern0pt}{\isachardoublequoteclose}
\end{isabelle}
If $m = 0,1$, we show \isa{RNM' xs = return (count\_space UNIV) 0} directly.
For such corner cases,  we need to give extra lemmas.

Finally, we conclude the ``differential privacy'' of the main body $ \mathrm{RNM'}_{m,\varepsilon}$.
We temporally use an asymmetric relation corresponding to (A) in Lemma \ref{lem:adjacent:counting_query}.
\begin{isabelle}
%\isacommand{lemma}\isamarkupfalse%
%\ DP{\isacharunderscore}{\kern0pt}divergence{\isacharunderscore}{\kern0pt}RNM{\isacharprime}{\kern0pt}{\isacharcolon}{\kern0pt}\isanewline
%\ \ \isakeyword{fixes}\ xs\ ys\ {\isacharcolon}{\kern0pt}{\isacharcolon}{\kern0pt}\ {\isachardoublequoteopen}real\ list{\isachardoublequoteclose}\isanewline
%\ \ \isakeyword{assumes}\ adj{\isacharcolon}{\kern0pt}\ {\isachardoublequoteopen}list{\isacharunderscore}{\kern0pt}all{\isadigit{2}}\ {\isacharparenleft}{\kern0pt}{\isasymlambda}\ x\ y{\isachardot}{\kern0pt}\ x\ {\isasymge}\ y\ {\isasymand}\ x\ {\isasymle}\ y\ {\isacharplus}{\kern0pt}\ {\isadigit{1}}{\isacharparenright}{\kern0pt}\ xs\ ys{\isachardoublequoteclose}\isanewline
%\ \ \isakeyword{shows}\ {\isachardoublequoteopen}DP{\isacharunderscore}{\kern0pt}divergence\ {\isacharparenleft}{\kern0pt}RNM{\isacharprime}{\kern0pt}\ xs{\isacharparenright}{\kern0pt}\ {\isacharparenleft}{\kern0pt}RNM{\isacharprime}{\kern0pt}\ ys{\isacharparenright}{\kern0pt}\ {\isasymepsilon}\ {\isasymle}\ {\isadigit{0}}\ {\isasymand}\ DP{\isacharunderscore}{\kern0pt}divergence\ {\isacharparenleft}{\kern0pt}RNM{\isacharprime}{\kern0pt}\ ys{\isacharparenright}{\kern0pt}\ {\isacharparenleft}{\kern0pt}RNM{\isacharprime}{\kern0pt}\ xs{\isacharparenright}{\kern0pt}\ {\isasymepsilon}\ {\isasymle}\ {\isadigit{0}}{\isachardoublequoteclose}
\isacommand{lemma}\isamarkupfalse%
\ differential{\isacharunderscore}{\kern0pt}privacy{\isacharunderscore}{\kern0pt}LapMech{\isacharunderscore}{\kern0pt}RNM{\isacharprime}{\kern0pt}{\isacharcolon}{\kern0pt}\isanewline
\ \ \isakeyword{shows}\ {\isachardoublequoteopen}differential{\isacharunderscore}{\kern0pt}privacy\ RNM{\isacharprime}{\kern0pt}\ {\isacharbraceleft}{\kern0pt}{\isacharparenleft}{\kern0pt}xs{\isacharcomma}{\kern0pt}\ ys{\isacharparenright}{\kern0pt}\ {\isacharbar}{\kern0pt}\ xs\ ys{\isachardot}{\kern0pt}\ list{\isacharunderscore}{\kern0pt}all{\isadigit{2}}\ {\isacharparenleft}{\kern0pt}{\isasymlambda}\ x\ y{\isachardot}{\kern0pt}\ x\ {\isasymge}\ y\ {\isasymand}\ x\ {\isasymle}\ y\ {\isacharplus}{\kern0pt}\ {\isadigit{1}}{\isacharparenright}{\kern0pt}\ xs\ ys\ {\isacharbraceright}{\kern0pt}\ {\isasymepsilon}\ {\isadigit{0}}{\isachardoublequoteclose}
\end{isabelle}
%By the symmetry between conditions (A) and (B), we obtain 
%\begin{isabelle}
%\isacommand{corollary}\isamarkupfalse%
%\ differential{\isacharunderscore}{\kern0pt}privacy{\isacharunderscore}{\kern0pt}LapMech{\isacharunderscore}{\kern0pt}RNM{\isacharprime}{\kern0pt}{\isacharunderscore}{\kern0pt}sym{\isacharcolon}{\kern0pt}\isanewline
%\ \ \isakeyword{shows}\ {\isachardoublequoteopen}differential{\isacharunderscore}{\kern0pt}privacy\ RNM{\isacharprime}{\kern0pt}\ {\isacharbraceleft}{\kern0pt}{\isacharparenleft}{\kern0pt}xs{\isacharcomma}{\kern0pt}\ ys{\isacharparenright}{\kern0pt}\ {\isacharbar}{\kern0pt}\ xs\ ys{\isachardot}{\kern0pt}\ list{\isacharunderscore}{\kern0pt}all{\isadigit{2}}\ {\isacharparenleft}{\kern0pt}{\isasymlambda}\ x\ y{\isachardot}{\kern0pt}\ x\ {\isasymge}\ y\ {\isasymand}\ x\ {\isasymle}\ y\ {\isacharplus}{\kern0pt}\ {\isadigit{1}}{\isacharparenright}{\kern0pt}\ xs\ ys\ {\isasymor}\ list{\isacharunderscore}{\kern0pt}all{\isadigit{2}}\ {\isacharparenleft}{\kern0pt}{\isasymlambda}\ x\ y{\isachardot}{\kern0pt}\ x\ {\isasymge}\ y\ {\isasymand}\ x\ {\isasymle}\ y\ {\isacharplus}{\kern0pt}\ {\isadigit{1}}{\isacharparenright}{\kern0pt}\ ys\ xs\ {\isacharbraceright}{\kern0pt}\ {\isasymepsilon}\ {\isadigit{0}}{\isachardoublequoteclose}
%\end{isabelle}
From the symmetry between conditions (A) and (B), that asymmetric relation can be extended to the symmetric relation corresponding to ``(A) or (B) holds''.
Then, by the assumption on \isa{c} and \isa{differential{\isacharunderscore}{\kern0pt}privacy{\isacharunderscore}{\kern0pt}preprocessing}, 
we conclude the differential privacy of \isa{LapMech{\isacharunderscore}{\kern0pt}RNM}.
\begin{isabelle}
\isacommand{theorem}\isamarkupfalse%
\ differential{\isacharunderscore}{\kern0pt}privacy{\isacharunderscore}{\kern0pt}LapMech{\isacharunderscore}{\kern0pt}RNM{\isacharcolon}{\kern0pt}\isanewline
\ \ \isakeyword{shows}\ {\isachardoublequoteopen}differential{\isacharunderscore}{\kern0pt}privacy\ {\isacharparenleft}{\kern0pt}LapMech{\isacharunderscore}{\kern0pt}RNM\ {\isasymepsilon}{\isacharparenright}{\kern0pt}\ adj\ {\isasymepsilon}\ {\isadigit{0}}{\isachardoublequoteclose}
\end{isabelle}

\subsection{Instantiation of Counting Queries}
We here check the details of \isa{counting{\isacharunderscore}{\kern0pt}query} to formalize Propositions \ref{prop:DP:RNM_naive} and \ref{prop:DP:RNM_naive}.
We introduce the following locale to implement a tuple $\vec{q} \colon \Nat^{|\mathcal{X}|} \to \Nat^m$ of counting queries.
\begin{isabelle}
\isacommand{locale}\isamarkupfalse%
\ Lap{\isacharunderscore}{\kern0pt}Mechanism{\isacharunderscore}{\kern0pt}RNM{\isacharunderscore}{\kern0pt}counting\ {\isacharequal}{\kern0pt}\isanewline
\ \ \isakeyword{fixes}\ n{\isacharcolon}{\kern0pt}{\isacharcolon}{\kern0pt}nat\ \isakeyword{and}\ m{\isacharcolon}{\kern0pt}{\isacharcolon}{\kern0pt}nat\ \isakeyword{and}\ Q\ {\isacharcolon}{\kern0pt}{\isacharcolon}{\kern0pt}\ {\isachardoublequoteopen}nat\ {\isasymRightarrow}\ nat\ {\isasymRightarrow}\ bool{\isachardoublequoteclose}\ \isanewline
\ \ \isakeyword{assumes}\ {\isachardoublequoteopen}{\isasymAnd}i{\isachardot}{\kern0pt}\ i\ {\isasymin}\ {\isacharbraceleft}{\kern0pt}{\isadigit{0}}{\isachardot}{\kern0pt}{\isachardot}{\kern0pt}{\isacharless}{\kern0pt}m{\isacharbraceright}{\kern0pt}\ {\isasymLongrightarrow}\ {\isacharparenleft}{\kern0pt}Q\ i{\isacharparenright}{\kern0pt}\ {\isasymin}\ UNIV\ {\isasymrightarrow}\ UNIV{\isachardoublequoteclose}\isanewline
\isakeyword{begin}
\end{isabelle}
Here \isa{n} is the number $|\mathcal{X}|$ of data types, and \isa{m} is the number of counting queries.
Each \isa{Q i} is a predicate corresponding to a subset $\mathcal{X}_{q_i}$ of types which counting query $q_i$ counts.  
The assumption is for the totality of each \isa{Q i}.
We also interpret locales \isa{L{\isadigit{1}}{\isacharunderscore}{\kern0pt}norm{\isacharunderscore}{\kern0pt}list} 
and \isa{results{\isacharunderscore}{\kern0pt}AFDP} to instantiate $\Nat^{|\mathcal{X}|}$, adjacency of datasets, and to recall basic results of DP. 
%\begin{isabelle}
%\isacommand{interpretation}\isamarkupfalse%
%\ L{\isadigit{1}}{\isacharunderscore}{\kern0pt}norm{\isacharunderscore}{\kern0pt}list\ {\isachardoublequoteopen}{\isacharparenleft}{\kern0pt}UNIV{\isacharcolon}{\kern0pt}{\isacharcolon}{\kern0pt}nat\ set{\isacharparenright}{\kern0pt}{\isachardoublequoteclose}{\isachardoublequoteopen}{\isacharparenleft}{\kern0pt}{\isasymlambda}\ x\ y{\isachardot}{\kern0pt}\ {\isasymbar}int\ x\ {\isacharminus}{\kern0pt}\ int\ y{\isasymbar}{\isacharparenright}{\kern0pt}{\isachardoublequoteclose}n
%\end{isabelle}

We then implement each counting query $q_i \colon \Nat^\mathcal{X} \to \Nat$.
\begin{isabelle}
\isacommand{primrec}\isamarkupfalse%
\ counting{\isacharprime}{\kern0pt}{\isacharcolon}{\kern0pt}{\isacharcolon}{\kern0pt}{\isachardoublequoteopen}nat\ {\isasymRightarrow}\ nat\ {\isasymRightarrow}\ nat\ list\ {\isasymRightarrow}\ nat{\isachardoublequoteclose}\ \isakeyword{where}\isanewline
\ \ {\isachardoublequoteopen}counting{\isacharprime}{\kern0pt}\ i\ {\isadigit{0}}\ {\isacharunderscore}{\kern0pt}\ {\isacharequal}{\kern0pt}\ {\isadigit{0}}{\isachardoublequoteclose}\ {\isacharbar}{\kern0pt}\isanewline
\ \ {\isachardoublequoteopen}counting{\isacharprime}{\kern0pt}\ i\ {\isacharparenleft}{\kern0pt}Suc\ k{\isacharparenright}{\kern0pt}\ xs\ {\isacharequal}{\kern0pt}\ {\isacharparenleft}{\kern0pt}if\ Q\ i\ k\ then\ {\isacharparenleft}{\kern0pt}nth{\isacharunderscore}{\kern0pt}total\ {\isadigit{0}}\ xs\ k{\isacharparenright}{\kern0pt}\ else\ {\isadigit{0}}{\isacharparenright}{\kern0pt}\ {\isacharplus}{\kern0pt}\ counting{\isacharprime}{\kern0pt}\ i\ k\ xs{\isachardoublequoteclose}
\end{isabelle}
\begin{isabelle}
\isacommand{definition}\isamarkupfalse%
\ counting{\isacharcolon}{\kern0pt}{\isacharcolon}{\kern0pt}{\isachardoublequoteopen}nat\ {\isasymRightarrow}\ nat\ list\ {\isasymRightarrow}\ nat{\isachardoublequoteclose}\ \isakeyword{where}\isanewline
\ \ {\isachardoublequoteopen}counting\ i\ xs\ {\isacharequal}{\kern0pt}\ counting{\isacharprime}{\kern0pt}\ i\ {\isacharparenleft}{\kern0pt}length\ xs{\isacharparenright}{\kern0pt}\ xs{\isachardoublequoteclose}
\end{isabelle}
Here, each counting query $q_k$ is implemented as  \isa{counting k}. 
The function \isa{nth{\isacharunderscore}{\kern0pt}total} is a totalized version of Isabelle/HOL's \isa{nth}, which is more convenient for measurability proofs.

We then implement the tuple $ \vec{q} = (q_0,\ldots,q_{m-1})$ as:
\begin{isabelle}
\isacommand{definition}\isamarkupfalse%
\ counting{\isacharunderscore}{\kern0pt}query{\isacharcolon}{\kern0pt}{\isacharcolon}{\kern0pt}{\isachardoublequoteopen}nat\ list\ {\isasymRightarrow}\ nat\ list{\isachardoublequoteclose}\ \isakeyword{where}\isanewline
\ \ {\isachardoublequoteopen}counting{\isacharunderscore}{\kern0pt}query\ xs\ {\isacharequal}{\kern0pt}\ map\ {\isacharparenleft}{\kern0pt}{\isasymlambda}\ k{\isachardot}{\kern0pt}\ counting\ k\ xs{\isacharparenright}{\kern0pt}\ {\isacharbrackleft}{\kern0pt}{\isadigit{0}}{\isachardot}{\kern0pt}{\isachardot}{\kern0pt}{\isacharless}{\kern0pt}m{\isacharbrackright}{\kern0pt}{\isachardoublequoteclose}
\end{isabelle}

\section{$L_1$-Metric on Lists}\label{sec:list_metric}
%To formalize datasets and tuples, we use lists.
%In this section we provide the formalize $L_1$-norms on lists.
%In previous sections, we have formalized $\Nat^{|\mathcal{X}|}$ and the adjacency datasets
%using the locale \isa{L{\isadigit{1}}{\isacharunderscore}{\kern0pt}norm{\isacharunderscore}{\kern0pt}list}.
%In this section, we show the details of the locale. 
%To formalize the metric ($L_1$-norm) on $\Nat^{|\mathcal{X}|}$, we give its general construction.
In previous sections, we have formalized $\Nat^{|\mathcal{X}|}$ and the adjacency datasets
the locale \isa{L{\isadigit{1}}{\isacharunderscore}{\kern0pt}norm{\isacharunderscore}{\kern0pt}list}.
In this section, we show the details of the locale. 

%To formalize the metric ($L_1$-norm) on $\Nat^{|\mathcal{X}|}$.
For a metric space $(A,d_A)$, we define the metric space $(A^n,d_{A^n})$ where the carrier set $A^n$ is the set of $A$-lists with length $n$, 
and the metric ($L_1$-norm) $d_{A^n}$ on $A^n$ is defined by
\[
d_{A^n}(xs,ys) = \sum_{0 \leq i < n} d_A(xs[i], ys[i]).
\]
We formalize it using the locale \isa{Metric\_space} of metric spaces in the standard Isabelle/HOL library.
We give the following locale for making the metric space $(A^n, d_{A^n})$ from a metric space $(A,d)$ (stored in \isa{Metric\_space}) and length $n$
% for making $(A^n, d_{A^n})$ given metric space $(A,d)$ stored in the locale \isa{Metric\_space} and the length $n$.
\begin{isabelle}
\isacommand{locale}\isamarkupfalse%
\ L{\isadigit{1}}{\isacharunderscore}{\kern0pt}norm{\isacharunderscore}{\kern0pt}list\ {\isacharequal}{\kern0pt}\ Metric{\isacharunderscore}{\kern0pt}space\ {\isacharplus}{\kern0pt}\isanewline
\ \ \isakeyword{fixes}\ n{\isacharcolon}{\kern0pt}{\isacharcolon}{\kern0pt}nat\isanewline
\isakeyword{begin}
\end{isabelle}
\begin{isabelle}
\isacommand{definition}\isamarkupfalse%
\ space{\isacharunderscore}{\kern0pt}L{\isadigit{1}}{\isacharunderscore}{\kern0pt}norm{\isacharunderscore}{\kern0pt}list\ {\isacharcolon}{\kern0pt}{\isacharcolon}{\kern0pt}{\isachardoublequoteopen}{\isacharparenleft}{\kern0pt}{\isacharprime}{\kern0pt}a\ list{\isacharparenright}{\kern0pt}\ set{\isachardoublequoteclose}\ \isakeyword{where}\isanewline
\ \ {\isachardoublequoteopen}space{\isacharunderscore}{\kern0pt}L{\isadigit{1}}{\isacharunderscore}{\kern0pt}norm{\isacharunderscore}{\kern0pt}list\ {\isacharequal}{\kern0pt}\ {\isacharbraceleft}{\kern0pt}xs{\isachardot}{\kern0pt}\ xs{\isasymin}\ lists\ M\ {\isasymand}\ length\ xs\ {\isacharequal}{\kern0pt}\ n{\isacharbraceright}{\kern0pt}{\isachardoublequoteclose}
\end{isabelle}
\begin{isabelle}
\isacommand{definition}\isamarkupfalse%
\ dist{\isacharunderscore}{\kern0pt}L{\isadigit{1}}{\isacharunderscore}{\kern0pt}norm{\isacharunderscore}{\kern0pt}list\ {\isacharcolon}{\kern0pt}{\isacharcolon}{\kern0pt}{\isachardoublequoteopen}{\isacharparenleft}{\kern0pt}{\isacharprime}{\kern0pt}a\ list{\isacharparenright}{\kern0pt}\ {\isasymRightarrow}\ {\isacharparenleft}{\kern0pt}{\isacharprime}{\kern0pt}a\ list{\isacharparenright}{\kern0pt}\ {\isasymRightarrow}\ real{\isachardoublequoteclose}\ \isakeyword{where}\isanewline
\ \ {\isachardoublequoteopen}dist{\isacharunderscore}{\kern0pt}L{\isadigit{1}}{\isacharunderscore}{\kern0pt}norm{\isacharunderscore}{\kern0pt}list\ xs\ ys\ {\isacharequal}{\kern0pt}\ {\isacharparenleft}{\kern0pt}{\isasymSum}\ i{\isasymin}{\isacharbraceleft}{\kern0pt}{\isadigit{1}}{\isachardot}{\kern0pt}{\isachardot}{\kern0pt}n{\isacharbraceright}{\kern0pt}{\isachardot}{\kern0pt}\ d\ {\isacharparenleft}{\kern0pt}nth\ xs\ {\isacharparenleft}{\kern0pt}i{\isacharminus}{\kern0pt}{\isadigit{1}}{\isacharparenright}{\kern0pt}{\isacharparenright}{\kern0pt}\ {\isacharparenleft}{\kern0pt}nth\ ys\ {\isacharparenleft}{\kern0pt}i{\isacharminus}{\kern0pt}{\isadigit{1}}{\isacharparenright}{\kern0pt}{\isacharparenright}{\kern0pt}{\isacharparenright}{\kern0pt}{\isachardoublequoteclose}
\end{isabelle}
%
%\begin{isabelle}
%\isanewline
%\quad (*We prove that {dist{\isacharunderscore}{\kern0pt}L{\isadigit{1}}{\isacharunderscore}{\kern0pt}norm{\isacharunderscore}{\kern0pt}list} is a metric.*)
%\isanewline
%\end{isabelle}
%
\begin{isabelle}
\isacommand{lemma}\isamarkupfalse%
\ Metric{\isacharunderscore}{\kern0pt}L{\isadigit{1}}{\isacharunderscore}{\kern0pt}norm{\isacharunderscore}{\kern0pt}list\ {\isacharcolon}\isanewline\ \ {\isachardoublequoteopen}Metric{\isacharunderscore}{\kern0pt}space\ space{\isacharunderscore}{\kern0pt}L{\isadigit{1}}{\isacharunderscore}{\kern0pt}norm{\isacharunderscore}{\kern0pt}list\ dist{\isacharunderscore}{\kern0pt}L{\isadigit{1}}{\isacharunderscore}{\kern0pt}norm{\isacharunderscore}{\kern0pt}list{\isachardoublequoteclose}
\isanewline
\isacommand{end}
\end{isabelle}
\begin{isabelle}
\isacommand{sublocale}\isamarkupfalse%
\ L{\isadigit{1}}{\isacharunderscore}{\kern0pt}norm{\isacharunderscore}{\kern0pt}list\ {\isasymsubseteq}\ MetL{\isadigit{1}}{\isacharcolon}{\kern0pt}\ Metric{\isacharunderscore}{\kern0pt}space\ {\isachardoublequoteopen}space{\isacharunderscore}{\kern0pt}L{\isadigit{1}}{\isacharunderscore}{\kern0pt}norm{\isacharunderscore}{\kern0pt}list{\isachardoublequoteclose}\ {\isachardoublequoteopen}dist{\isacharunderscore}{\kern0pt}L{\isadigit{1}}{\isacharunderscore}{\kern0pt}norm{\isacharunderscore}{\kern0pt}list{\isachardoublequoteclose}
\end{isabelle}
We formalize the metric space $(\Nat^{|\mathcal{X}|}, \|-\|_1)$ by the following interpretation (we set $(A,d) = (\Nat,|-|)$):
\begin{isabelle}
\isacommand{interpretation}\isamarkupfalse%
\ L{\isadigit{1}}{\isadigit{1}}nat{\isacharcolon}
\isanewline \ \ {\kern0pt}\ L{\isadigit{1}}{\isacharunderscore}{\kern0pt}norm{\isacharunderscore}{\kern0pt}list\ {\isachardoublequoteopen}{\isacharparenleft}UNIV{\isacharcolon}{\kern0pt}{\isacharcolon}{\kern0pt}nat\ set{\isacharparenright}{\isachardoublequoteclose}\ {\isachardoublequoteopen}{\isacharparenleft}{\kern0pt}{\isasymlambda}x\ y{\isachardot}{\kern0pt}\ real{\isacharunderscore}{\kern0pt}of{\isacharunderscore}{\kern0pt}int\ {\isasymbar}int\ x\ {\isacharminus}{\kern0pt}\ int\ y{\isasymbar}{\isacharparenright}{\kern0pt}{\isachardoublequoteclose}\ n
\end{isabelle}
%If we want to use the fact that $\Nat^{|\mathcal{X}|}$ with $L_1$-norm is indeed a metric space, we further interpret, 
%\tetsuya{TODO:update} 
%\begin{isabelle}
%\isacommand{interpretation}\isamarkupfalse%
%\ MetL{\isadigit{1}}{\isadigit{1}}nat{\isacharcolon}{\kern0pt}\ Metric{\isacharunderscore}{\kern0pt}space\ {\isachardoublequoteopen}L{\isadigit{1}}{\isadigit{1}}nat{\isachardot}{\kern0pt}space{\isacharunderscore}{\kern0pt}L{\isadigit{1}}{\isacharunderscore}{\kern0pt}norm{\isacharunderscore}{\kern0pt}list{\isachardoublequoteclose}\ {\isachardoublequoteopen}L{\isadigit{1}}{\isadigit{1}}nat{\isachardot}{\kern0pt}dist{\isacharunderscore}{\kern0pt}L{\isadigit{1}}{\isacharunderscore}{\kern0pt}norm{\isacharunderscore}{\kern0pt}list{\isachardoublequoteclose}
%\end{isabelle}
To formalize the group privacy, we formalize Lemma \ref{lem:adj_k:group}.
%\begin{lemma}
%Let $k\in \Nat$, $D,D' \in \Nat^{|\mathcal{X}|}$. 
%If $\|D - D' \| \leq k$ then we obtain $(D,D') \in R^k$
%where $R = \{(D,D') | D,D'\colon\text{adjacent}\}$.
%\end{lemma}
%It is formalizes in Isabelle/HOL as follows:
\begin{isabelle}
\isacommand{interpretation}\isamarkupfalse%
\ L{\isadigit{1}}{\isadigit{1}}nat{\isadigit{2}}{\isacharcolon}\isanewline\ \ L{\isadigit{1}}{\isacharunderscore}{\kern0pt}norm{\isacharunderscore}{\kern0pt}list\ {\isachardoublequoteopen}{\isacharparenleft}UNIV{\isacharcolon}{\kern0pt}{\isacharcolon}{\kern0pt}nat\ set{\isacharparenright}{\isachardoublequoteclose}{\isachardoublequoteopen}{\isacharparenleft}{\kern0pt}{\isasymlambda}x\ y{\isachardot}{\kern0pt}\ real{\isacharunderscore}{\kern0pt}of{\isacharunderscore}{\kern0pt}int\ {\isasymbar}int\ x\ {\isacharminus}{\kern0pt}\ int\ y{\isasymbar}{\isacharparenright}{\kern0pt}{\isachardoublequoteclose}\ {\isachardoublequoteopen}Suc\ n{\isachardoublequoteclose}
\end{isabelle}
\begin{isabelle}
\isacommand{lemma}\isamarkupfalse%
\ L{\isadigit{1}}{\isacharunderscore}{\kern0pt}adj{\isacharunderscore}{\kern0pt}iterate{\isacharunderscore}{\kern0pt}Cons{\isadigit{1}}{\isacharcolon}{\kern0pt}\isanewline
\ \ \isakeyword{assumes}\ {\isachardoublequoteopen}xs\ {\isasymin}\ L{\isadigit{1}}{\isadigit{1}}nat{\isachardot}{\kern0pt}space{\isacharunderscore}{\kern0pt}L{\isadigit{1}}{\isacharunderscore}{\kern0pt}norm{\isacharunderscore}{\kern0pt}list{\isachardoublequoteclose}\isanewline
\ \ \ \ \isakeyword{and}\ {\isachardoublequoteopen}ys\ {\isasymin}\ L{\isadigit{1}}{\isadigit{1}}nat{\isachardot}{\kern0pt}space{\isacharunderscore}{\kern0pt}L{\isadigit{1}}{\isacharunderscore}{\kern0pt}norm{\isacharunderscore}{\kern0pt}list{\isachardoublequoteclose}\isanewline
\ \ \ \ \isakeyword{and}\ {\isachardoublequoteopen}{\isacharparenleft}{\kern0pt}xs{\isacharcomma}{\kern0pt}\ ys{\isacharparenright}{\kern0pt}\ {\isasymin}\ adj{\isacharunderscore}{\kern0pt}L{\isadigit{1}}{\isadigit{1}}nat\ {\isacharcircum}{\kern0pt}{\isacharcircum}{\kern0pt}\ k{\isachardoublequoteclose}\isanewline
\ \ \isakeyword{shows}\ {\isachardoublequoteopen}{\isacharparenleft}{\kern0pt}x{\isacharhash}{\kern0pt}xs{\isacharcomma}{\kern0pt}\ x{\isacharhash}{\kern0pt}ys{\isacharparenright}{\kern0pt}\ {\isasymin}\ adj{\isacharunderscore}{\kern0pt}L{\isadigit{1}}{\isadigit{1}}nat{\isadigit{2}}\ {\isacharcircum}{\kern0pt}{\isacharcircum}{\kern0pt}\ k{\isachardoublequoteclose}
\end{isabelle}

\begin{toappendix}
\section{Measurable Space of Lists}\label{sec:list_space}
We formalize the measurable space of finite lists and the measurability of list operators.

The standard library of Isabelle/HOL contains the formalization of the measurable spaces of streams and trees, where the stream space 
is defined using the countably infinite product space~(see also \cite{10.1145/3018610.3018628}), and the tree space is defined by 
constructing its $\sigma$-algebra over trees directly.
To our knowledge, the measurable space of \emph{finite} lists is not implemented yet.

In this work, we construct the measurable space $X^\ast$ of finite lists on a measurable space $X$.
The construction is similar to the quasi-Borel spaces of finite lists~\cite{hirata22,hirata_et_al:LIPIcs.ITP.2023.18}.
We give the measurable space $\coprod_{k \in \Nat} \prod_{i \in \{0,\ldots,k\}} X$ 
then we introduce the $\sigma$-algebra of $\mathrm{ListM}(X)$ induced by the bijection 
$\mathrm{ListM}(X) \cong \coprod_{k \in \Nat} \prod_{i \in \{0,\ldots,k\}} X$.
To prove the measurability of $\mathrm{Cons} \colon X \times \mathrm{ListM}(X) \to \mathrm{ListM}(X)$ ($(x,xs) \mapsto (x \# xs)$),
we need the countable distributivity, that is, the measurability of bijections of type $X \times \coprod_{k \in \Nat} Y_k \cong \coprod_{k \in \Nat}(X \times Y_k)$.
We omit the details of formalization.

In our Isabelle/HOL library, we provide the constant
\[
\isa{
\ listM\ {\isacharcolon}{\kern0pt}{\isacharcolon}{\kern0pt}\ {\isachardoublequoteopen}{\isacharprime}{\kern0pt}a\ measure\ {\isasymRightarrow}\ {\isacharprime}{\kern0pt}a\ list\ measure{\isachardoublequoteclose}
}
\]
that constructs a measurable space of finite lists.
The underlying set of \isa{listM\ M} is exactly \isa{lists\ {\isacharparenleft}{\kern0pt}space\ M{\isacharparenright}}:
%For a measurable space \isa{M}, the underlying set if \isa{listM\ M} is \isa{lists\ {\isacharparenleft}{\kern0pt}space\ M{\isacharparenright}}.
%\tetsuya{TODO:update}
%
\begin{isabelle}
\isacommand{lemma}\isamarkupfalse%
\ space{\isacharunderscore}{\kern0pt}listM{\isacharcolon}{\kern0pt}\isanewline
\ \ \isakeyword{shows}\ {\isachardoublequoteopen}space\ {\isacharparenleft}{\kern0pt}listM\ M{\isacharparenright}{\kern0pt}\ {\isacharequal}{\kern0pt}\ {\isacharparenleft}{\kern0pt}lists\ {\isacharparenleft}{\kern0pt}space\ M{\isacharparenright}{\kern0pt}{\isacharparenright}{\kern0pt}{\isachardoublequoteclose}
\end{isabelle}
%We have formalized measurability theorems of basic list operations, for example
%\begin{isabelle}
%\isacommand{lemma}\isamarkupfalse%
%\ measurable{\isacharunderscore}{\kern0pt}unit{\isacharunderscore}{\kern0pt}ListM{\isacharbrackleft}{\kern0pt}measurable{\isacharbrackright}{\kern0pt}{\isacharcolon}{\kern0pt}\isanewline
%\ \ \isakeyword{shows}\ {\isachardoublequoteopen}{\isacharparenleft}{\kern0pt}{\isasymlambda}p{\isachardot}{\kern0pt}\ {\isacharbrackleft}{\kern0pt}p{\isacharbrackright}{\kern0pt}{\isacharparenright}{\kern0pt}\ {\isasymin}\ M\ {\isasymrightarrow}\isactrlsub M\ listM\ M{\isachardoublequoteclose}
%\end{isabelle}
We formalize the measurability of Cons in the following uncurried version, because in general, there is no function spaces of measurable spaces in general~\cite{10.1215/ijm/1255631584}.
\begin{isabelle}
\isacommand{lemma}\isamarkupfalse%
\ measurable{\isacharunderscore}{\kern0pt}Cons{\isacharbrackleft}{\kern0pt}measurable{\isacharbrackright}{\kern0pt}{\isacharcolon}{\kern0pt}\isanewline
\ \ \isakeyword{shows}\ {\isachardoublequoteopen}{\isacharparenleft}{\kern0pt}{\isasymlambda}\ {\isacharparenleft}{\kern0pt}x{\isacharcomma}{\kern0pt}xs{\isacharparenright}{\kern0pt}{\isachardot}{\kern0pt}\ x\ {\isacharhash}{\kern0pt}\ xs{\isacharparenright}{\kern0pt}\ {\isasymin}\ M\ {\isasymOtimes}\isactrlsub M\ {\isacharparenleft}{\kern0pt}listM\ M{\isacharparenright}{\kern0pt}\ {\isasymrightarrow}\isactrlsub M\ {\isacharparenleft}{\kern0pt}listM\ M{\isacharparenright}{\kern0pt}{\isachardoublequoteclose}
\end{isabelle}
For the same reason, measurability of other list operations also need to be uncurried. 
In particular, the measurability of \isa{rec{\isacharunderscore}{\kern0pt}list} is a bit complicated.
%Since we need everything uncurried for the measurability, we need several versions of measurability theorems of $\mathtt{rec\_list}$ (we omit them).
\begin{isabelle}
\isacommand{lemma}\isamarkupfalse%
\ measurable{\isacharunderscore}{\kern0pt}rec{\isacharunderscore}{\kern0pt}list{\isacharprime}{\kern0pt}{\isacharprime}{\kern0pt}{\isacharprime}{\kern0pt}{\isacharcolon}{\kern0pt}\isanewline
\isakeyword{assumes}\ {\isachardoublequoteopen}{\isacharparenleft}{\kern0pt}{\isasymlambda}{\isacharparenleft}{\kern0pt}x{\isacharcomma}{\kern0pt}y{\isacharcomma}{\kern0pt}xs{\isacharparenright}{\kern0pt}{\isachardot}{\kern0pt}\ F\ x\ y\ xs{\isacharparenright}{\kern0pt}\ {\isasymin}\ N\ {\isasymOtimes}\isactrlsub M\ M\ {\isasymOtimes}\isactrlsub M\ {\isacharparenleft}{\kern0pt}listM\ M{\isacharparenright}{\kern0pt}\ {\isasymrightarrow}\isactrlsub M\ N{\isachardoublequoteclose}\isanewline
\ \ \isakeyword{and}\ {\isachardoublequoteopen}T\ {\isasymin}\ space\ N{\isachardoublequoteclose}\isanewline
\isakeyword{shows}\ {\isachardoublequoteopen}rec{\isacharunderscore}{\kern0pt}list\ T\ {\isacharparenleft}{\kern0pt}{\isasymlambda}\ y\ xs\ x{\isachardot}{\kern0pt}\ F\ x\ y\ xs{\isacharparenright}{\kern0pt}\ {\isasymin}\ {\isacharparenleft}{\kern0pt}listM\ M{\isacharparenright}{\kern0pt}\ {\isasymrightarrow}\isactrlsub M\ N{\isachardoublequoteclose}
\end{isabelle}
Once we have the measurability of \isa{rec{\isacharunderscore}{\kern0pt}list},
the measurability of other list operators can be proved short. 
For example, we have the following measurability theorems.
\begin{isabelle}
\isacommand{lemma}\isamarkupfalse%
\ measurable{\isacharunderscore}{\kern0pt}append{\isacharbrackleft}{\kern0pt}measurable{\isacharbrackright}{\kern0pt}{\isacharcolon}{\kern0pt}\isanewline
\ \ \isakeyword{shows}\ {\isachardoublequoteopen}{\isacharparenleft}{\kern0pt}{\isasymlambda}\ {\isacharparenleft}{\kern0pt}xs{\isacharcomma}{\kern0pt}ys{\isacharparenright}{\kern0pt}{\isachardot}{\kern0pt}\ xs\ {\isacharat}{\kern0pt}\ ys{\isacharparenright}{\kern0pt}\ {\isasymin}\ {\isacharparenleft}{\kern0pt}listM\ M{\isacharparenright}{\kern0pt}\ {\isasymOtimes}\isactrlsub M\ {\isacharparenleft}{\kern0pt}listM\ M{\isacharparenright}{\kern0pt}\ {\isasymrightarrow}\isactrlsub M\ {\isacharparenleft}{\kern0pt}listM\ M{\isacharparenright}{\kern0pt}{\isachardoublequoteclose}
\end{isabelle}
%\begin{isabelle}
%\isacommand{lemma}\isamarkupfalse%
%\ measurable{\isacharunderscore}{\kern0pt}concat{\isacharbrackleft}{\kern0pt}measurable{\isacharbrackright}{\kern0pt}{\isacharcolon}{\kern0pt}\isanewline
%\ \ \isakeyword{shows}\ {\isachardoublequoteopen}concat\ {\isasymin}\ listM\ {\isacharparenleft}{\kern0pt}listM\ M{\isacharparenright}{\kern0pt}\ {\isasymrightarrow}\isactrlsub M\ {\isacharparenleft}{\kern0pt}listM\ M{\isacharparenright}{\kern0pt}{\isachardoublequoteclose}
%\end{isabelle}
\begin{isabelle}
\isacommand{lemma}\isamarkupfalse%
\ measurable{\isacharunderscore}{\kern0pt}map{\isadigit{2}}{\isacharbrackleft}{\kern0pt}measurable{\isacharbrackright}{\kern0pt}{\isacharcolon}{\kern0pt}\isanewline
\ \ \isakeyword{assumes}\ {\isacharbrackleft}{\kern0pt}measurable{\isacharbrackright}{\kern0pt}{\isacharcolon}{\kern0pt}\ {\isachardoublequoteopen}{\isacharparenleft}{\kern0pt}{\isasymlambda}{\isacharparenleft}{\kern0pt}x{\isacharcomma}{\kern0pt}y{\isacharparenright}{\kern0pt}{\isachardot}{\kern0pt}\ f\ x\ y{\isacharparenright}{\kern0pt}\ {\isasymin}\ M\ {\isasymOtimes}\isactrlsub M\ M{\isacharprime}{\kern0pt}\ {\isasymrightarrow}\isactrlsub M\ N\ {\isachardoublequoteclose}\isanewline
\ \ \isakeyword{shows}\ {\isachardoublequoteopen}{\isacharparenleft}{\kern0pt}{\isasymlambda}{\isacharparenleft}{\kern0pt}xs{\isacharcomma}{\kern0pt}ys{\isacharparenright}{\kern0pt}{\isachardot}{\kern0pt}\ map{\isadigit{2}}\ f\ xs\ ys{\isacharparenright}{\kern0pt}\ {\isasymin}\ {\isacharparenleft}{\kern0pt}listM\ M{\isacharparenright}{\kern0pt}\ {\isasymOtimes}\isactrlsub M\ {\isacharparenleft}{\kern0pt}listM\ M{\isacharprime}{\kern0pt}{\isacharparenright}{\kern0pt}\ {\isasymrightarrow}\isactrlsub M\ {\isacharparenleft}{\kern0pt}listM\ N{\isacharparenright}{\kern0pt}{\isachardoublequoteclose}
\end{isabelle}
We also provide the measurability of Isabelle/HOL's list operations
\isa{case\_list}, \isa{map}, \isa{foldr}, \isa{foldl}, \isa{fold}, \isa{rev}, \isa{length}, \isa{drop}, \isa{take} and \isa{zip}, and the total function version of \isa{nth} (\isa{nth\_total}).

\begin{isabelle}
\isacommand{primrec}\isamarkupfalse%
\ nth{\isacharunderscore}{\kern0pt}total\ \ {\isacharcolon}{\kern0pt}{\isacharcolon}{\kern0pt}\ {\isachardoublequoteopen}{\isacharprime}{\kern0pt}a\ {\isasymRightarrow}\ {\isacharprime}{\kern0pt}a\ list\ {\isasymRightarrow}\ nat\ {\isasymRightarrow}\ {\isacharprime}{\kern0pt}a{\isachardoublequoteclose}\ \isakeyword{where}\isanewline
\ \ {\isachardoublequoteopen}nth{\isacharunderscore}{\kern0pt}total\ d\ {\isacharbrackleft}{\kern0pt}{\isacharbrackright}{\kern0pt}\ n\ {\isacharequal}{\kern0pt}\ d{\isachardoublequoteclose}\ {\isacharbar}{\kern0pt}\isanewline
\ \ {\isachardoublequoteopen}nth{\isacharunderscore}{\kern0pt}total\ d\ {\isacharparenleft}{\kern0pt}x\ {\isacharhash}{\kern0pt}\ xs{\isacharparenright}{\kern0pt}\ n\ {\isacharequal}{\kern0pt}\ {\isacharparenleft}{\kern0pt}case\ n\ of\ {\isadigit{0}}\ {\isasymRightarrow}\ x\ {\isacharbar}{\kern0pt}\ Suc\ k\ {\isasymRightarrow}\ nth{\isacharunderscore}{\kern0pt}total\ d\ xs\ k{\isacharparenright}{\kern0pt}{\isachardoublequoteclose}
\end{isabelle}
\begin{isabelle}
\isacommand{lemma}\isamarkupfalse%
\ cong{\isacharunderscore}{\kern0pt}nth{\isacharunderscore}{\kern0pt}total{\isacharunderscore}{\kern0pt}nth{\isacharcolon}{\kern0pt}\isanewline
\ \ \isakeyword{shows}\ {\isachardoublequoteopen}{\isacharparenleft}{\kern0pt}{\isacharparenleft}{\kern0pt}n\ {\isacharcolon}{\kern0pt}{\isacharcolon}{\kern0pt}\ nat{\isacharparenright}{\kern0pt}\ {\isacharless}{\kern0pt}\ length\ xs\ {\isasymand}\ {\isadigit{0}}\ {\isacharless}{\kern0pt}\ length\ xs{\isacharparenright}{\kern0pt}\ {\isasymLongrightarrow}\ nth{\isacharunderscore}{\kern0pt}total\ d\ xs\ n\ {\isacharequal}{\kern0pt}\ nth\ xs\ n{\isachardoublequoteclose}
\end{isabelle}
\begin{isabelle}
\ cong{\isacharunderscore}{\kern0pt}nth{\isacharunderscore}{\kern0pt}total{\isacharunderscore}{\kern0pt}default{\isacharcolon}{\kern0pt}\isanewline
\ \ \isakeyword{shows}\ {\isachardoublequoteopen}{\isasymnot}{\isacharparenleft}{\kern0pt}{\isacharparenleft}{\kern0pt}n\ {\isacharcolon}{\kern0pt}{\isacharcolon}{\kern0pt}\ nat{\isacharparenright}{\kern0pt}\ {\isacharless}{\kern0pt}\ length\ xs\ {\isasymand}\ {\isadigit{0}}\ {\isacharless}{\kern0pt}\ length\ xs{\isacharparenright}{\kern0pt}\ {\isasymLongrightarrow}\ nth{\isacharunderscore}{\kern0pt}total\ d\ xs\ n\ {\isacharequal}{\kern0pt}\ d{\isachardoublequoteclose}
\end{isabelle}
\begin{isabelle}
\isacommand{lemma}\isamarkupfalse%
\ measurable{\isacharunderscore}{\kern0pt}nth{\isacharunderscore}{\kern0pt}total{\isacharbrackleft}{\kern0pt}measurable{\isacharbrackright}{\kern0pt}{\isacharcolon}{\kern0pt}\isanewline
\ \ \isakeyword{assumes}\ {\isachardoublequoteopen}d\ {\isasymin}\ space\ M{\isachardoublequoteclose}\isanewline
\ \ \isakeyword{shows}\ {\isachardoublequoteopen}{\isacharparenleft}{\kern0pt}{\isasymlambda}\ {\isacharparenleft}{\kern0pt}n{\isacharcomma}{\kern0pt}xs{\isacharparenright}{\kern0pt}{\isachardot}{\kern0pt}\ nth{\isacharunderscore}{\kern0pt}total\ d\ xs\ n{\isacharparenright}{\kern0pt}\ {\isasymin}\ {\isacharparenleft}{\kern0pt}count{\isacharunderscore}{\kern0pt}space\ UNIV{\isacharparenright}{\kern0pt}\ {\isasymOtimes}\isactrlsub M\ listM\ M\ {\isasymrightarrow}\isactrlsub M\ M{\isachardoublequoteclose}
\end{isabelle}
\end{toappendix}

\section{Related Work}%1page
Barthe et al.~proposed the relational program logic apRHL reasoning about differential privacy~\cite{Barthe:2012:PRR:2103656.2103670} with its Coq implementation.
The work attracted the interest of many researchers, and many variants of the logic have been studied~\cite{2016arXiv160105047B,Barthe:2015:HAR:2676726.2677000,Barthe:2016:APC:2976749.2978391}.
These underlying semantic models are based on discrete models of probabilistic programs.
After that, Sato et al.~introduced measure-theoretic models for apRHL~\cite{Sato2016MFPS,DBLP:conf/lics/SatoBGHK19}, and Sato and Katsumata extended the model to support quasi-Borel spaces~\cite{Sato_Katsumata_2023}.

For other formulations of differential privacy,
Mironov and Bun et al. introduced R\'{e}nyi differential privacy(RDP)~\cite{Mironov17} and zero-concentrated differential privacy(zCDP)~\cite{BunS16} respectively.
They give more rigorous evaluations of the differential privacy of programs with noise sampled from Gaussian distributions.
Kariouz et al.~introduced the hypothesis testing interpretation of $(\varepsilon,\delta)$-differential privacy~\cite{KairouzOV15} for tightening the composability of DP.
After that, Balle et al. applied that to giving a tighter conversion from RDP to DP~\cite{DBLP:journals/corr/abs-1905-09982,ALCKS2021},
and Dong et al. give another formulation of differential privacy based on the trade-off curve between Type I and Type II errors in the hypothesis testing for two adjacent datasets~\cite{10.1111/rssb.12454}.

There are several related studies for formal verification of probabilistic programs in proof assistants.
Eberl et al.~constructed an executable first-order functional probabilistic programming language in Isabelle/HOL~\cite{10.1007/978-3-662-46669-8_4}. 
Lochbihler formalized protocols with access to probabilistic oracles in Isabelle/HOL for reasoning about cryptographic protocols~\cite{10.1007/978-3-662-49498-1_20},
and Basin et al.~implemented a framework CryptHOL for rigorous game-based proofs in cryptography in Isabelle/HOL~\cite{crypthol}.
Hirata et al.~developed an extensive Isabelle/HOL library of quasi-Borel spaces~\cite{hirata22,hirata_et_al:LIPIcs.ITP.2023.18}, which is a model for denotational semantics of higher-order probabilistic programming languages with continuous distributions~\cite{HeunenKammarStatonYang2017LICS,Scibior:2017:DVH:3177123.3158148}).
Bagnall and Stewart embedded {\sc MLCert} in Coq for formal verification of machine learning algorithms~\cite{Bagnall_Stewart_2019}.
Affeldt et al.~formalized probabilistic programs with continuous random samplings and conditional distributions~\cite{10.1145/3573105.3575691} by formalizing the semantic model based on s-finite kernels~\cite{10.1007/978-3-662-54434-1_32}.
Tristan et al.~developed an automated measurability prover for probabilistic programs in the continuous setting via reparametrization of the uniform distribution in Lean~\cite{tristan2020verification}.

Although Isabelle/HOL seems the most advanced for the semantic model of probabilistic programs,
the libraries of Coq and Lean are rich enough to formalize DP in the continuous setting.
The Lean {\tt Mathlib} library contains the formalization of basic measure theory (see \cite{LeanDoc}).
Affeldt et al. have made significant progress of the formalization of basic measure theory in Coq.
Recently, they have finished the implementation of the Radon-Nikod\'{y}m theorem and the fundamental theorem of calculus for Lebesgue integration in {\tt MathComp-Analysis}~\cite{ishiguro2023ppl,affeldt2024jfla}.

\section{Conclusion and Future Work}%1page
In this paper, we have proposed an Isabelle/HOL library for formalizing of differential privacy.
Our work enables us to formalize differential privacy in Isabelle/HOL in the continuous setting.
To our knowledge, this is the first formalization of differential privacy supporting continuous probability distributions.
We plan to extend our library to formalize more (advanced) results on differential privacy.
We show several possible future works.
%\tetsuya{TODO: Add reorganization: Hirata's coproduct measure, Laplace Transform etc..}
\begin{itemize}
\item 
We plan to formalize further noise-adding mechanisms for differential privacy. 
For example, the Gaussian mechanism adds the noise sampled from the Gaussian distribution~(see \cite[Section A]{DworkRothTCS-042}).
To formalize it in Isabelle/HOL, we expect that the AFP entry \isa{The Error Function}~\cite{Error_Function-AFP} will be useful.

\item 
Several variants of differential privacy can be formulated by replacing the divergence $\Delta^\varepsilon$ with other ones.
We plan to formalize such variants. In particular, we aim to formalize RDP based on the R\'{e}nyi divergence.
For the technical basis of this, we would need to formalize $f$-divergences and their continuity~\cite{Csiszar67, 1705001_2006}.

\item 
We expect to extend our library to support higher-order functional programs by combining the work \cite{hirata22,hirata_et_al:LIPIcs.ITP.2023.18} of Hirata~{et~al.} for semantic foundations of higher-order probabilistic programs.
%For this, we use the conversion in~\cite{Sato_Katsumata_2023}.
%For this, we convert the divergence $\Delta^\varepsilon$ to the quasi-Borel models as in~\cite[Section 5.5]{Sato_Katsumata_2023}.

\item
Thanks to the AFP entries \isa{A Formal Model of IEEE Floating Point Arithmetic}~\cite{IEEE_Floating_Point-AFP} and \isa{Executable Randomized Algorithms}~\cite{Executable_Randomized_Algorithms-AFP}, 
it is possible to formalize the differential privacy of floating-point mechanisms and to generate executable programs for differential privacy in Isabelle/HOL.

\item
In addition, it is possible to rewrite the proofs shorter using the AFP entries.
For example, it might be possible to apply \isa{Laplace Transform}~\cite{Laplace_Transform-AFP} to shorten the formalization of Laplace mechanism, and 
to apply the very recent entry \isa{Coproduct Measure} \cite{Coproduct_Measure-AFP} to rewrite the formalization of measurable spaces of finite lists.
%For an application of our library, it is possible to implement
%variants of apRHL and linear dependent type systems for differential privacy and give a formal proof of their soundness in Isabelle/HOL.
\end{itemize}

%% The acknowledgments section 
%\begin{acks}
%\end{acks}
\bibliographystyle{ACM-Reference-Format}
\renewcommand{\appendixbibliographystyle}{ACM-Reference-Format}
\bibliography{reference}
%\appendix
\end{document}